\pdfoutput=1
\RequirePackage{silence} 
\documentclass[ oneside,openright,titlepage,numbers=noenddot,
                headinclude,footinclude,cleardoublepage=empty,abstract=on,
                BCOR=5mm,paper=a4,fontsize=12pt,american,usenames,dvipsnames,dottedtoc%
                ]{scrreprt}

\PassOptionsToPackage{utf8}{inputenc}
	\usepackage{inputenc}

\PassOptionsToPackage{eulerchapternumbers,listings,%
					 pdfspacing,
					 subfig,beramono,eulermath,parts}{classicthesis}                                        

\newcommand{\myTitle}{Real-world Graph Analysis: Techniques for Static, Dynamic and Temporal Communities\xspace}

\newcommand{\myName}{Davide Rucci\xspace}

\newcommand{\myFaculty}{Dipartimento di Informatica\xspace}

\newcommand{\myUni}{Università di Pisa\xspace}

\newcommand{\myTime}{May 2024 \xspace}
\newcommand{\myVersion}{}

\newcounter{dummy} 
\providecommand{\mLyX}{L\kern-.1667em\lower.25em\hbox{Y}\kern-.125emX\@}

\renewcommand{\O}{\ensuremath{\mathcal{O}}}
\newcommand{\GVE}{\ensuremath{G = (V, E)}}
\newcommand{\bigO}{\ensuremath{\mathcal{O}}}
\newcommand{\turan}{Tur\'{a}n}
\newcommand{\erdos}{Erd\"{o}s}
\newcommand{\directedG}{\ensuremath{\overrightarrow{G}}}
\newcommand{\enump}{\textsc{EnumP}}

\renewcommand{\thefigure}{\arabic{figure}} 

\newcommand{\changed}[1]{#1}
\newcommand{\unionh}{\ensuremath{\cup_h}}

\PassOptionsToPackage{american}{babel}   
	\usepackage{babel}                  

\usepackage{csquotes}
 \PassOptionsToPackage{%
 	backend=biber,bibencoding=utf-8,%
 	language=auto,%
 	style=numeric-comp,%
     sorting=nyt, 
     maxbibnames=10, 
     natbib=true 
 }{biblatex}
     \usepackage{biblatex}

\PassOptionsToPackage{fleqn}{amsmath}       
    \usepackage{amsmath}
\usepackage{amsthm, amsfonts, amssymb}

\theoremstyle{definition}
\newtheorem{de}{Definition}[chapter]
\theoremstyle{definition}
\newtheorem{definition}{Definition}[chapter]
\theoremstyle{remark}

\theoremstyle{plain}
\newtheorem{prop}{Proposition}[chapter]
\theoremstyle{plain}
\newtheorem{theorem}{Theorem}[chapter]
\theoremstyle{plain}
\newtheorem{lemma}{Lemma}[chapter]
\theoremstyle{plain}

\theoremstyle{definition}
\newtheorem{obs}{Observation}[chapter]

\newtheorem{corollary}{Corollary}[chapter]

\PassOptionsToPackage{T1}{fontenc} 
    \usepackage{fontenc}     
\usepackage{textcomp} 
\usepackage[]{scrhack} 
\usepackage{xspace} 
\usepackage{mparhack} 
\PassOptionsToPackage{printonlyused,smaller}{acronym} 
    \usepackage{acronym} 

\usepackage[sanserif, full]{complexity}
\usepackage{tabularx}
\usepackage{tikz}
\usepackage[ruled, vlined, linesnumbered, lined]{algorithm2e}
\usepackage{subcaption}
\usepackage[shortlabels]{enumitem}
\usepackage{multirow}
\usepackage{makecell}
\usepackage{svg}
\usepackage{longtable}
\usepackage{wrapfig}
\usepackage{todonotes}
\usepackage{changepage}
\usepackage[export]{adjustbox}
\usepackage[final]{microtype}
\usepackage{setspace}

\usepackage{etoolbox}
\apptocmd{\sloppy}{\hbadness 10000\relax}{}{}
    
\usepackage{tabularx} 
\usepackage{caption}
\usepackage{subcaption}  

\usepackage{listings} 
    \usepackage{hyperref}  
    \hypersetup{pdfpagelabels,urlcolor=black}
\PassOptionsToPackage{pdftex}{graphicx}
    \usepackage{graphicx} 
 \usepackage{epstopdf}
\hypersetup{%
    colorlinks=true, linktocpage=true, pdfstartpage=3, pdfstartview=FitV,%
    breaklinks=true, pdfpagemode=UseNone, pageanchor=true, pdfpagemode=UseOutlines,%
    plainpages=false, bookmarksnumbered, bookmarksopen=true, bookmarksopenlevel=1,%
    hypertexnames=true, pdfhighlight=/O,
    urlcolor=webbrown, linkcolor=RoyalBlue, citecolor=webgreen, 
    pdftitle={\myTitle},%
    pdfauthor={\textcopyright\ \myName, \myUni, \myFaculty},%
    pdfsubject={},%
    pdfkeywords={},%
    pdfcreator={pdfLaTeX},%
    pdfproducer={LaTeX with hyperref and classicthesis}%
}   

\makeatletter
\@ifpackageloaded{babel}%
    {%
       \addto\extrasamerican{%
                }%
       \addto\extrasngerman{%
                }%
    }{\relax}
\makeatother

\usepackage{classicthesis} 

\AtEveryBibitem{\clearlist{language}}
\BiblatexSplitbibDefernumbersWarningOff

\begin{document}
\frenchspacing
\raggedbottom
\selectlanguage{american} 
\pagenumbering{roman}
\pagestyle{plain}

\begin{titlepage}












  \thispagestyle{empty}
  \begin{center}

\begingroup
  \includegraphics[width=.8\linewidth]{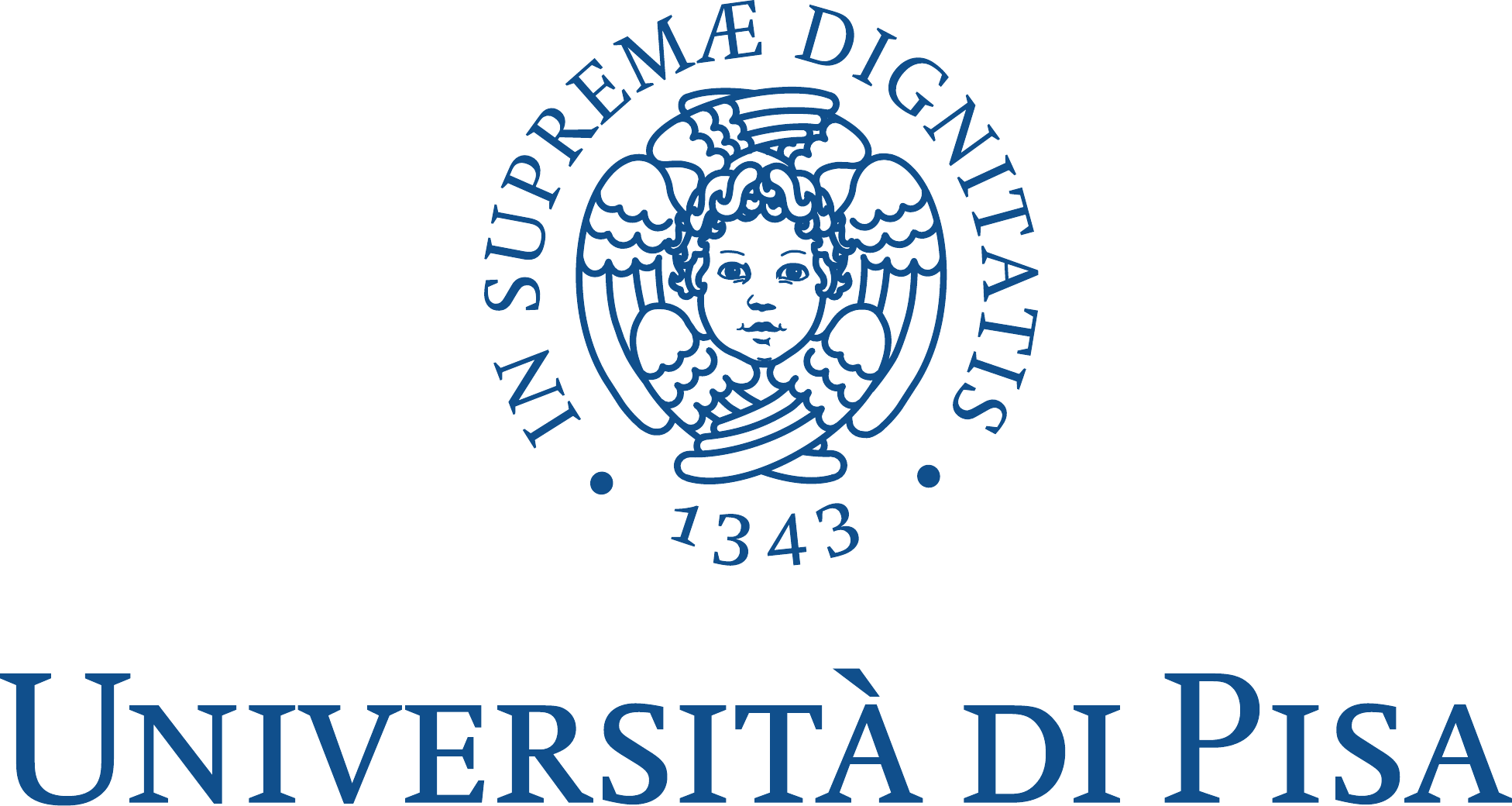}
\endgroup

\vspace{0.5cm}

\Large
\spacedallcaps{{Dipartimento di Informatica}}\\
\vspace{.5cm}
\hrule
\vspace{0.5cm}
\large
\spacedlowsmallcaps{Ph.D. Thesis}\\

\vspace{1cm}
\Large
\linespread{1.2}
\spacedlowsmallcaps{Real-world Graph Analysis: Techniques for Static, Dynamic and Temporal Communities}\\ 
\large
\spacedlowsmallcaps{Davide Rucci}
\linespread{1}

\vspace{1.5cm}

\begin{minipage}[t]{5cm}
\flushright
\begin{center}
\textbf{Supervisor} \\
\spacedlowsmallcaps{Alessio Conte} \\[.5em]
\normalsize{Università di Pisa}
\end{center}
\end{minipage}
\hfill
\begin{minipage}[t]{5cm}
\flushright
\begin{center}
\textbf{Supervisor}\\
\spacedlowsmallcaps{Roberto Grossi}\\[.5em]
\normalsize{Università di Pisa}
\end{center}
\end{minipage}
\vfill
\begin{minipage}[t]{5cm}
\flushright
\begin{center}
\textbf{Referee} \\
\spacedlowsmallcaps{Hiroki Arimura} \\[.5em]
\normalsize{Hokkaido University}
\end{center}
\end{minipage}
\hfill
\begin{minipage}[t]{5cm}
\flushright
\begin{center}
\textbf{Referee}\\
\spacedlowsmallcaps{Lhouari Nourine}\\[.5em]
\normalsize{University Clermont-Auvergne}
\end{center}
\end{minipage}

\vfill
\vspace{0.5cm}
\spacedlowsmallcaps{May 2024}\\
\spacedlowsmallcaps{{Cycle XXXVI}}\\

\end{center}
\end{titlepage}   
\thispagestyle{empty}

\hfill

\vfill

\noindent\myName: \textit{\myTitle},\\ Doctoral Dissertation, 
\\
\textcopyright\ \myTime
\cleardoublepage\cleardoublepage
\cleardoublepage
\pdfbookmark[1]{Abstract}{Abstract}
\begingroup
\let\clearpage\relax
\let\cleardoublepage\relax
\let\cleardoublepage\relax

\chapter*{Abstract}
    {\setstretch{1.0}\color{black}
    Graphs are widely used in various fields of computer science. 
    They have also found application in unrelated areas, leading to a diverse range of problems.
    These problems can be modeled as relationships between entities in various contexts, such as social networks, protein interactions in cells, and route maps.
    Therefore it is logical to analyze these data structures with diverse approaches, whether they are numerical or structural, global or local, approximate or exact.
    In particular, the concept of community plays an important role in local structural analysis, as it is able to highlight the composition of the underlying graph while providing insights into what the organization and importance of the nodes in a network look like.
    
    This thesis pursues the goal of extracting knowledge from different kinds of graphs, including static, dynamic, and temporal graphs, with a particular focus on their community substructures.
    To tackle this task we use combinatorial algorithms that can list all the communities in a graph according to different formalizations, such as cliques, $k$-graphlets, and $k$-cores.
    Listing (or enumeration) algorithms indeed give the full set of communities in any graph, in contrast to more algebraic techniques that exploit matrix multiplication, eigenvalues, etc. to obtain an exact or approximate count of certain types of communities, an approach that can be faster in terms of running time, but that does not suit our need for the complete list of subgraphs to perform subsequent analyses.
    
    We first develop new algorithms to enumerate subgraphs, using traditional and novel techniques such as push-out amortization, and CPU cache analysis to boost their efficiency.
    
    We then extend these concepts to the analysis of real-world graphs across diverse domains, ranging from social networks to autonomous systems modeled as temporal graphs. In this field, there is currently no widely accepted adaptation, even for straightforward subgraphs like $k$-cores, and the available data is expanding both in terms of quantity and scale.

    As a result, our findings advance the state of the art both from a theoretical and a practical perspective and can be used in a static or dynamic setting to further speed up and refine graph analysis techniques.
    }


\endgroup			

\vfill
\cleardoublepage
\pdfbookmark[1]{Acknowledgments}{acknowledgments}



\begingroup
\let\clearpage\relax
\let\cleardoublepage\relax
\let\cleardoublepage\relax
\chapter*{Acknowledgements}

    \noindent First and foremost, I would like to thank my supervisors, \emph{Roberto Grossi} and \emph{Alessio Conte}, for their support and guidance during my PhD experience.
    I want to thank also my internal PhD committee, namely \emph{Anna Monreale} and \emph{Nadia Pisanti}, for their advice and help throughout these years (and for bringing me a physical copy of my first published PhD paper!).
    I wish to thank my parents, \emph{Laura} and \emph{Claudio} for their endless love and support, I could not have been here without you, you are loved.
    I thank all my co-authors and all the people I have met along this journey.
    Thanks to \emph{Giulia}, my fellow PhD partner with whom I shared many thoughts and funny moments. I have worked very well with you.
    Thanks to \emph{Arnbj\"{o}rg Soff\'{i}a} and \emph{Emil}, my Danish colleagues, who spent so much of their time with me, without complaining! 
    I am glad that I met you during my visiting period at DTU, the time we spent together was amazing, our deep conversations and funny moments will always have a special place in my heart. 
    I hope to meet you again sometime in the future.

\endgroup

\pagestyle{scrheadings}
\cleardoublepage
\refstepcounter{dummy}
\renewcommand{\contentsname}{Table of Contents}
\pdfbookmark[1]{\contentsname}{tableofcontents}
\setcounter{tocdepth}{2} 
\setcounter{secnumdepth}{3} 
\manualmark
\markboth{\spacedlowsmallcaps{\contentsname}}{\spacedlowsmallcaps{\contentsname}}
\tableofcontents 
\automark[section]{chapter}
\renewcommand{\chaptermark}[1]{\markboth{\spacedlowsmallcaps{#1}}{\spacedlowsmallcaps{#1}}}
\renewcommand{\sectionmark}[1]{\markright{\thesection\enspace\spacedlowsmallcaps{#1}}}
\clearpage

\begingroup 
    \let\clearpage\relax
    \let\cleardoublepage\relax
    \let\cleardoublepage\relax
    \refstepcounter{dummy}
    \pdfbookmark[1]{\listfigurename}{lof}

    \vspace{8ex}

    \refstepcounter{dummy}
    \pdfbookmark[1]{\listtablename}{lot}
        
    \vspace{8ex}
    
    \refstepcounter{dummy}
    \pdfbookmark[1]{\lstlistlistingname}{lol}

    \vspace{8ex}
       
    \refstepcounter{dummy}
    \pdfbookmark[1]{Acronyms}{acronyms}
\endgroup

\cleardoublepage\pagenumbering{arabic}
\cleardoublepage
  \chapter{Introduction}

Networks are one of the fundamental models in informatics, as they can represent the structure of complex phenomena and their dynamics, and both their size and their availability has rapidly grown in the past two decades.
Graph-theoretical methods have been the main tool to analyze networks for a long time \cite{BARNES1983235}, and they become increasingly important as graphs pop up in basically most areas of data mining and optimization; one of the key tasks here is identifying \emph{communities} and \emph{densely connected} entities, formalized in this thesis as \emph{subgraphs}.

\section{Community Analysis in Graphs}
Subgraphs are, in fact, a well-known example of \emph{patterns} or \emph{motifs} that can be extracted from graphs to discover meaningful insights on both the local and global structure of the network they represent. In recent years there has been growing attention on the graph motif problem~\cite{lacroix2006motif, ciriello2008review_motifs, yu2020motif_survey} which asks to list or count a specific kind of subgraph in a complex network. 
These subgraphs not only offer valuable information on the underlying network's properties, they also serve as building blocks in many fields such as biology and bioinformatics, sociology, computer networks, and transport engineering: indeed, in these kinds of networks, subgraphs are able to capture and highlight the community structure of the network itself \cite{santofortunato, community_detection_2}.
For instance, communities in a social network allow us to find groups of people with similar interests so that they can be exploited for targeted advertising \cite{munir2023integrated}, but also to find outliers that may be performing some kind of unfair activity such as buying followers \cite{faloutsos_corescope}. 
In biological networks, motifs can highlight groups of proteins that may be correlated either by functional similarities or by means of a disease \cite{10.1093/bioinformatics/btg177, Liu2016ELPA, Kong2020_escherichia_coli}.
Moreover, in computer networks, they might help in the fast recover from a faulty router by directing traffic to other routers in the same community. While this may seem odd because computer networks are human designed, routing algorithms will change the routing tables and therefore the shape of the network, creating new links that would not emerge until that point in time. In fact, community detection on computer networks may suggest the development of new routing algorithms \cite{daha2021cdra}.
Other related questions on graphs may attack, for instance, the evolution of a \texttt{P2P} network over time, the study of a virus (or its immunization) propagation, or we may simply want to extract the most significant subsets from a graph in order to characterize it according to some significance metric \cite{stegehuis2016epidemic, WANG2020106118_epidemic}.

To perform such kinds of analysis we must first retrieve the list of all communities in any given network, thus we need algorithmic techniques that can extract such lists from any graph. 
This thesis, therefore, pursues \emph{two goals}: on one hand to design novel algorithms for the combinatorial extraction of graph substructures, i.e., subgraphs of the input graph that respect a specific property, and also the refinement of existing algorithms. 
On the other hand, to use such subgraphs to analyze specific kinds of networks, understanding how are they effective at providing information about a given graph.

First, to obtain the aforementioned lists, we exploit the so-called \emph{combinatorial enumeration} algorithms. They are algorithms able to list the subgraphs by searching through a space of feasible solutions, outputting only those of interest while avoiding duplicates.
These algorithms are usually considered as the last resort for solving a problem whenever we do not know more efficient methods or we are not able to deeply exploit the structure and properties of a problem; however in this case we have no other choice as many kinds of subgraph are present in a graph in an exponential amount, requiring at least exponential time just to be printed out.
This is why, for enumeration algorithms, research has shifted towards an \emph{output sensitive} complexity analysis, a particular kind of analysis where the complexity is given as a function of the output size and not in the input's, as it is usually the case, allowing for more fine-grained running time bounds \cite{capelli2017complexity}.
In this thesis we perform amortized output sensitive analysis, giving time and space bounds for the output of \emph{one} solution in order to design a new algorithm that advances the state of the art, particularly for $k$-\emph{graphlet} enumeration.
Moreover, we develop a practical optimization of an existing $k$-graphlet enumeration algorithm by exploiting the CPU cache at its fullest, showing vast, yet realistic, improvements over several competitors for the same task.

Obviously, combinatorial enumeration is not the only choice possible: indeed there exist numerous \emph{analytical} algorithms that extract information about communities from the adjacency matrix of a graph, exploiting mainly eigenvalue analysis and matrix multiplication \cite{Singh2015, survey_graphlet}.
However, due to their algebraic nature, such algorithms are more prone to \emph{counting} rather than listing, preventing us from performing the detailed analysis on communities that we want to pursue.
Moreover, they often provide \emph{approximated} results in exchange for faster running times \cite{bressan2018motif}.
Nonetheless, we think that these algorithms are interesting and may provide useful insights on networks, and we examine some of these methods as well.

Another option for analyzing a network is to compute its global and local properties, that include the travel distance between nodes, the centrality of nodes, the diameter of the network among others \cite{yan2009applying, landherr2010critical, kourtellis2013identifying, zhang_betweeness_2015}.
Some of these properties may be of easy calculation by means of a simple traversal, while others require more time and space to be computed, and this can become a problem considering the continuous growth in size of nowadays real-world networks. 
The goal here is first to identify the correct and meaningful measures to compute on these graphs to extract the largest amount of information possible, but also to design new measures that can better take into account the domain of the network under consideration. 
In these second category we find, among others, new definitions for centrality measures in a graph based on connected induced subgraphs \cite{subgraph_centrality_measures} and also the use of graphlets in kernels for machine learning models, like Support Vector Machines \cite{graphlet_kernels}.

For all these reasons, we believe that the work on graph pattern/motif extraction and subsequent analysis is well motivated, especially due to the sheer size of the data available nowadays.

\section{Thesis Organization}
This thesis is structured as follows:

Chapter~\ref{chap:background} introduces most of the formal concepts that will be used throughout the dissertation, and recaps the literature that has been studied and used for producing the next chapters\footnote{Some chapters may introduce additional, more specific, notation and provide more related work from the literature.}.
This includes a review of the theory behind enumeration algorithms, as well as the introduction to the most used general strategies for enumeration with well-known examples on simple and temporal graphs.

Chapter~\ref{chap:ksquare} describes a novel algorithm for listing a particular kind of subgraphs: $k$-graphlets.
This algorithm advances the state of the art by dropping every dependency on the size of the graph in its time complexity per graphlet reported.

Chapter~\ref{chap:cage} takes a complementary approach to $k$-graphlet enumeration, and describes a state-of-the-art algorithm for $k$-graphlet enumeration, with great focus on practical performances exploiting the CPU cache in order to speed up the computation on real-world graphs.

Chapter~\ref{chap:temporal} presents a case study on the usefulness of $k$-cores -- another possible formalization of the community concept -- in a temporal setting.
This allows us to discuss how to extend traditional concepts from static graphs to new contexts like temporal networks, where each connection is marked with a timestamp and the network evolves over a particular time span.

Chapter~\ref{chap:denmark} shows how dynamic graphs can be exploited to obtain a $(1+\epsilon)$ approximation of the densest subgraph in a static graph.

Finally, Chapter~\ref{chap:conclusions} summarizes the results obtained and suggests directions for future work.

\section{Publications of the author}
Part of the work presented in this thesis resulted in the following papers, either published or in the peer review process.

\AtNextBibliography{\renewbibmacro*{pageref}{}}
\nocite{Rucci:2023Spire, Rucci:2024SAC, Rucci:2024Algorithmica, Rucci:2024Iwoca}
 \printbibliography[heading=none,keyword=a]
 \paragraph{Other Publications.}
 \printbibliography[heading=none,keyword=b]

  \chapter{Preliminaries and Related Work}\label{chap:background}

In this chapter we lay down the preliminary concepts and notions on graphs and subgraphs that we use in the rest of thesis.
After giving the basic definitions and notation used for graphs, we move to presenting the main, general, techniques used for enumeration on graphs with a brief recap of the complexity classes associated with enumeration problems, in order to better understand the direction towards which we point our algorithms. 
One of the key concepts here is amortization, an algorithm analysis method that will be heavily exploited especially in Chapter~\ref{chap:ksquare}.
Finally, we recap the literature related to our work on both enumeration and community analysis.

\section{Graph Notions}

Formally, a \emph{graph} is usually defined as a pair $G = (V, E)$, where $V$ is the set of \emph{vertices} or \emph{nodes} and $E \subseteq V \times V$ is the set of arcs or \emph{edges} that represent any kind of interaction between two vertices.
An edge $e \in E$ is a pair $\{u, v\}$ where $u$ and $v$ are its \emph{endpoints}, and $u$ and $v$ are said to be adjacent.
An edge of the kind $\{u, u\}$ is called a \emph{self-loop} where both the endpoints coincide. 
When removing a vertex $v \in V$ from $G$, the new graph obtained will be denoted by $G \setminus \{v\}$ or just $G \setminus v$ and it will not contain all edges that were \emph{incident} to $v$ in $G$.
It is also possible to \emph{shrink} an edge $\{u, v\} \in E$ and obtain the new graph $G' = G / \{u, v\}$ where vertices $u$ and $v$ are compressed into a single, new, vertex that is adjacent to all the neighbors of both $u$ and $v$ in $G$.
The neighborhood of a vertex $u$ is the set of all the vertices adjacent to $u$ in $G$, which changes according to the category of graph that we choose among the ones from the following list.
We may use shorthands $n = |V|$ and $m = |E|$ whenever $G$ is clear from the context.

\subsection{Graphs and Subgraphs Formalization}\label{bg:sec:subgraphs}

\paragraph{Undirected Graphs.}
A graph $G = (V, E)$ is said to be \emph{undirected} if the edges in $E$ consist of \emph{unordered} pairs, i.e. the edge $\{u, v\}$ is the same as the edge $\{v, u\}$.
Thus, given a vertex $v \in V$, we define its \emph{neighborhood} as $N(v) = \{u \in v \,\mid\, \{v, u\} \in E\}$, and its \emph{degree} as $d(v) = |N(v)|$.
The edges \emph{incident} to $v$ are defined as $E(v) = \left\{\{u, v\} \,\mid\, \{u, v\} \in E \right\}$.

\paragraph{Directed Graphs.}
In a directed graph $\overrightarrow{G} = (V, E)$, the edges are represented as \emph{ordered pairs}, so that the edge $(u, v)$, directed from vertex $u$ to vertex $v$, is different from the edge $(v, u)$, directed from vertex $v$ to vertex $u$.
The degree is split into the \emph{indegree} and the \emph{outdegree}. 
The former is referred to as $d^-(u)$ and the latter as $d^+(u)$, for a vertex $u \in V$, and they correspond to the cardinality of the set of \emph{incoming} and \emph{outgoing} neighbors, i.e., of $N^-(u) = \{ v \,\mid\, (v, u) \in E\}$ and of $N^+(u) = \{ v \,\mid\, (u, v) \in E\}$ respectively.

By convention directed and undirected graphs with no other particular property are usually called \emph{simple} graphs.

\paragraph{Multigraphs.}
It may happen that a graph $G = (V, E)$ has \emph{duplicated} edges, i.e., multiple edges with the same endpoints, that can be either directed or undirected, thus $E$ is a \emph{multiset} and each edge $e$ has a multiplicity $m(e)$ associated.
In a multigraph, (in/out)degrees are defined taking into account the multiplicities of edges, while we define the neighborhood in the same way we did for the above kinds of graphs.

\paragraph{Temporal Graphs.}
A temporal graph $G_\tau = (V, E_\tau \subseteq V \times V \times [1, \tau])$ is a pair of vertices and \emph{temporal edges} ${u, v, t}$, marked with a timestamp $1 \leq t \leq \tau$, and $\tau$ is called the \emph{lifespan} of $G_\tau$. 
When fixing a particular value $i$ for the time (i.e. the $t$ value of the edges) we retrieve the static graph $G_i = (V, E_i = \{\{u, v\} : \{u, v, i\} \in E_\tau\})$ that is called the $i$-th snapshot of $G_\tau$; thus, $\tau$ actually denotes the total number of snapshots available for $G_\tau$. 
For any given snapshot $G_t$, we denote the neighborhood of node $v$ as $N(v) = \{w : \{v,w\}\in E_t\}$, and its degree $d(v) = |N(v)|$. Of course, a temporal graph may be either undirected or directed.

\paragraph{Dynamic Graphs.} 
Dynamic graphs can be seen as temporal graphs in which both vertices and edges may appear and disappear at any given time.
The main difference is that we usually do not know the sequence of insertions and deletions ahead of time, as it is often the case for temporal graphs where we are provided with a fixed set of snapshots.
In particular, in this context we have to provide methods like \texttt{Insert(G, e)} and \texttt{Delete(G, e)} that have to handle the insertion and deletion of edges, and occasionally vertices, in reasonable running times.

\paragraph{Subgraphs.}
A graph $H = (V', E')$ is called a \emph{subgraph} of $G = (V, E)$ if $V' \subseteq V$ and $E' \subseteq E$; additionally it is a (vertex-) \emph{induced} subgraph (by $V'$) if $E' = E \cap (V' \times V')$, i.e. if $E'$ contains only edges of $G$ whose endpoints are both in $V'$.
A subgraph induced by $V' \subseteq V$ will be denoted by $G[V']$.
Based on the constraints that we pose on a generic subgraph, we obtain different specializations: we summarize the ones that are studied or mentioned in this thesis here, summarized also in Figure~\ref{bg:fig:examples}.
\begin{itemize}\itemsep0em
    %
    \item ($k$-) \emph{Graphlet}: A \emph{connected induced subgraph}, that is a set $S \subseteq V$ such that in $G[S]$ there always exists a path\footnote{Formally, a path is an ordered sequence of vertices $v_1, v_2, \dots, v_k \in V$ such that the edges $\{v_1, v_2\}, \{v_2, v_3\}, \dots, \{v_{k-1}, v_k\} \in E$.} from any vertex of $S$ to any other of $S$ \cite{przulj_modeling_2004, milo_motifs_2002}. When $k$ is specified, $|S| = k$. In practical settings, $k \leq 10$ (see Figure~\ref{bg:subfig:graphlets}). In the particular case of $k = 3$, these subgraphs are called \emph{triangles}, and they have received a lot of attention in recent years as they are easier to enumerate and count, among other advantages.
    \item $k$-\emph{core}: A maximal set of vertices $C \subseteq V$ such that each $v \in C$ has at least $k$ neighbors in $C$ or, equivalently, every vertex in $G[C]$ has degree $\geq k$ (see Figure~\ref{bg:subfig:2-core}) \cite{k_cores}.
    \item ($k$-) \emph{Clique}: A subgraph induced by $K \subseteq V$ such that $G[K]$ is connected and complete, i.e., every vertex in $K$ is neighbor to every other vertex in $K$. A clique is \emph{maximal} whenever $K$ cannot be further extended with a new vertex without breaking the clique definition. Setting a value for  $k$ means requiring that $|K| = k$ (see Figure~\ref{bg:subfig:max-cliques}).
    \item ($k$-) \emph{Edge Subgraph}: Similar to a graphlet, but with $S \subseteq E$ instead of $V$ (see Figure~\ref{bg:subfig:edge-graphlets}). It is to be noted that finding edge subgraphs is equivalent to finding graphlets in $L(G)$, the \emph{line graph}\footnote{The line graph of $G$, $L(G)$, is obtained by creating a vertex in $L(G)$ for each of the edges in $E$, and then adding an edge in $L(G)$ whenever two edges of $G$ share an endpoint.} of $G$.
\end{itemize}

\begin{figure}[!ht]
    \centering
    \begin{subfigure}[b]{\linewidth}
        \centering
        \includegraphics[width=.35\textwidth]{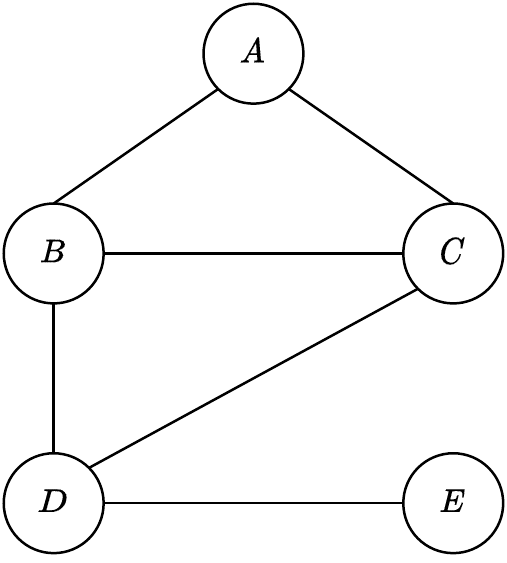}
        \caption{}\label{bg:subfig:graph}
    \end{subfigure}
    \vfill
    \begin{subfigure}[t]{.6\linewidth}
        \centering
        \includegraphics[width=\textwidth]{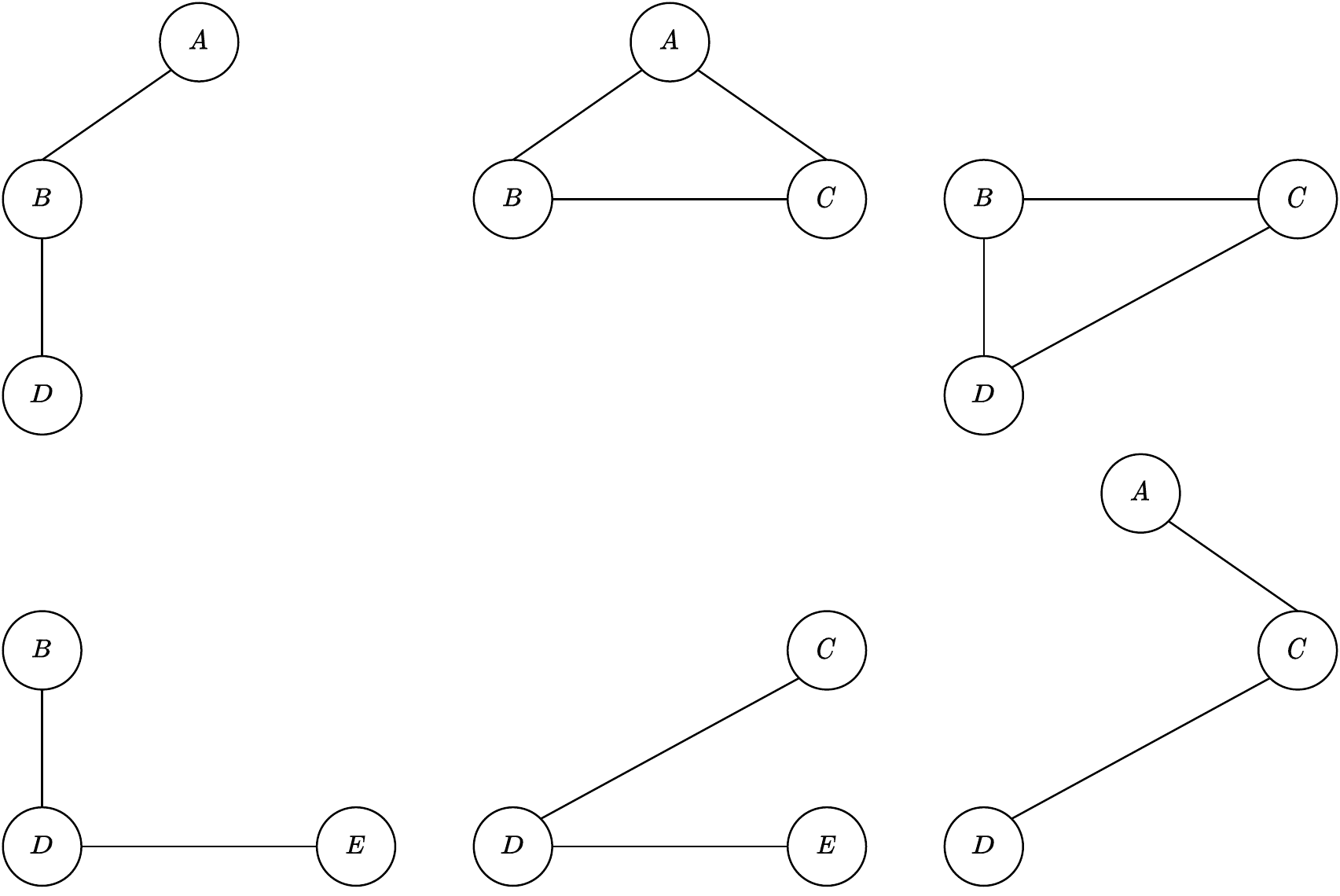}
        \caption{}\label{bg:subfig:graphlets}
    \end{subfigure}
    \begin{subfigure}[t]{.38\linewidth}
        \centering
        \includegraphics[width=.5\textwidth]{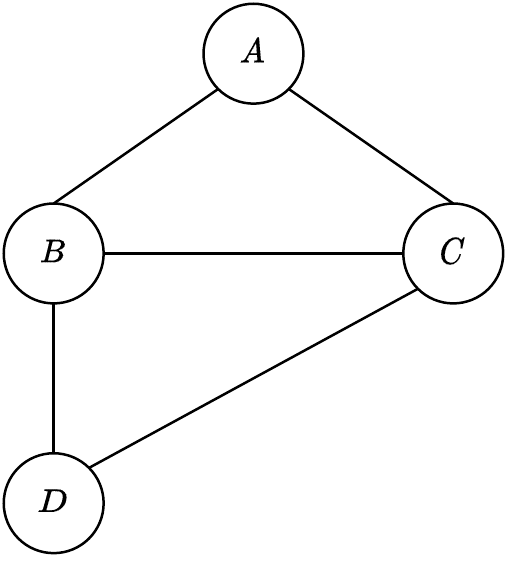}
        \caption{}\label{bg:subfig:2-core}
    \end{subfigure}
    \vfill
    \vfill
    \begin{subfigure}[b]{.4\linewidth}
        \centering
        \includegraphics[width=.95\textwidth]{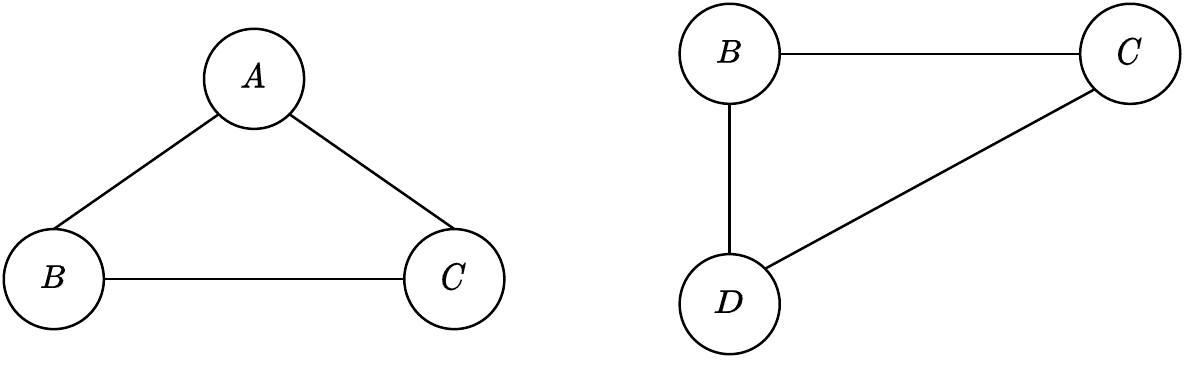}
        \caption{}\label{bg:subfig:max-cliques}
    \end{subfigure}
    \hfill
    \begin{subfigure}[b]{.5\linewidth}
        \centering
        \includegraphics[width=\textwidth]{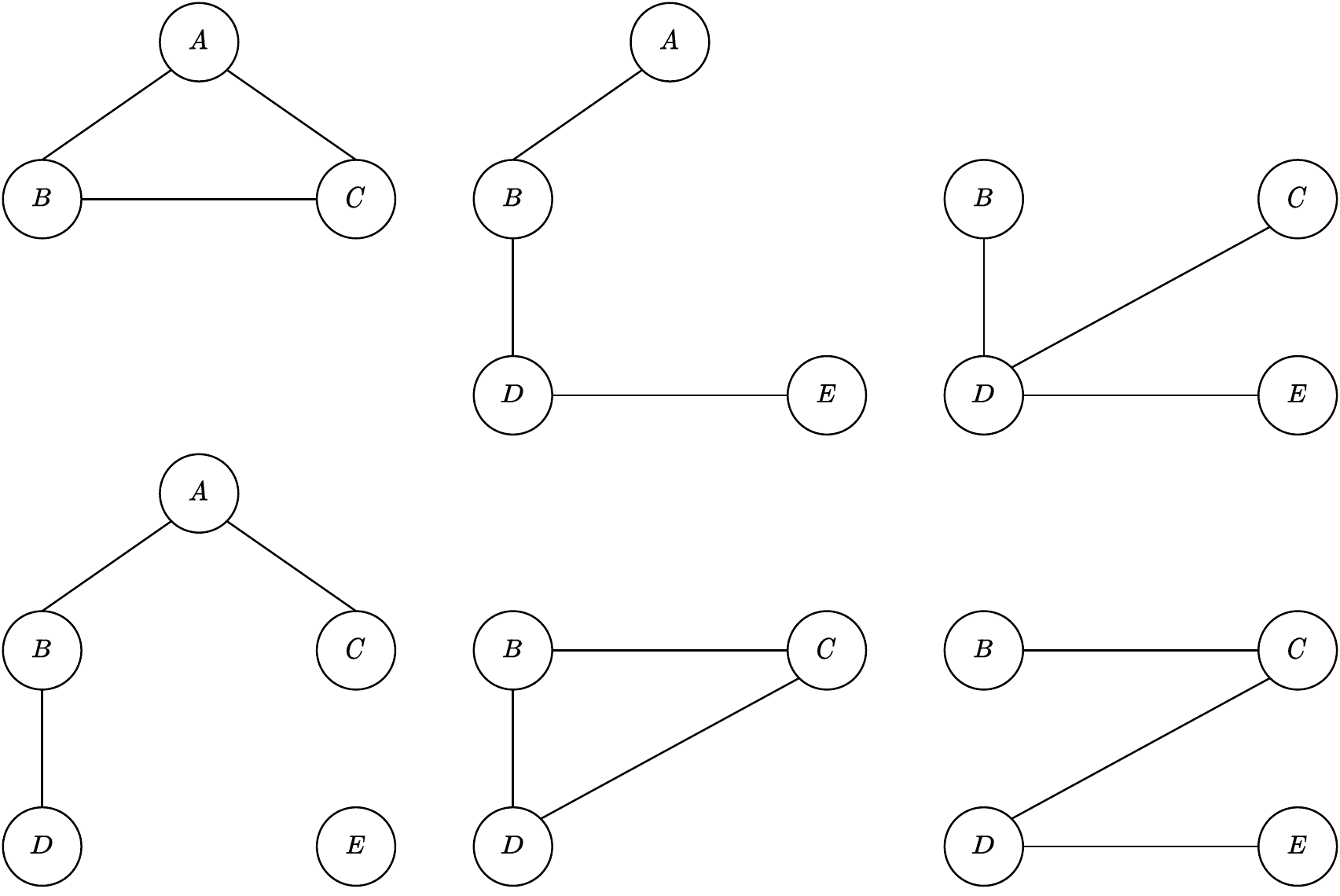}
        \caption{}\label{bg:subfig:edge-graphlets}
    \end{subfigure}
    \caption{(a) a simple graph with 5 vertices and 6 edges; (b) all of its $3$-graphlets; (c) its $2$-core; (d) its maximal cliques; (e) \emph{some} of its 3-edge subgraphs.}
    \label{bg:fig:examples}
\end{figure}

More specialized concepts will be introduced later in their respective Chapters when needed.

As previously noted, these kinds of subgraphs usually appear in an exponential amount even in small graphs (see Figure \ref{fig:brady_graphlets_right} for an example about $k$-subgraphs), with the exception of the $k$-cores which are one for each connected component and require polynomial time to be found.
For this reason, enumeration algorithms \emph{must} take exponential time when there is an exponential amount of solutions to be listed. 
However, in the next section we will refine the process of complexity analysis of these algorithms in order to achieve better and tighter bounds.

\section{Enumeration Theory and Techniques}
When dealing with \NP{}-complete problems, we are left with no choice but to list all the solutions and later answer \emph{yes} or \emph{no} to the original question. 
This process is formally known as \emph{enumeration} and has a rich and deep literature behind. 
We start by introducing two of the most important enumeration techniques, then we will discuss the complexity classes associated with this task and then we will enter the details of the two strategies.

The most used techniques for enumeration nowadays are called \emph{binary partition} and \emph{reverse search}. 
The first is the most traditional one and it usually consists in a recursive approach, enumerating first all the solutions that contain a given element (a node in our case), then removing that element from the input and proceed. 
The latter approach is more recent and it consist in materializing the search space and then traverse it according to a rule that has to be specialized every time to the problem considered. 
The search space is made up by the solutions only, but the paths from one solution to another may be more than one, so one must be careful on the way in which the space is traversed to avoid duplicate solutions and the consequent waste of time.

\subsection{Complexity classes, output sensitive analysis and amortization}\label{sub:complexity_classes}
The work on enumeration has spawned, during the years, an entire theory that explains why two algorithms can achieve significantly different performances despite the fact that both have (maybe the same) exponential time complexity.
In this theory, many complexity classes have been designed with their own hierarchy. We provide here the main results, which are in turn taken from \cite{capelli2017complexity}. We start by formalizing the concept of enumeration problem, using the mathematical definition of binary relation.

\begin{de}[Enumeration Problem]
    Given a binary relation $\mathcal R$ on graphs, and an input graph $G$, list all $H$ such that $(G, H) \in \mathcal R$.
\end{de}

In our case, $H$ is a subgraph of $G$ that satisfies the property required by $\mathcal R$.
Given an enumeration problem $\pi$ and an input instance $x$, we call $\mathcal{A}(x)$ the set of all solutions. In the following we assume that the length of $y \in \mathcal{A}(x)$ is polynomial in $|x|$ and that the computational cost model for algorithms is the RAM \cite{COOK_RAM_Model}, thus we will denote with $T(x,i)$ the time (steps) taken from a Turing machine on input $x$ before the output of the $i$-th solution.
\begin{de}[EnumP]
A problem $\pi$ is in the class \enump{} if the problem of checking $y \in \mathcal{A}(x)$ is in \P.
\end{de}
The \textsc{EnumP} class, informally, consists of the tasks of listing all the solutions of \NP{}  problems. 

\paragraph{Output Sensitive Analysis.} Since we are dealing with an enumeration problem, we have to consider the total amount of time spent computing all the solutions. As this number can intrinsically be exponential (see Figure \ref{fig:brady_graphlets_right}), we may as well give a bound on the computational complexity of an algorithm by using its output size, rather than its input size.
This kind of analysis is more convenient and also more fine-grained with respect to the classical analysis because it will allow us to know, for instance, how much time will it take to get just one solution. 
Ideally, we would like an enumeration algorithm to take \emph{output polynomial time}:
\begin{de}[Output Polynomial Time]
A problem in \enump{} is in \textsc{OutputP} if there exists a polynomial $p(x, y)$ and a machine $M$ that solves it such that for all $x$, $T(x, |\mathcal{A}(x)|) < p(|x|, |\mathcal{A}(x)|)$.
\end{de}
This definition tells us that a problem takes output polynomial time (also \emph{total polynomial time}) if the time taken to list all the elements in $\mathcal{A}(x)$ is bounded by a polynomial in $|\mathcal{A}(x)|$. 

It turns out that these two classes can be seen as analogue to, respectively, \NP{} and \P{} for the enumeration task. In fact, it can be proved that
\begin{prop}[\cite{capelli2017complexity}]
\textsc{OutputP} = \enump{} if and only if \P{} = \NP.
\end{prop}

As we have seen that there exists a correlation between \textsc{OutputP} and \P, one may be tempted to say that \textsc{OutputP} contains all the tractable problems from the enumeration point of view.
Of course this is not the case since the number of solutions may be exponential, but this gives us the motivation to go a step further in the analysis and put some requirements on the time that we have to wait before we are able to get the \emph{next} solution.
\begin{de}[Incremental Polynomial Time]
A problem in \enump{} is in \textsc{IncP}$_a$ if there exists a Turing machine $M$ able to solve it such that for all $x$, for all $0 < t \leq |\mathcal{A}(x)|$ and for constants $a,b,c \in \mathbb{N}$ it holds that $|T(x, t) - T(x, t-1)| < ct^a|x|^b$.
\end{de}
With this definition we are being more specific on what we require to be polynomial: not only do we want the total output time to be polynomial, but also the time for a solution to be output must be polynomial in the size of the already listed solutions.
Using the definition of incremental polynomial time we capture the fact that giving more time to an algorithm \emph{guarantees} more solutions to be output, because we have a polynomial bound on the so-called \emph{delay}.

\begin{de}[Delay of an Enumeration Algorithm]
The \emph{delay} of an enumeration algorithm is the time that interleaves between two consecutively output solutions.
\end{de}

The goal that we wish to achieve is, at least, \emph{polynomial delay} that is the class corresponding to \textsc{IncP}$_0$:
\begin{de}[Polynomial Delay]
A problem in \enump{} is in \textsc{DelayP} (i.e. has polynomial delay) if there exists a Turing machine $M$ which solves it, such that for all $x$ and for all $0 < t \leq |\mathcal{A}(x)|$ it holds that $|T(x, t) - T(x, t-1)| < C|x|^a$ for $C, a \in \mathbb{N}$ constants.
\end{de}
An algorithm that has polynomial delay gives us the strong guarantee that only a polynomial (in the size of the input) amount of time will have to be waited for the next solution to be output. 
Many common problems, like the listing of cliques or graphlets, admit a polynomial delay solution; in fact, often the simplest algorithms are able to achieve this complexity thanks to the following result.

The problem \textsc{ExtSol(A)} involves determining, for a given pair of inputs $x$ and $w_1$, whether there exists another input $w_2$ such that the concatenated string $s = w_1w_2$ belongs to the set $\mathcal{A}(x)$.
\begin{prop}[Time complexity of the Extension Problem\cite{DBLP:conf/stacs/MaryS16}]
 If \textsc{ExtSol(A)} $\in$ \P{}, then the initial enumeration problem $\pi \in $ \textsc{DelayP}.
\end{prop}

\textsc{ExtSol(A)} is at the core of many binary partition algorithms, where we partition, for instance, the set of solutions into those that begin with a 0 and those beginning with a 1 recursively. Then, using \textsc{ExtSol}, we decide whether those subsets of solutions are empty or if we should continue recurring. 

Another useful technique to produce a polynomial delay algorithm is the following. 
Given a graph $G = (V, E)$ and a property $\mathcal{P}$, we want to list all the induced subgraphs of $G$ that satisfy $\mathcal{P}$ (e.g. maximal cliques).
Now define the \emph{input restricted problem} to be the same problem with the additional constraint that $\exists v \in V $ such that $G \setminus \{v\} \in \mathcal{P}$.
It is known that if the input restricted problem can be solved in polynomial time, i.e. it has only a polynomial number of solutions, then the general problem is solvable in polynomial delay.
While the \textsc{ExtSol} strategy generally fits in binary partition algorithm, this latter technique adapts well to reverse search.

\paragraph{Amortized Analysis and Constant Delay.} 
The last class that we present here is one of the most restrictive, hence the most desirable: it is the class of enumeration problems that admit a \emph{constant delay} algorithm. 
However, the concept of constant delay is very sensitive to details like the computational model adopted, the cost model and the complexity measures defined. There exist at least four different flavors of this concept \cite{capelli2017complexity}, and the one that we will make use of is the constant average delay or \emph{Constant Amortized Time}.
\begin{de}[Constant Amortized Time]
An algorithm has a constant amortized time (CAT) complexity if its total running time divided by the number of solutions is constant.
\end{de}
Take, for instance, the listing of all integers between $0$ and $2^n-1$ in ascending order. Adding one to an integer number may change up to $n$ bits, but on average only 2 bits change and the time needed is proportional to the number of changed bits. Thus, over the whole enumeration, only a constant number of operations per solution is required. \cite{Ruskey03combinatorialgeneration} contains many examples of CAT algorithms.
The idea, more generally, is to look at the total time and divide it across all solutions found to get the amortized time \emph{per solution}.
We will make heavy use of amortized analysis later in the thesis.

\subsection{Binary Partition}
Binary Partition is a common and intuitive technique that can be used in enumeration.
As pointed out earlier, it is a recursive approach that works around every single element of the input: first enumerate all the solutions that contain this element, then remove (or ignore) it for the rest of the computation (hence the name \emph{binary}). 
Despite the ease of use of this strategy, many algorithms achieve very good time bounds or, if not, they perform really well in practice thanks to the amortized analysis of the recursion tree generated by a specific algorithm.

The amortized time per solution required by a binary partition algorithm depends on how the corresponding recursion tree behaves, since the listed solutions are contained in leaf nodes.
If the recursion tree contains leaves that do not produce any output (i.e. when the algorithm notices that no further solutions can be discovered along a path) we have an algorithm with exponential or possibly unbounded delay.
If, instead, the tree always has solutions in the leaves we are dealing with a polynomial amortized time algorithm, with two possible alternatives. 
Let us denote the running time of a recursive call by $p(n)$ and suppose that the tree contains unary nodes, that is a path of nodes of length at most $h$ (the height of the tree). In this case the amortized time per solution of the algorithm will be $O(h \cdot p(n))$. On the other hand, if no unary node exists and at the same time every leaf is a solution, the amortized time will be bounded by $O(p(n))$, because each node of binary tree will always produce two children (or be a leaf) and the number of internal nodes is bounded by the number of leaves, hence the total time will be $O(N \cdot p(n))$, where $N$ is the number of solutions produced.
Our goal then would be to keep $p(n)$ as small as possible, ideally $p(n) = O(1)$.
This is quite a difficult task because it is hard to use only constant time during a recursive call, as we have to at least perform some kind of preparation for the next child in the tree. 


\subsection{Push-Out Amortization}
It is a technique designed by Uno \cite{Uno_pushout} for analyzing amortized time enumeration algorithms.
In traditional binary partition algorithms we can observe how the most expensive nodes are those near the root of the recursion tree, usually because the size of the input in early stages is large, so there is more work that needs to be done even for just reading it. 
When we approach the lower nodes, those near the leaves, the problem becomes easier at the point that a leaf may take a constant amount of time.
The key idea to push-out amortization is that we should reverse this concept and make the children nodes cost, together, more than their parent node. If this is the case the parent can \emph{offload} part of its complexity to the children recursively until the leaves which should, instead, keep their very low time bound. 
If an algorithm allows for a push-out amortized analysis, then its cost per solution will be bounded by the time spent into a leaf computation, as stated by Theorem \ref{th:pushout}.

In the following we denote by $X$ a generic recursive call of an enumeration algorithm, that is an internal node of the recursion tree. We denote the set of children spawned by $X$ with $C(X)$, and the sum of the computational time spent by all the children of a node is denoted by $\Bar{T}(X) = \sum\limits_{Y \in C(X)}^{} T(Y)$. Finally, $T^*$ denotes the running time of a leaf iteration.
\begin{theorem}[Push-Out Amortization \cite{Uno_pushout}]\label{th:pushout}
If every internal node of the recursion tree satisfies the following \textbf{push-out condition}:
\[
\Bar{T}(X) \geq \alpha T(X) - \beta |C(X) + 1| T^*
\]
for some constants $\alpha > 1, \beta \geq 0$, $\alpha, \beta \in \mathbb{R}$, then the amortized running time of every recursive call is $T(X) = O(T^*)$.
\end{theorem}

The main intuition behind this technique is that, in enumeration, the number of leaf nodes is fixed and so is the total running time of the bottom level. For this reason, using this technique we can amortize each time the cost of a node on its children until we arrive at the leaves which will always produce solutions. 
This strategy is actually a theoretical rearrangement of the computational cost among all recursive calls in a way such that a piece of running time is delegated to the children and the rest is kept by the parent. 
More details on this can be found in the proof of Theorem \ref{th:pushout} in \cite{Uno_pushout}, as well as some applicative examples.

\subsection{Reverse Search}
Reverse search is a more recent technique compared to binary partition, proposed by Avis and Fukuda in 1996 \cite{avis_reverse_1996}. 
It can be seen as the opposite of binary partition: in fact it is more of a bottom-up technique in contrast with binary partition which is a traditional top-down approach.
Informally, the idea of reverse search is to start already from a solution and then travel the search space jumping between two ``connected'' solutions, without ever creating a non-output instance. Careful attention, therefore, is required when defining the rule according to which we are going to move in the search space, as its computational complexity will influence the overall complexity of the reverse search algorithm.
We can then visualize the search space as a graph containing all the solutions, with edges between two vertices representing the (possibly multiple) ways to construct a new solution from an existing one. 
Suppose that we are able to travel an edge in $O(p(n))$ time, where $n$ represents the size of the input. 
If we are able to visit each solution vertex only once, it is clear that the total running time of the reverse search algorithm will be $O(|\mathcal A(x)| \cdot p(n))$, providing us with an output-sensitive algorithm by construction.
The main challenge here is therefore the design of the so-called \emph{parent rule}, which has to be very efficient and also not visit the same solution twice.
Summarizing, the advantages of a good implementation of reverse search with respect to binary partition are the following: $(a)$ guaranteed output-sensitivity, $(b)$ polynomial space complexity, and $(c)$ prone to parallelization.
While these are definitely appealing advantages, reverse search algorithm are more difficult to analyze and cannot be applied blindly to any problem. In fact, it is possible that the problem corresponding to the parent rule is \NP{}-complete itself.
On the other hand, binary partition is easier to describe, understand, and often implement.
Furthermore they can be very fast in practice despite a worse asymptotical complexity with respect to reverse search.
The original authors of reverse search already provided algorithms for six enumeration problems, among which we also find the enumeration of all graphlets in a graph regardless of their size. It is worth noting that there exists a push-out amortized algorithm for the same latter task, achieving $O(1)$ time per solution \cite{Uno_pushout}.

For the aforementioned reasons, this thesis puts the focus on the binary partition method instead of reverse search.

\section{Related Work}
In this section we give a brief overview of some of the most well-known algorithms for graphlet enumeration and counting, consisting in the current state of the art, as well as some analytical methods. We highlight the advantages and drawbacks that each category has, and we discuss the reasing behind our choices for this thesis.

\subsection{Analytical Methods}
When we do not need to use the solutions, we may just want to count them. For example one can count the number of cliques in graph with a given number of vertices, the number of triangles, the number of 5-graphlets and so on. 
Although not always the case, it is possible to count the number of solution without having to enumerate them all first. 
The adjacency matrix of a graph can be exploited to extract very useful information about patterns inside a network. 
In fact, a large number of algorithms is based on matrix multiplication, for instance, but in general any kind of matrix operation has been applied to extract knowledge about the structure of a given graph.

The counting can be \emph{exact} or \emph{approximated}, based on the needs, and here we give some examples from both categories. Clearly, requesting an exact count for a general pattern is not possible as the associtated problem is still $\NP$-hard, but in some specific cases this can be done efficiently enough. Otherwise we resort to an estimate of the counting that, of course, we wish to be as close as possible to the true number.

\paragraph{Analytical Methods at a Glance.}
Apart from matrix multiplications, analytical methods use many other features of the adjacency matrix $A$, in particular its eigendecomposition.
For instance, it is known that the number of triangles in a graph can be computed using the eigenvalues: in fact we have that $|triangles| = 1/6 \sum_i(\lambda_i^3)$, where $\lambda_i$ are the eigenvalues of the adjacency matrix. This result is due to Tsourakakis in 2008 \cite{tsourakakis_triangles_2008}, which was later extended by Faloutsos et al. \cite{faloutsos_spectral_triangles_2011} with the following Theorem.
\begin{theorem}[Theorem 1 of \cite{faloutsos_spectral_triangles_2011}]
Given a graph $G = (V, E)$, the number of triangles that the edge $(i, j) \in E$ participates in is:
\[
t_{(i, j)} = \sqrt{d_id_j}\left(1-\frac{\left\Vert\frac{A^{(i)}}{\sqrt{d_i}} - \frac{A^{(j)}}{\sqrt{d_j}} \right\Vert^2}{2} \right)
\]
where $A^{(k)}$ is the $k$-th column of the adjacency matrix $A$ and $d_k$ is the degree of node $k$.
\end{theorem}
While this is an exact computation, it still requires much time given the size of graphs nowadays. For this reason, the authors propose to \emph{sketch} all the columns of $A$ (with size $O(\log\frac{|V|}{\delta})$) and then for each edge estimate the number of triangles $t_{ij}$ using the inner product of the sketches. 

This kind of spectral analysis was then adapted to run on large sparse graphs using the \textsc{MapReduce} environment provided by \textsc{Hadoop}, and it was able to handle graphs with up to 6 billion edges and give an accurate estimate of the number of triangles in it \cite{faloutsos_heigen_2011}. This algorithm, \textsc{HEigen}, is based on the low-rank approximation of the adjacency matrix in order to extract its top $k$ eigenvalues.

Triangle counting through eigenvalues has also been extended to dynamic graphs, in particular when they are seen through a streaming lens, i.e. when the edges of the graph come and go at each timestep. The algorithm \textsc{ThinkD} (2019), due Shin, Kim, Hooi and Faloutsos \cite{faloutsos_thinkd_2019}, is able to correctly estimate the number of triangles in this environment through a more accurate sampling of the edges. Previous algorithms in fact were not able to handle edge deletions, or suffered from low accuracy in the results.

It is worth mentioning that applying similar techniques to dynamic or temporal graphs may require to deal with tensors, see for instance the survey in \cite{faloutsos_tensor_survey}.

These methods, together in the broader family of community detection, allow in particular for anomaly and fraud detection in graphs \cite{faloutsos_corescope}.
In fact, since real-world graphs are not random\footnote{i.e. they cannot be represented with a Erdos-Renyi model \cite{erdos_renyi}.}, they contain patterns that can highlight semantic differences in the regions of the graphs. 
Consider, for example, the Ebay system of reviews that can be exploited by creating fake positive reviews in order to boost the sale performances of (malicious) sellers.
This kind of behavior can be identified with the techniques specified above, since malicious sellers and reviewers will form many more triangles among them with respect to the rest of the graph, where the good users are placed \cite{ebay_negotiations_analysis, faloutsos_ebay_2008}.
Similar analysis can be performed on the Twitter social network, where we can investigate why people follow each other, i.e., find the motivation for which the users followed each other in the first place \cite{faloutsos_twitter_2015};
other examples include the discovery of influential users in the network, since they will participate in many more triangles than ``common'' users. 
This, of course, can also be maliciously exploited using the \emph{retweet} action within a circle of people in order to boost their visibility. 
Activities like this, however, can be discovered again by triangles, cluster, homogeneity analysis \cite{faloutsos_twitter_deception_2015}, and also with $k$-cores \cite{faloutsos_corescope}.

One important thing common to all these analytical methods, is the fact that, as we said earlier, they are not able to instantly identify where and which are the communities considered fraudulent or anomalous, and this is because they do not list them. 
In fact, analytical methods can detect and signal the anomaly, but then other operations are required in order to understand the reasons for these events, and often a manual inspection of the data is needed, trading with the speed of the algorithms used.
This is in contrast with enumeration, which by nature already provides the elements of the communities highlighted and can therefore lift a heavy amount of subsequent work, although in exchange with a more time-consuming algorithm.
Finally we note that due to the sheer size of nowadays graphs, it is often not possible to work with the matrix representation of graphs as we prefer the adjacency list representation that is also more prone to compression \cite{webgraph}.

For all the above reasons, we choose to put our focus more on combinatorial enumeration algorithms.


\subsubsection{Graphlet Counting}
There are many algorithms that are capable of giving an exact count of the desired subgraph. 
Their main drawback, however, lies in the limitations they often pose on the size of the graphlets that can be found, either because counting larger graphlet will make the algorithm really slow or because the arguments used to count will suffer from combinatorial explosion and it will not be possible to prove their correctness. 

One algorithm belonging to the latter category is PGD by Ahmed et al. \cite{PGD}.
This method is a completely analytical extension the classic triangle counting algorithm except that during its execution it computes several primitives that are then used to not only count 3- and 4-graphlets, but also the frequency of each orbit\footnote{Orbits are automorphism classes of graphlets, and they allow to visualize the \emph{shape} of a graphlet: for instance, a 5-graphlet may resemble a kite or a line or a pentagon, etc.} for these classes of subgraphs.
All of this comes at the cost of being able to work only with 3- and 4- graphlets. Nonetheless it is one of the fastest counting methods and it is also highly parallelizable.
Another thing to point out is that this algorithm does not use the adjacency matrix, nor it relies on any linear algebra technique, it is constructed by pure combinatorial arguments.

\textsc{Orca} \cite{orca} is instead one of the algorithms based on matrix multiplications, and it is able to count graphlets of up to 5 vertices along with their orbit. It works by setting up a system of linear equations per vertex of the input graph, with the variables being the frequencies of each orbit. The constructed matrix will have rank equal to the number of orbits minus 1, so in order to solve it we have to find just the value of one orbit's frequency and then use any standard linear algebra method to solve the whole system.
Usually, the choice of which orbit to count first corresponds either to the clique or the one with least frequency if the graph is sparse enough.
The \textsc{Jesse} \cite{jesse_1, jesse_2} algorithm builds on \textsc{Orca} and tries to extend it by selecting the least expensive equations.

For approximated algorithms, we mention that they are split in several categories, from randomized methods to path sampling or random walks. 
Path sampling algorithms, for instance, try to count all the orbit frequencies by testing how many times a path (which is in turn a graphlet) made by the same number of nodes ``fits'' inside other type of (non-induced) graphlets.
As a final note, not all the approximated counting algorithms are able to provide frequencies for the orbits and, if they do, they often put a limit on the size of the graphlets.

A comprehensive survey of algorithms for graphlet counting, exact and approximated, can be found in \cite{survey_graphlet}.

\subsection{Graphlet Enumeration}\label{sec:ks_simple_baseline}

Graphlets are one of the simplest kind of subgraph, requiring only the connectivity among a set of vertices and thus greatly increasing the number of solutions to be found through enumeration.
Despite this, many algorithms have been proposed over the years since their introduction in 2004 \cite{przulj_modeling_2004}, and they are now able to handle large graphs quite efficiently.

The classical approach to list all $k$-graphlets in a graph $G=(V,E)$ is based on binary partition. Intuitively, for a graph $G = (V,E)$ and a node $v \in V$, we initialize $S = \{v\}$, and enumerate the $k$-graphlets containing $v$ as the smallest node (so to avoid duplication) by considering a node $u \in N(S)$ and two cases: the $k$-graphlets containing $u$, and those that do not. In the first case, enumeration continues with $S := S \cup \{u\}$, and $k$ reduced to $k-1$; in the second, with the same $S$ and the $k$-graphlets in $G \setminus \{u\}$.
In particular, after setting an arbitrary scanning order for the nodes $v \in V$, each graphlet is built by recursively adding a member of $N(S)$ to $S$, after the initial call with $S = \{v\}$. We get to a recursive leaf when we reach one of the following:
\begin{itemize}\itemsep0em
    \item (\emph{success leaf}) when $|S| = k$;
    \item (\emph{failure leaf}) when $|S| < k$ and $|N(S)| = 0$.
\end{itemize}

A detailed analysis on failure leaves is available in Chapter~\ref{chap:cage}.

Historically, one of the first algorithm able to enumerate all $k$-graphlets is ESU \cite{wernicke_2006}, while the current state-of-the-art algorithm is that of Komusiewicz and Sommer \cite{SIMPLE_graphlets}, hereafter called KS-Simple and summarized in Algorithm~\ref{alg:baseline}, that follows the above binary partition scheme, but with a clever optimization provided below in Proposition~\ref{prop:KS}.

\begin{algorithm}[htb]
\small
\DontPrintSemicolon
  \KwIn{A graph $G=(V,E)$, an integer $k$}
  \KwOut{All $k$-graphlets of $G$}
  \SetKwProg{myproc}{Function}{}{}
  \SetKwFunction{enum}{ENUM}
  \SetKw{true}{true}
  \SetKw{false}{false}
  
  $X \gets \emptyset$\; 
  \ForAll{$v \in V$}{\label{baseline:a}
    \tcp{graphlets in $G \setminus X$ from $v$}
    \enum{$G, \{v\}, k, X$} \tcp*[h]{list graphlets containing $v$ as first node}\;\label{line:enum}
    $X\gets X\cup \{v\}$ \tcp*[h]{remove $v$}\;
    \label{baseline:b}
  }
  \myproc{\enum{$G$,$S$,$k$,$X$}}{\label{baseline:c}
  \tcp{find all graphlets in $G \setminus X$ containing all nodes in $S$}
    \lIf{$|S| = k$}{
      \textbf{yield} $S$;  \Return \true
    } 
    
    \textsl{found} $\gets$ \false\;
    
    \ForAll{$u \in N(S) \setminus X$}{ \label{baseline:d}
      \uIf{\enum{$G, S \cup \{u\}, k, X$} \label{baseline:d1}}{\textsl{found} $\gets$ \true \tcp*[h]{include $u$}}
      \lElse{\textbf{break}\label{baseline:d2}}

      $X\gets X\cup \{u\}$ \tcp*[h]{exclude $u$}\; \label{baseline:e}
    }
    
    \Return \textsl{found}
  
  } 
  
  \caption{KS-Simple: $k$-graphlet enumeration in $O(k^2\Delta)$ time per graphlet.}
  \label{alg:baseline}
\end{algorithm}

\begin{prop}[\cite{SIMPLE_graphlets}]
\label{prop:KS}
    If a $k$-graphlet containing $S$ does not exist in $G$ after taking $u \in N(S)$ then no $k$-graphlet can exist in $G \setminus \{u\}$ with the same $S$.  
\end{prop}
As an example, consider the graph in Figure \ref{fig:brady_graphlets_left} while enumerating all 4-graphlets containing vertex $b$ in $G \setminus \{a\}$. After the output of $\{b, c, d, e\}$, vertex $e$ is removed and no more 4-graphlets (in particular containing $b$) exist. Thanks to Proposition \ref{prop:KS} we can stop the computation as soon as the next recursive call, which will obviously fail. This property allows KS-Simple to have a running time of $O(k^2\Delta)$ per graphlet, where $\Delta$ is the maximum degree of the graph.

\begin{figure}[ht]
    \centering
    \begin{subfigure}[t]{.45\linewidth}
    \includegraphics[width=.78\textwidth]{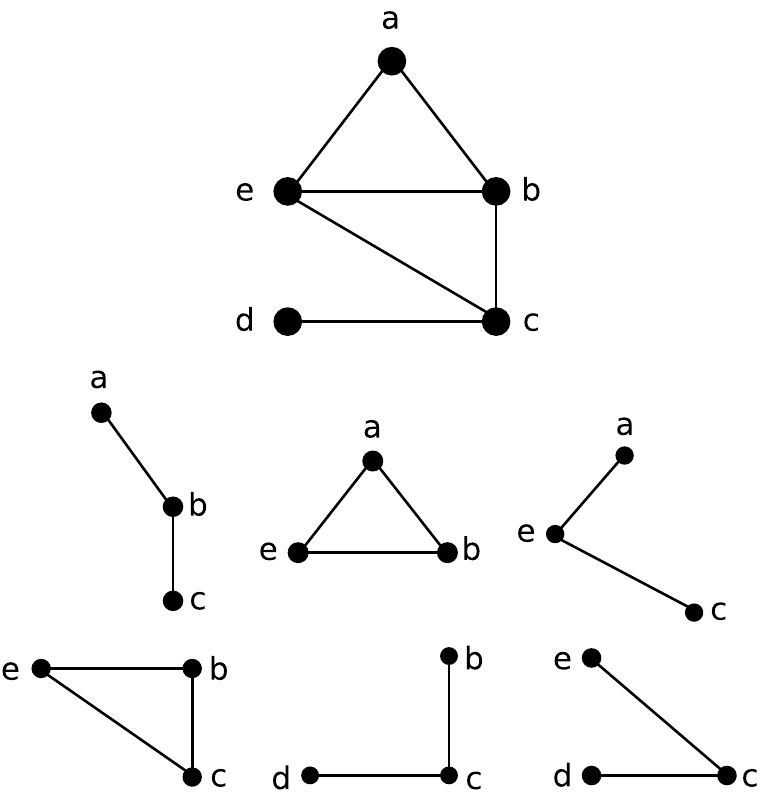}
    \caption{}
    \label{fig:brady_graphlets_left}
    \end{subfigure}
    \begin{subfigure}[t]{.45\linewidth}
    \includegraphics[width=1.1\textwidth]{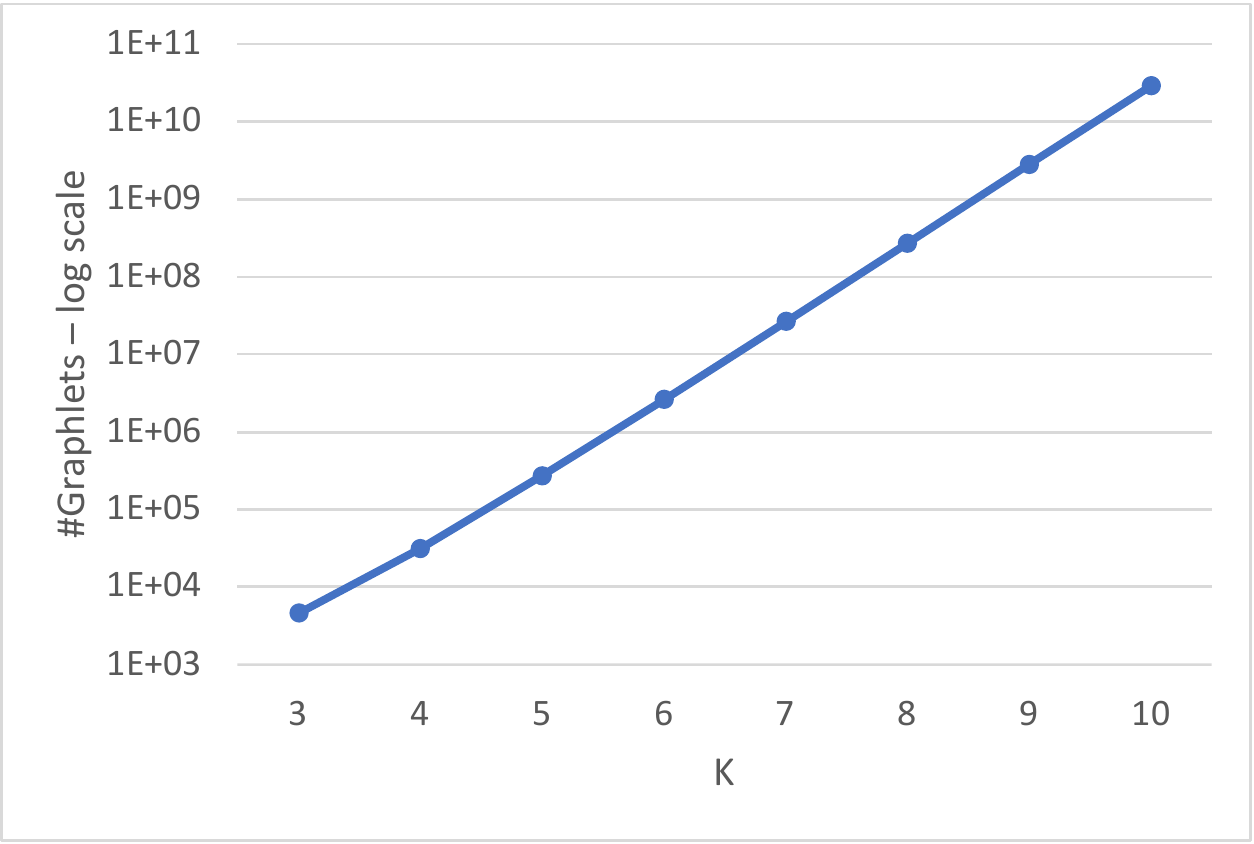}
    \caption{}
    \label{fig:brady_graphlets_right}
    \end{subfigure}
    \caption{(a) an example graph with its 3-graphlets below; (b) $k$-graphlet count in the \textit{Brady} network \cite{lasagne}.}
    \label{fig:brady_graphlets}
\end{figure}

If we lift the constraint on the size of the graphlets and instead want to enumerate all the connected induced subgraphs in a graph, we can reach $O(1)$ time per solution using push-out amortization \cite{Uno_pushout}. 

\begin{algorithm}[htb]
  \DontPrintSemicolon
  \KwIn{A graph $G = (V,E)$}
  \KwOut{All graphlets contained in $G$}

  \SetKwProg{myproc}{Function}{}{}
  \SetKwFunction{enum}{ENUM}
  \SetKw{output}{output}
  
  \myproc{\enum{$G, S, r$}}{
    \lIf{$d(r) = 0$}{\output S and \Return}
    Choose $v \in N(r)$\;
    \enum{$G/\{r,v\}, S \cup \{v\}, r$}\tcp*[l]{Contract the edge $\{r,v\}$}
    \enum{$G \setminus \{v\}, S, r$}\;
  } 
  
  \caption{Push-out amortized algorithm for graphlet enumeration \cite{Uno_pushout}}\label{alg:uno_graphlet}
\end{algorithm}

Algorithm \ref{alg:uno_graphlet} particularly shows the simplicity of binary partition algorithms: it is a classical algorithm which accumulates, by contraction, all vertices of the graphlet being constructed into one starting vertex ($r$) and when nothing else can be added it outputs the set $S$ containing the list of vertices used.
Despite its simplicity, this algorithm is proved to use only constant time between the output of two consecutive solutions together with linear additional space.
\begin{theorem}[\cite{Uno_pushout}]
Algorithm \ref{alg:uno_graphlet} runs in $O(1)$ amortized time and uses $O(|V| + |E|)$ additional space.
\end{theorem}

\paragraph{Reverse Search.} Avis and Fukuda, in their original paper about reverse search \cite{avis_reverse_1996}, already proposed an adaptation of this technique to unbounded-size graphlets. They proved the following theorem:
\begin{theorem}[\cite{avis_reverse_1996}]
There exists an implementation of Reverse Search for connected induced subgraph enumeration with $O(|V|\cdot|E|)$ delay and $O(|V| + |E|)$ space complexity.
\end{theorem}
The algorithm they propose can be easily adapted to enumerate all subgraph composed by \emph{at least} $k$ vertices or \emph{at most} $k$ vertices, but not for an \emph{exact} $k$.
Reverse search in this field has not much evolved over the years, but we point out that there is a very recent paper focused on mining maximal graphlets in graphs that adopts this strategy \cite{salem_rasma_2021}.

\section{Applications}
\subsection{On Simple Graphs}
Simple graphs have been heavily studied in the past because they are the traditional formalization of the entity-relationship concept, therefore there is plenty of work available in the literature.
Here, we give a look at some of the most recent applications of connected induced subgraphs in simple graphs.
We start by citing a work published in 2020 which is a hybrid between combinatorics and high-performance computing: the Graphlet Transform \cite{graphlet_transform}.
The graphlet transform is, informally, a function applied to the vector of frequencies of the different kind of graphlets contained in a graph.
It is based on the principle that each graphlet type has its own characteristics (i.e. the orbits of the nodes \cite{przulj_modeling_2004}), so counting how many graphlet of each kind exist in the graph gives important information about the local structural properties of the network.
The authors provide an efficient, low memory, way of computing the transform on general graphs for graphlets of up to 4 vertices.
This is because the exact frequency of each graphlet is required, and computing this number for graphlets of size higher than 4 quickly becomes impractical for the needs of the paper, therefore this would be one of the applications that would benefit from the improvements in the state of the art graphlet enumeration.

The GHuST framework \cite{GHuST} is another proposal for network analysis where the authors provide several measures that take advantage of the local structure of the graph in order to describe its (global) complex-network topology.
The authors study measures like the \emph{triangle rate}, \emph{triangle concentration}, \emph{triangle degree} and many others, all properties that can be computed with 2- and 3-node graphlets again for computational complexity concerns.

Przulj \cite{przulj_biological_2007}, the first to introduce the concept of graphlets for analysis, proposes ways for analyzing biological networks using the graphlet degree distribution (GDF), defined essentially to be the number of edges that touch a particular orbit of the graphlets, i.e. how many times a nodes assumes that particular orbit.

These kinds of subgraphs not only offer valuable information on the underlying network's properties, they also serve as building blocks in many fields such as biology and bioinformatics, sociology, computer networks, and transport engineering: in fact, in these kinds of networks, subgraphs are able to capture and highlight the community structure of the network itself \cite{santofortunato}, a concept that is ever so important when people or any other entities are tightly interconnected. 
Graphlets are also getting momentum for their numerous fields of application, such as understanding the role of genes in pathways and cancer mechanisms~\cite{windels2022graphlet}, and finding important nodes in a network~\cite{aparicio2019finding}. Special cases of $k$-graphlets, such as the aforementioned triangles when $k = 3$, or cliques, have been investigated for community detection purposes, as they help speeding up the process of identifying (possibly larger) communities in graphs \cite{klymko2014using, shaping_communities_triangles, triangles_social_cohesion, latent_triangle}. Also, edge $k$-graphlets find their usage in wildly different domains~\cite{chin2018subtrees}, such as graph databases where they are used for indexing purposes~\cite{subtree_graph_indexing}, and planar graphs~\cite{SUN2022127404}, to name a few.  

Moreover, additional uses are possible, like kernelization and centrality measures for nodes as we already pointed out in the Introduction.
Graphlet kernelization, in particular, is the process in which the subgraphs are used as kernel functions in Support Vector Machines, this time to allow for machine learning models to compare graphs.
For this reason the authors of \cite{graphlet_kernels} propose a kernel function that takes into account 3-, 4- and 5-graphlets  either using sampling or working on bounded degree graphs. Since we are comparing two graphs, a graphlet kernel is defined on a pair of graphs:
\begin{de}[Graphlet Kernel]
Given two graphs $G$ and $G'$, the graphlet kernel $k_g$ is defined as:
\[
k_g(G, G') = f^{{T}}_Gf_{G'}
\]
where $f_G$ is the vector of frequencies of all graphlets (up to 5 nodes) in $G$.
\end{de}
The main concern of this strategy is, again, the time required for the enumeration and counting of the considered graphlets; in particular, enumeration is used in bounded degree graphs given the fact that this is an easier case for it and for 3-, 4-, 5-graphlets. 
With the right algorithms this can be expanded to larger subgraphs.

Focusing on the graph analysis part there is a work that proposed the extension of graphlets to directed graphs, in order to analyze the \emph{world trade network} (WTN), to get a clearer picture of the modularity of that graph which, in short, models the exchange of goods between countries \cite{directed_graphlet_2016}. 
In this network, vertices represent countries and a directed edge connects the country exporting the goods to the one receiving them: it is necessarily a directed graph since in the import-export context the role of two countries is very different. 
Other previous works on this graph included, for instance, the discovery of its power-law degree distribution, and that the graph is organized in a ``core-periphery'' fashion.
Moreover, it was shown that even the position of a country in the WTN affected its economy. 
After adapting the concept of graphlet to directed graphs by considering only subgraphs of up to 5 vertices containing only non-anti-parallel edges\footnote{A graphlet is ignored if it contains both the edge $(u, v)$ and the edge $(v,u)$, in order to preserve the shape of it (i.e. the number of edges and the orbits).}, the authors of \cite{directed_graphlet_2016} proceed with the analysis of not only the WTN, but also biological protein-protein interaction networks (metabolic networks).
By analyzing various metrics involving graphlets, like the Relative Graphlet Frequency Distribution Distance\footnote{This metric compares two networks according to their relative distributions of directed graphlet frequencies, i.e. the number of all the different kinds of graphlets contained in the two graphs.}, the Graphlet Degree Distribution Agreement\footnote{This metric compares graphlet degrees, i.e., $d_{G}^{j}(k)$ represents the number of nodes in $G$ that touch orbit $j$ exactly $k$ times.} and the Graphlet Correlation Distance\footnote{This is the Euclidean distance between two Directed Graphlet Correlation Matrices (using Spearman's correlation).}, they uncovered that the WTN has more of a ``core-broker-periphery'' organization rather than only core-periphery, a more robust organization where the broker nodes act as intermediaries between two other countries. Using the same techniques, the authors are also able to make a prediction of the GDP\footnote{Gross Domestic Product.} of each country.
In the biological enviornment, they were able to observe that enzymes involved in similar patterns of human metabolic reactions are also involved in similar biological functions. Not only that, enzymes that are similarly wired in these reactions in other species, conserved the biological function observed in humans.

\subsection{On Temporal Graphs}\label{sect:state_temporal_graphs}
When dealing with temporal graphs, many different formalization of the concept of connected induced subgraph are possible. What do we consider as a temporal graphlet? There are many different possible answers. 
For example we could say that a clique is relevant whenever it is formed, no matter how many time frames it occupies, or we could consider only large cliques that persist for at least $t$ time frames in the temporal graph. 
For graphlets the same reasoning can be applied, but we can also monitor their evolution over time. While maximal cliques cannot be expanded by their definition, graphlets can evolve and expand: if we select one vertex of a given graphlet, we may observe changes in the number and shapes of graphlets in which it participates, so there are countless different approaches that in principle could be all valid.
Floros et al. \cite{graphlet_spectrograms_temporal} propose an analysis of virus-research collaboration networks using temporal graphlets of up to 5 vertices.
The work aims to analyze the collaboration between medical authors during, but not limited to, the COVID-19 pandemic. The analysis is based on spectral descriptors that are, informally, vectors describing the frequency of each vertex involved in each kind of graphlet. While this is not a new concept in the literature regarding graphlets \cite{the_graphlet_spectrum}, it is a powerful approach that is, however, limited due to the sheer size of vectors and matrices involved in the game and the exponential number of graphlets if we allow their size to go up to 10 (e.g. see Figure \ref{fig:brady_graphlets_right}).

One very recent, and interesting analysis conducted on temporal graph uses graphlets as event precursors \cite{graphlet_event_precursor}. 
The context is the one of social networks, and the work aims to establish a connection between user interactions and real-life (unexpected) events, like, for instance, the visit of the French President Macron to Rouen in October 2019. Another analyzed event is the agreement of the European Council on the Common Agricultural Policy, reached in late October 2020.
The authors of the analyses aim to show that events like these can be predicted by inspecting graphlet structure on the two previous days in social networks (e.g. Twitter and Facebook), with the graphlet size $k \leq 5$.
The authors use the algorithm \textsc{Orca} \cite{orca} for the enumeration of graphlets and orbits in these graphs, for each snapshot. 
Then they construct frequency vectors for each of the 30 different graphlet types for each snapshot and normalize the values in them by subtracting the mean and dividing by the standard deviation of each entry of the different vectors. The real analysis takes place by considering the difference between entries of the vectors in subsequent snapshots, calling it the \emph{velocity}, then they also compute the \emph{acceleration} by subtracting two subsequent velocities: the top graphlets with highest velocity and acceleration are considered as event precursors.
The analysis then consists in detecting the most significant graphlet categories according to the variation in velocity over time. What they discovered is that they could see a significant change in the structure of the graph roughly two hours before the event.
The analysis is intended to be conducted, in real-life, with domain experts that can understand what that event will be and then alert someone in charge of handling that event.
The proposed approach is interesting, although still in early stages and we believe that it could be improved and extend by the means of an improved state of the art for graphlet enumeration algorithm.

Another interesting topic on temporal graphs is their analysis based on common properties and centrality measures, which have to be redefined to take into account the dynamic aspect of these graphs.
Using redefined versions\footnote{These versions resemble the traditional ones for simple graphs, with the exceptions that they now take into account \emph{temporal paths} instead of classical paths.} of the \emph{closeness} and \emph{betweeness} centrality, the authors of \cite{pereira_twitter_2016} perform an analysis of a snapshot of the Twitter social network and examine the changes in those measures.
The importance of this kind of analysis, especially in social networks, is due to the fact that these measures along with the temporality of the graph, allow to identify not only the most influential nodes in the network, but also to keep track of how the information flows across the network, and which nodes helped the most in spreading it. 
It is a topic that is really important nowadays given the abnormal presence of fake news on the Internet.

There are also new proposed centrality measures for temporal graphs: in \cite{bonchi_temporal_centrality}, there is a new temporal betweeness centrality definition, based on a new definition of \emph{shortest-fastest paths}, in turn. The latter definition takes into account not only space or time for the paths but both at the same time, infusing more information into a single concept, providing claimed better results in network analysis compared with old centrality measures.

\chapter{New Algorithm for Graphlet Enumeration}\label{chap:ksquare}

In this Chapter we introduce the first algorithm for $k$-graphlet enumeration with theoretical $O(k^2)$ output-sensitive amortized time complexity (see Section \ref{sub:complexity_classes}), dropping the state-of-the-art dependency on $\Delta$, the maximum degree of the graph \cite{SIMPLE_graphlets}.

This Chapter is part of a manuscript that has been submitted to the \emph{Algorithmica} journal \cite{Rucci:2024Algorithmica} in collaboration with other authors.
The manuscript also covers the topic of the enumeration of $k$-edge subgraphs (see Chapter \ref{chap:background}) in $O(k)$ amortized time per subgraph, but for the sake of this thesis we only present the part related to $k$-graphlet enumeration.

The time complexity of the algorithm we are going to design here is $O(|V| + |E| + Nk^2)$, an output-sensitive time complexity as the cost grows with $N$, which is the number of $k$-graphlets reported.
Additionally, we show that for bounded-degree graphs, i.e. for graphs with $\Delta = O(1)$, the amortized cost per graphlet decreases to $O(k)$.
Finally, we point out that $k$ is usually very small in the context of graphlet enumeration, and often $k = O(1)$, bringing the cost down to $O(1)$ time per graphlet reported.

We first introduce the additional notions that we need in order to achieve the aforementioned bound, including an intermediate $O(|E|)$ amortized time algorithm.
Then, we discuss the $O(k^2)$ amortized time main algorithm, which is based on the binary partition strategy.
As this is a recursive algorithm, we will deeply exploit its corresponding recursion tree to prove its running time:
our goal is to establish a connection between the number of leaves and the number of internal nodes, i.e. between the number of solutions and the number of recursive calls needed to enumerate all of them.
In particular, if we are able to have a tree where each node either is a leaf or produces exactly two children, the number of internal nodes will be bounded by the number of leaves, allowing us to amortize the cost of the former on the latter.

\section{Background Concepts}
We start by introducing the notions of \emph{fruitful} call and \emph{removable} vertex.
\begin{definition}[Fruitful Call]
    In a recursive $k$-graphlet enumeration algorithm, a recursive call is said to be \emph{fruitful} if and only if it or its descendants will produce at least one $k$-graphlet.
\end{definition}

This definition is crucial for us to avoid unnecessary work that will lead to a non-output leaf, i.e., a call where the algorithm does not produce children calls while also not finding a $k$-graphlet, ultimately leading to wasted time.
The next definition builds upon this concept.

\begin{definition}[Removable Vertex]\label{def:removable_vertex}
    Given a graph $G = (V, E)$ and two vertices $x, r \in V, x \neq r$, $x$ is said to be \emph{removable} if and only if the recursive call associated to $G \setminus \{x\}$, using $r$, is fruitful.
\end{definition}

A removable vertex is a safe vertex to perform the binary partition on.
In fact, while using vertex $r$, we are guaranteed that we will ultimately find a $k$-graphlet while also removing that vertex from $G$.
These two concepts will drive our algorithms through the enumeration space, avoiding calls that will lead to non-amortizable time costs.

Next we show how to identify fruitful calls and removable vertices in any graph $G$. 
Both procedures involve a modified version of the well-known BFS traversal, adjusted to our need of finding $k$ connected vertices.
In particular, given a vertex $r$, the BFS will check for the existence of at least $k$ connected vertices to $r$, with the exception of vertex $x$ when checking if $x \in V$ is removable, as explained in the following Lemma.
\begin{lemma}\label{lemma:removable_check_time}
    We can check if a recursive call is fruitful and if a vertex $x \in V$ is removable in $O(k^2)$ time each.
\end{lemma}
\begin{proof}
    Both checks implement a modified version of the classical BFS traversal of the graph $G$ starting from vertex $r \in V$, hence the initial cost is $O(|V| + |E|)$.
    However, we stop the BFS as soon as we discover $k$ new vertices, thus these BFS variants will explore at most $k$ adjacency lists for a total of $O(k^2)$ edges, bringing the total cost to $O(k + k^2) = O(k^2)$ each.
    For the \textsc{Removable} check, there is the additional cost of ignoring the vertex $x$: this is $O(k)$ as $x$ can be found at most once per adjacency list, and the BFS will explore at most $k$ of them.
    Therefore the total cost of $O(k^2)$ does not change.
\end{proof}

It may happen, sometimes, that we are forced to use a vertex in a particular run of the algorithms, for instance when that vertex is the only neighbor available to be picked, or when its removal will disconnect the portion of the graph currently in our hands.
We formalize this idea with the definition of \emph{mandatory vertex}, which is defined as follows.

\begin{definition}[Mandatory vertex]\label{def:mandatory_node}
    Given a graph $G = (V, E)$ and a vertex $r \in V$, a vertex $v \in V, v \neq r$ is \emph{mandatory} for $G$ and $r$, if the recursive call associated to $G \setminus \{v\}$ and $r$ is not fruitful.
\end{definition}

Whenever clear from the context we may omit $r$ for conciseness. 
Notice that a mandatory vertex corresponds to a bridge in the graph, as removing it will disconnect $G$; therefore, the following holds.

\begin{corollary}
    If a vertex $v \in V$ is mandatory for $G = (V, E)$ and $r \in V$, then $v$ is an articulation point of $G$.
\end{corollary}

\section{Linear Amortized Time Algorithm}

The last ingredient for our $O(k^2)$ amortized time strategy consists of Algorithm~\ref{alg:linear_enum_graphlet}, which is a linear amortized time algorithm for $k$-graphlet enumeration. 
While it may appear as a step back with respect to the state-of-art amortized complexity of $O(k^2\Delta)$ 
\cite{SIMPLE_graphlets}, we will use this procedure only in cases where we can ensure that the input graph size is sufficiently small.

\begin{algorithm}[htb]
  \DontPrintSemicolon
  \KwIn{A graph $G$, a vertex $r$, a vertex set $S$, an integer $k$}
  \KwOut{All $k$-vertex induced graphlets of $G$ containing vertex $r$.}
  \SetKwProg{myproc}{Function}{}{}
  \SetKwFunction{enum}{LINEAR\_ENUM}
  \SetKwFunction{removable}{REMOVABLE}

  \myproc{\enum{$G$,$S$,$k$}}{%
      \lIf{$k = 0$}{\textbf{output} $S$ \textbf{and} \Return}
      Mark all mandatory vertices reachable from $r$ using a linear time articulation-point-finding algorithm \cite{tarjan_articulation_points, DBLP:journals/ipl/Tarjan74}\;\label{alg:linear_enum_graphlet:line:tarjan}
      \While{$\exists$ a mandatory vertex $u$ connected to $r$}{
         $G \gets G / \{r, u\}$\;
         $S \gets S \cup \{u\}$\;
         $k \gets k-1$\;\label{alg:linear_enum_graphlet:line:mandatory}
         
         \lIf{$k = 0$}{%
           \textbf{output} $S$ \textbf{and}
           \Return
        }%
      }

      $z \gets $ arbitrary neighbor of $r$\tcp*[l]{$z$ is guaranteed non-mandatory} 
      \enum{$G / \{r, z\}, S \cup \{z\}, k-1$}\;
      \enum{$G \setminus \{z\}, S, k$}\;

  }%
  \caption{$O(|E|)$ amortized time algorithm for $k$-graphlet enumeration}
  \label{alg:linear_enum_graphlet}
\end{algorithm}

In Algorithm~\ref{alg:linear_enum_graphlet}, line~\ref{alg:linear_enum_graphlet:line:tarjan} asks to find all mandatory vertices reachable from $r$: this can be achieved in $O(|E|)$ time by using an algorithm for finding all articulation points of the component connected to $r$; 
examples of algorithms able to do so in linear time can be found in \cite{tarjan_articulation_points, DBLP:journals/ipl/Tarjan74}.
After marking them, we need to count how many vertices get disconnected from $r$ for each articulation point.
All cut vertices that leave the graph with less than $k$ connected vertices to $r$ are declared mandatory (this check can be done at the same time when marking articulation points).
After the marking phase, we start including mandatory vertices connected to $r$ to the $k$-graphlet under construction in $S$ by shrinking the edge $\{r, u\}$, where $u \in N(r)$ is mandatory, and we repeat the process until either we produce a solution, or we do not have any mandatory node connected to $r$ left.
Then, we do binary partition on an arbitrary neighbor of $r$, that is now guaranteed to be non-mandatory, otherwise we would have picked it at the earlier stage.

The following Theorem proves the correctness of this strategy.

\begin{theorem}\label{th:correctness_linear_enum}
    Algorithm~\ref{alg:linear_enum_graphlet} finds all $k$-graphlets in $G$ that contain vertex $r$.
\end{theorem}
\begin{proof}
    Algorithm~\ref{alg:linear_enum_graphlet} starts a BFS traversal of $G$ from $r$, and marks all mandatory nodes according to Definition~\ref{def:mandatory_node}.

    The connected component of up to $k$ vertices connected to $r$ will be shrunk, one vertex at a time, into $r$, decreasing the value of $k$ each time.
    If this procedure lowers the value of $k$ to zero, then $S$ is output as it contains $k$ connected vertices, as all of them where reachable from $r$, by definition of connected component. 
    Thus $G[S]$ is a $k$-graphlet of $G$ and this call will produce a leaf in the recursion tree associated to the execution of the algorithm.
    
    Otherwise, the algorithm picks an arbitrary neighbor $z$ of $r$, if it exists, and proceeds recursively to enumerate $k$-graphlets using vertex $z$, then without vertex $z$ as per the binary partition rule, creating two recursive children in the recursion tree.
    These recursive children will be disjoint because one will use vertex $z$ and the other removes it from the graph.
    Therefore we can conclude that Algorithm~\ref{alg:linear_enum_graphlet} outputs all and only $k$-graphlets in $G$ that include vertex $r$.
\end{proof}
Notice that, in Algorithm~\ref{alg:linear_enum_graphlet}, vertex $z$ could not exist in some cases.
However, we use this algorithm as a sub-routine of a broader strategy that will ensure either that $z$ exists, or that there exists a connected component of $k$ mandatory vertices connected to $r$, ensuring that Algorithm~\ref{alg:linear_enum_graphlet} will never fail at finding a solution.
Now we can prove the amortized running time to be $O(|E|)$.

\begin{theorem}\label{th:linear_graphlet_enum_delay}
    Algorithm~\ref{alg:linear_enum_graphlet} requires $O(|E|)$ amortized time per $k$-graphlet.
\end{theorem}
\begin{proof}
    First, we have to mark all mandatory vertices of the graph: this can be done in linear $O(|V| + |E|)$ time by using one of the well-known algorithms in the literature for finding bridges and articulation points such as \cite{DBLP:journals/ipl/Tarjan74, tarjan_articulation_points}.
    In particular, our mandatory vertices will be all the articulation points that, when removed, leave less than $k$ nodes connected to $r$ in the graph.
    Once all mandatory vertices are marked, the shrinking operation of line \ref{alg:linear_enum_graphlet:line:mandatory} will be executed at most $k$ times, each requiring $O(d(u))$ time, for a total of $O(|E|)$ time as it is a sum of distinct node degrees. 
    If we reach $k = 0$, we will produce a solution, hence this call will be a leaf node of the recursion tree, with $O(|E|)$ time complexity.

    Otherwise, we perform binary partition on $z$, and we can pay the $O(\Delta)$ cost taken for shrinking $\{r, z\}$ and for removing $z$ from $G$ because $\Delta \leq |V| = O(|E|)$ for any connected graph, producing two fruitful children calls.

    Since every recursive call either produces a solution or two fruitful children calls, the recursion tree associated to Algorithm~\ref{alg:linear_enum_graphlet} will have $N$ leaves, where $N$ is the number of $k$-graphlets found, and at most $N-1$ internal nodes, each with cost $O(|E|)$.
    Thus, the total running time will be $O(N\cdot|E|)$, which is $O(|E|)$ amortized time per $k$-graphlet.
    %
\end{proof}

\section{\texorpdfstring{$O(k^2)$}{O(k2)} Amortized Time Algorithm}

We are finally able to introduce the algorithm with $O(k^2)$ amortized time complexity for $k$-graphlet enumeration, which is summarized in Algorithm~\ref{alg:ksquare_graphlet}.

\begin{algorithm}[htb]
  \DontPrintSemicolon
  \KwIn{A graph $G$, an integer $k$}
  \KwOut{All $k$-vertex induced graphlets of $G$}
  \SetKwProg{myproc}{Function}{}{}
  \SetKwFunction{enum}{ENUM}
  \SetKwFunction{removable}{REMOVABLE}
  \SetKwFunction{linearenum}{LINEAR\_ENUM}
  \SetKwFunction{fruitful}{FRUITFUL}

  \ForAll{$v \in V(G)$ in an arbitrary ordering}{
    \If{\fruitful{$G, v, k$}} {
      \enum{$G, v, \{v\}, k$} \tcp*[h]{Find all graphlets containing $v$}\;
    }
      $G\gets G \setminus \{v\}$ \tcp*[h]{Remove $v$ from subsequent iterations}\;
  }

  \myproc{\enum{$G$,$r$,$S$,$k$}}{%

    \lIf{$k = 0$}{
        \textbf{output} $S$;
        \Return
    }
    \While(\tcp*[h]{Follow the chain}){$|N(r)| = 1$}{\label{alg:ksquare_graphlet:line:while}
      Let $\{v\} = N(r)$\;
      $S \gets S \cup \{v\}$\;
      $k \gets k - 1$\;
      $G \gets G / \{r, v\}$\;
      \lIf{$k = 0$}{
        \textbf{output} $S$;
        \Return
      }
    }
    $x, y \gets$ two arbitrary neighbors of $r$\;
    \uIf{\removable{$G, r, x, k$}}{
      \enum{$G / \{r, x\}, S \cup \{x\}, k-1$}\;\label{alg:ksqare_graphlet:line:take_x}
      \enum{$G \setminus \{x\}, S, k$}\;\label{alg:ksqare_graphlet:line:remove_x}
    }
    \uElseIf{\removable{$G, r, y, k$}}{
      \enum{$G / \{r, y\}, S \cup \{y\}, k-1$}\;\label{alg:ksqare_graphlet:line:take_y}
      \enum{$G \setminus \{y\}, S, k$}\;\label{alg:ksqare_graphlet:line:remove_y}
    }
    \Else(\tcp*[h]{$|V| < 2k$ by Lemma \ref{lemma:g_2k_nodes}}){
      \linearenum{$G, r, S, k$}\;
    }
  } 

  \caption{$O(k^2)$ amortized time $k$-graphlet enumeration}\label{alg:ksquare_graphlet}
\end{algorithm}

Let us explain the details of Algorithm~\ref{alg:ksquare_graphlet}. 
It follows an usual binary partition strategy in its outer loop, which first calls the internal enumeration procedure on a vertex $v$ including it in all the $k$-graphlets listed, and then removes $v$ from the input graph entirely.
We have added a fruitful check to avoid wasting time by using a vertex that will not be part of any solution.

The function \textsc{Enum} is a recursive procedure that adds one suitable vertex at a time to the set $S$, which again will contain all the vertices part of a solution, i.e. such that $G[S]$ is connected and $|S| = k$.
First, it checks for the special case of $r$ having only one neighbor $v$: in this scenario we have no other choice but to include $v$ in $S$, shrink the edge $\{r, v\}$ and proceed. 
If this special case happens exactly $k$ times, $S$ will be output as it contains the vertices of a $k$-graphlet.
Otherwise we proceed by choosing a suitable vertex to proceed with the recursion. 
In particular, we choose two arbitrary neighbors (which are guaranteed to exist) and check if at least one of them is removable according to Definition~\ref{def:removable_vertex}; if we succeed, we spawn the corresponding two recursive children taking that vertex and then removing it, otherwise we switch to Algorithm~\ref{alg:linear_enum_graphlet}.
When this happens, we can guarantee that the graph to be processed by Algorithm~\ref{alg:linear_enum_graphlet} will have only $O(k^2)$ edges, due to the following results.

\begin{lemma}\label{lemma:reachable_vertices}
The sets of vertices that become unreachable in $k$ steps after removing $x$ or $y$ from $G$ in Algorithm \ref{alg:ksquare_graphlet}, are disjoint.
\end{lemma}
\begin{proof}
Let $z$ be a vertex that becomes unreachable from $r$ after the removal of one of $x$ or $y$, say $y$ without loss of generality. Since $z$ is unreachable, every shortest path from $r$ to it, must pass through $y$.

For the sake of contradiction, assume that $\exists\ x \neq y \in N(r)$ such that $z$ is unreachable from $r$ in $G \setminus x$. This means that the shortest paths from $r$ to $z$ must pass also through $x$ in addition to $y$.
Let $\pi$ be one of these shortest paths and assume, without loss of generality, that $\pi = \{r, x, \dots, y, \gamma, z\}$ where $\gamma$ is a subset of the path. 
Let $\pi' = \{r\} \cup \{y, \gamma, z\}$ obtained by removing $x$: this is another shortest path from $r$ to $z$ with $|\pi'| < |\pi|$, meaning that $\pi$ was not a shortest path from $r$ to $z$, a contradiction.
We can then conclude that if $z$ is not reachable after removing $y$, it must be reachable from $x$, and the thesis follows.
\end{proof}

\begin{lemma}\label{lemma:g_2k_nodes}
    If $x, y$ are both non-removable, then $|V| < 2k$.
\end{lemma}
\begin{proof}
    If $x$ and $y$ are non-removable, a BFS from $r$ in $G \setminus \{x\}$ and $G \setminus \{y\}$  will find less than $k-1$ nodes, otherwise $x$ and $y$ would be removable.

    Since nodes becoming unreachable from $r$ in $G \setminus \{x\}$ are disjoint from the ones in $G \setminus \{y\}$ by Lemma~\ref{lemma:reachable_vertices}, we can infer that when removing $x$ (resp. $y$), graphlets will include the vertex $y$ and the set $Y$ (resp. $X$) of nodes reachable from $y$ (resp. $x$), and they will be less than $k-1$ by our hypothesis that also $y$ is non removable, i.e. $|N(r) \cup \{r\} \cup Y| < k$ and $| \{r\} \cup X| < k$. Summing both sides of the inequalities we have $|\{r\} \cup X \cup Y| < 2k$.
\end{proof}

When $G$ has at most $2k$ vertices, by construction it must have $O(k^2)$ edges, as we wanted to use \textsc{Linear\_Enum} on a sufficiently small graph.
We are now ready to prove the correctness and to bound the amortized running time of Algorithm~\ref{alg:ksquare_graphlet}, our main result.

\begin{theorem}\label{th:correctness_graphlets}
    Algorithm \ref{alg:ksquare_graphlet} outputs all and only $k$-graphlets in a graph $G = (V, E)$, without duplicates.
\end{theorem}
\begin{proof}
    We call $k_0$ the initial value of $k$ during the recursion of the \textsc{Enum} procedure.
    The algorithm loops through all vertices of the input graph $G$ according to a fixed, arbitrary ordering. 
    Before calling the internal procedure \textsc{Enum}, the algorithm checks if there exists at least one solution in the graph $G$ using vertex $v$. 
    If the call to \textsc{Fruitful} fails, $v$ is immediately removed from $G$ and the algorithm proceeds with the next vertex.

    The correctness of the \textsc{Enum} procedure follows from the correctness of \textsc{Linear\_Enum} (see Theorem~\ref{th:correctness_linear_enum}), as it performs binary partition on a suitable vertex found by checking if it is removable.
    The main differences here lie in the while loop of line~\ref{alg:ksquare_graphlet:line:while} and the switch to the \textsc{Linear\_Enum} procedure whenever we cannot find a removable vertex.
    In particular, while vertex $r$ has only one neighbor $v$, it is added to $S$, the edge $\{r, v\}$ is shrunk and $k$ is decreased.
    If we find $k$ vertices in this way, $S$ is output and it will contain only connected vertices because they were connected to $r$ at every iteration, thus $G[S]$ is connected and $|S| = k_0$. 
    This type of call creates a leaf in the recursion tree associated to \textsc{Enum}, and outputs a solution.
    When we exit the while loop, $r$ is guaranteed to have $|N(r)| \geq 2$, because $|N(r)| \neq 1$ and the outer-level \textsc{Fruitful} call ensures that there is a solution in $G$ using $r$, so $|N(r)|$ cannot be zero.
    If one of $x, y \in N(r)$ is removable we are guaranteed that we will find a solution by binary partitioning on $x$ (resp. $y$) by the same reasoning of Theorem~\ref{th:correctness_linear_enum}.
    If neither $x$ or $y$ are removable, we know that $G$ contains less than $2k$ nodes by Lemma~\ref{lemma:g_2k_nodes} and we switch to Algorithm~\ref{alg:linear_enum_graphlet}, whose correctness is proven by Lemma~\ref{th:correctness_linear_enum}.
    
    Due to the fixed ordering, every vertex $u < v$ will be already removed from $G$ by the time we initialize $r = v$, therefore no $k$-graphlet will be listed twice.
\end{proof}

\begin{theorem}\label{th:ksquare-graphlet-complexity}
    Algorithm \ref{alg:ksquare_graphlet} runs in $O(k^2)$ amortized time per $k$-graphlet or $O(|E| + Nk^2)$ total time, where $N$ is the number of $k$-graphlets reported.
\end{theorem}
\begin{proof}
    For clarity, we call $k_0$ the initial $k$, whenever we need to distinguish them, as the value of $k$ is decreased during the recursion.
    Note that $k \le k_0$ always.

    First, note that each node must be deleted from the graph (not just ignored), this can be achieved by using appropriate data structures that allow for $O(d(v))$ deletion time for any $v \in V$ and shrinking.
    The same reasoning applies when restoring a deleted vertex when the recursion backtracks.\footnote{This method is folklore by exploiting threaded doubly-linked adjacency lists, keeping pointers to the last element of each list, and for each vertex, a pointer to all neighbors' adjacency lists where it appears, allowing for $O(1)$ time deletion and restoring of each value.}
    In order to prove this result, we show that the recursion tree associated with Algorithm~\ref{alg:ksquare_graphlet} is made by recursive calls that cost $O(k^2)$ (after proper amortization) and either output a solution or have two recursive children.
    This is sufficient to prove the result as it gives a recursion tree with $N$ leaves, where $N$ is the number of $k$-graphlets found, and at most $N-1$ internal nodes, each with cost $O(k^2)$.
    
    Firstly, observe that every recursive call to \textsc{Enum} and \textsc{Linear\_Enum} produces at least one solution due to the usage of \textsc{Fruitful} and \textsc{Removable}. We now show that the recursive calls cost $O(k^2)$ amortized time each.

    The while loop of line \ref{alg:ksquare_graphlet:line:while} will be executed at most $k$ times, and each iteration will cost $O(1)$ time as $|N(r)| = 1$.
    If we reach $k = 0$, we output a solution and this \textsc{Enum} call will be a leaf of the recursion tree, where the algorithm spent $O(k)$ time. 
    If we exit the while loop with $k > 0$ we will proceed with the binary partition strategy on a suitable node. 
    To find such a node we need to call the \textsc{Removable} procedure, which will perform a truncated BFS with cost $O(k^2)$ due to Lemma~\ref{lemma:removable_check_time}.
    Now, let us analyze each possible case of removable vertices separately:

    Assume that we perform the binary partition on $x$ (the argument is the same for $y$), we have the following two cases.

    \textbf{(a)} $d(x) \leq 2k$:
    the binary partition on vertex $x$ will require $O(k)$ time for preparing the graphs $G / \{r, x\}$, and $G \setminus \{x\}$ (and to restore $G$ when backtracking).
    The total cost of this recursive call will then be $O(k + k^2) = O(k^2)$, and it will produce two fruitful (because of the \textsc{Removable} check) children calls.

    \textbf{(b)} $d(x) > 2k$:
    the construction of $G / \{r, x\}$ and $G \setminus \{x\}$ now cost $O(d(x))$ that has to be amortized, as it could be $O(|V|)$; we do this by suitably dividing this cost among some of the descendant calls of the recursive call $C$ corresponding to \textsc{Enum}($G / \{r, x\}, r, S \cup \{x\}, k-1$) of line~\ref{alg:ksqare_graphlet:line:take_x}.
    Let $x_1, x_2, \dots, x_{d(x)}$ be the neighbors of $x$. 
    In $G / \{r, x\}$, at least $d(x) - k > d(x) / 2$ of these nodes are removable; note that these are now neighbors of $r$ and they could be all removed at the same time as $r$ would still have more than $k$ neighbors left. 
    This means that there exist at least $d(x) / 2$ ``negative'' recursive calls, i.e. performed on line~\ref{alg:ksqare_graphlet:line:remove_x} and line~\ref{alg:ksqare_graphlet:line:remove_y}, descending directly from $C$.
    We divide the $O(d(x))$ cost among these calls, observing that each of them receives an additional cost charge of $O(1)$ (see Figure~\ref{fig:graphlet-amortization}).
    Observe also that every recursive node can receive this cost only from one ancestor in the tree.
    Finally, observe that we do not actually choose which nodes are considered by the binary partition, so it could be that the algorithm switches to Algorithm~\ref{alg:linear_enum_graphlet} during this process.
    If this happens recall that, by Lemma~\ref{lemma:g_2k_nodes}, the remaining graph has less than $2k$ nodes; the number of remaining charges that were not yet amortized must also be $\le 2k$, as each charge was assigned based on a specific node that is still in the graph, so the cumulative cost of these remaining charges is just $O(k)$, which we can assign to the current recursive call.

    The last case, \textbf{(c)}, happens whenever neither $x$ or $y$ are removable.
    In this case, $G$ has at most $2k$ vertices and $O(k^2)$ edges by Lemma~\ref{lemma:g_2k_nodes}, and we can switch to the enumeration strategy of Algorithm~\ref{alg:linear_enum_graphlet}, with a graph that has $|E| = O(k^2)$. 
    Following Theorem~\ref{th:linear_graphlet_enum_delay}, this will grow a subtree such that every node will either output a solution or spawn two recursive children in $O(k^2)$ amortized time because the graph can only get smaller in subsequent recursive calls.

    Therefore, we can conclude that each recursive call spends $O(k^2)$ time and is charged by at most one ancestor with an additional $O(k)$ or $O(1)$ time, thus the total amortized time of Algorithm~\ref{alg:ksquare_graphlet} is bounded by $O(k^2)$.
\end{proof}

\begin{figure}[ht]
    \centering
    \includegraphics[width=.8\textwidth]{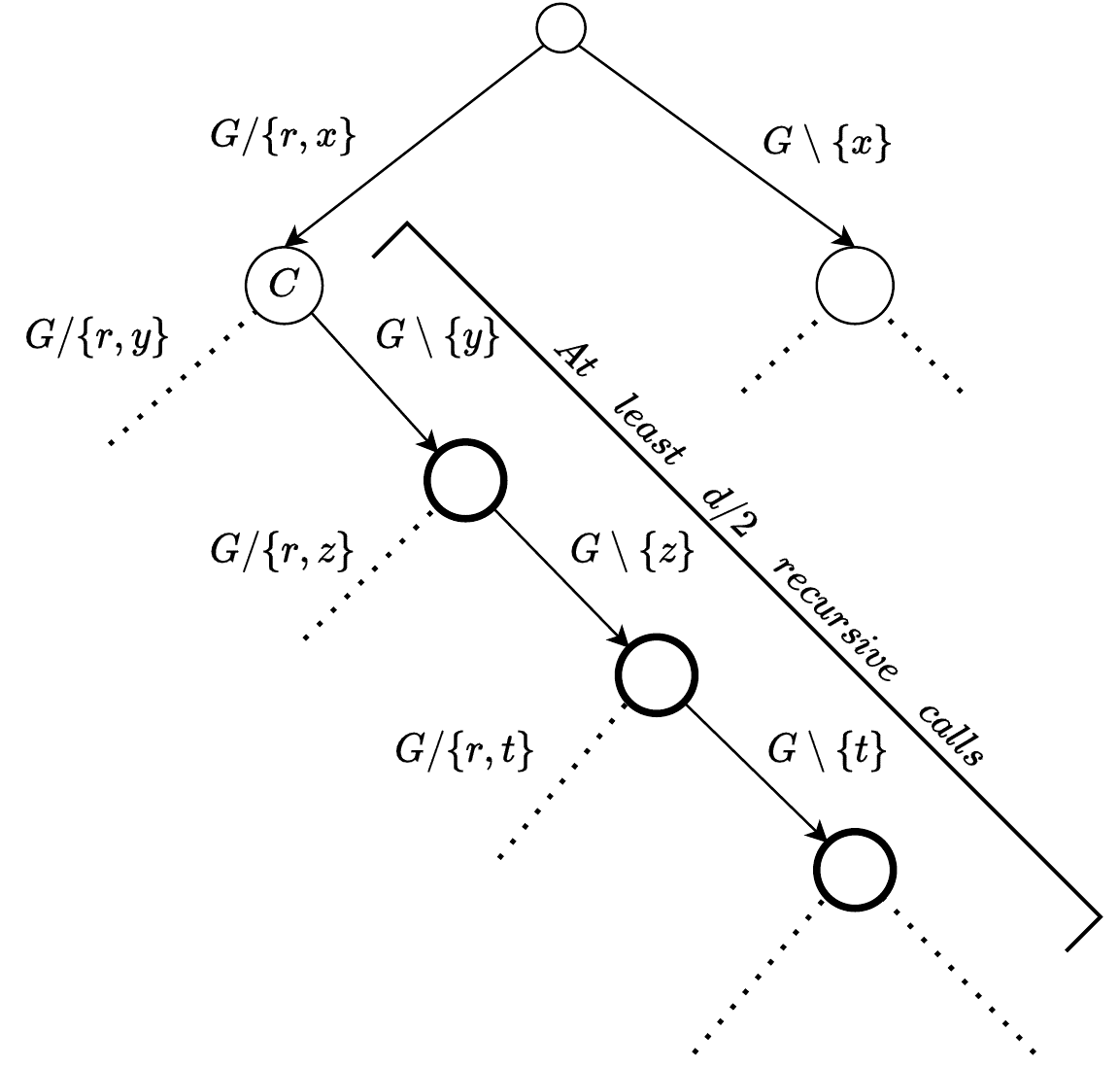}
    \caption{Amortizing the cost of the recursive call $C$ on its \emph{negative} descendants: we charge $O(1)$ to each highlighted call. These calls, by the proof of Theorem~\ref{th:ksquare-graphlet-complexity}, are at least $d(x) / 2$, so we can amortize the $O(d(x))$ cost on them.}
    \label{fig:graphlet-amortization}
\end{figure}

\subsection{Bounded-degree graphs}

The focus of the above analysis is providing bounds that depend only on $k$ and not on the size of the graph, and this is the main contribution of our proposed algorithm since typical scenarios consider values of $k$ in the single digits, while graphs can have millions or billions of vertices, and the maximum degree $\Delta$ can be in the same order of magnitude. 

Nonetheless, it is worth considering that there are graphs classes where the maximum degree $\Delta$ can be small. Specifically, graphs where $\Delta = O(1)$ are called \textit{bounded-degree} graphs. For such cases, we observe that a finer complexity analysis yields $O(k \cdot \min\{k, \Delta\})$ amortized time, which indeed results in $O(k)$ time per graphlet in bounded-degree graphs.

\begin{corollary}
    Algorithm~\ref{alg:ksquare_graphlet} has $O(k \cdot \min\{k, \Delta\})$ amortized time complexity, that is, $O(k)$ in bounded-degree graphs.
\end{corollary}
\begin{proof}
    The proof is obtained by retracing that of Theorem~\ref{th:ksquare-graphlet-complexity}, showing that every factor $O(k^2)$ in its complexity can be bounded by $O(k \cdot \min\{k, \Delta\})$: firstly, the truncated BFS for computing \textsc{Fruitful} and \textsc{Removable} can be performed in $O(k\Delta)$ time, as each adjacency list is not bigger than $\Delta$ and we are looking for just $k$ vertices; this means the corresponding $O(k^2)$ cost is bounded by $O(k\Delta)$ as well, and thus by $O(k \cdot \min\{k, \Delta\})$.

    Secondly, whenever the algorithm switches to Algorithm~\ref{alg:linear_enum_graphlet}, the resulting graph has $\le 2k$ vertices; each of these has at most degree $\Delta$, except for vertex $r$ that may have higher degree due to the shrinking operations. However, the degree of $r$ cannot be higher than the number of nodes in the current graphs, which is $O(k)$, thus the total number of edges $|E|$ is bounded by $O(k + k\Delta) = O(k\Delta)$. Since $|E|$ is clearly also bounded by $O(k^2)$ we have that the amortized cost of Algorithm~\ref{alg:linear_enum_graphlet} is also bounded by $O(k \cdot \min\{k, \Delta\})$.    
\end{proof}

\subsection{Extension to edge \texorpdfstring{$k$}{k}-graphlets}
Algorithm~\ref{alg:ksquare_graphlet} can be used to also list \emph{edge} $k$-graphlets, i.e., subgraphs constituted by exactly $k$ edges (see Section~\ref{bg:sec:subgraphs} and Figure~\ref{bg:subfig:edge-graphlets}) in $O(k^2)$ amortized time.
This is possible by transforming the input graph $G = (V, E)$ in its corresponding \emph{line graph} $L(G) = (V', E')$ where each edge of $G$ becomes a vertex of $L(G)$, and $L(G)$ contains an edge whenever two vertices of $G$ share an incident edge.

Formally,
$ V' = \left\{e \,\mid\, e \in E\right\} $ and $E' = \left\{ \{e, f\} \,\mid\, e, f \in E, e \cap f \neq \emptyset \right\}$.
We can then apply Algorithm~\ref{alg:ksquare_graphlet} to $L(G)$ to find graphlets there, in $O(k^2)$ amortized time per graphlet.
The $k$-graphlets found in $L(G)$ correspond to edge $k$-graphlets in $G$, and the construction of $L(G)$ requires $O(|V| + |E|)$ by using common graph traversal algorithms such as breadth- or depth-first search.

  \chapter{Cache-Aware Graphlet Enumeration}\label{chap:cage}
In Chapter~\ref{chap:ksquare} we provided a new algorithm that improves the state of the art for $k$-graphlet enumeration, bringing the theoretical time complexity down to $O(k^2)$, and even to $O(1)$ under certain circumstances, per graphlet enumerated.
However, it requires a high amount of shrinking and deletion operations on the vertices while unfolding the recursion, that can lead to bottlenecks in practical implementations. 
Moreover, as mentioned earlier in Chapter \ref{chap:background}, we conducted an analysis on the practical performance of existing algorithms, in particular KS-Simple \cite{SIMPLE_graphlets}, to discover that its corresponding recursion tree has often a tiny amount of failure leaves. 
This tells us that the work done by binary partition algorithms (for $k$-graphlet listing) is, for the most part, useful work.
As an example, say that the ratio of failure leaves to the total number of leaves in the recursion tree is $20\%$: this means that $80\%$ of the total work done by the algorithm has provided a solution and it was, indeed, useful work. 
We conclude that, if the ratio above is sufficiently small, then the shrinking and deletion operations of Algorithm~\ref{alg:ksquare_graphlet} may actually hinder the benefits of having only successful leaves in the recursion tree in practical scenarios.

Therefore, in this chapter we take a journey from the failure leaves observation towards the design of a practical, engineered and optimized algorithm for $k$-graphlet enumeration.
In particular, the key question that we address is ``how far can we push existing graphlet enumeration algorithms?''.
The question addresses the purest version of graphlet discovery, and its relevance is in building an efficient and general tool that can be exploited in any graphlet-related application.
We will show a hard limit faced by current enumeration strategies, and how to overcome it.
We refine the binary partition strategy to take full advantage of cache memory, while also cutting down the size of the recursion tree beyond what is possible with current enumeration algorithms. 
Recall Figure \ref{fig:brady_graphlets_right}, that showed how the number of $k$-graphlets increases exponentially with $k$, even in a small graph like the Brady network \cite{lasagne} (with 1,117 nodes and 1,330 edges);
the existing approaches deal with this problem by keeping $k$ small \cite{PGD,Kavosh,FASE,escape_algorithm,bhuiyan2012guise,melckenbeeck2019optimising,SIMPLE_graphlets} (e.g., 3 or 4), or by giving up exact results in favor of estimation \cite{bressan2018motif,graphlet_kernels,wang2014efficiently,elenberg2015beyond}. Indeed, according to a recent survey on motifs \cite{jazayeri2020motif}, papers focused on estimation are becoming more popular in recent years.

The contribution of this Chapter are the following: 
\begin{enumerate}\itemsep0em
\item We empirically observe that current enumeration methods cannot further reduce their enumeration tree as the number of \textit{failure} leaves, i.e. leaves which do not report a graphlet, is negligible (9\% or less) even with simple strategies.

\item Based on this, we aim to go beyond the already good performance of simple enumerators: we achieve this by designing an algorithm that \emph{better exploits cache memory} in the CPU, as we show using the Intel VTune Profiler tool. Intel CPUs have three cache levels (L1-L2-L3) to speed up the memory access to the RAM: the RAM module is the largest but slowest memory; 
L3 cache is faster but smaller than RAM (and typically shared among the CPU cores); L2 and L1 are even smaller and faster (typically one private L1-L2 per core). Our algorithm exhibits zero or few L3 cache misses (compared to RAM loads and stores) and L2-L3 cache-bound times have tiny values, in the range 0--5\%.
    
\item Finally, we break through the hard limit on the number of recursive calls, \emph{collapsing the three lowest levels of the recursion tree}. As the tree grows exponentially, this affects the largest levels.
    We also generate sets of solutions at once, in a compressed, yet easy to access, format. 
\end{enumerate} 
We call the resulting algorithm CAGE (Cache-Aware Graphlet Enumeration), which constitutes an evolution of the reference enumeration algorithm \cite{SIMPLE_graphlets} (KS-Simple) (see Section \ref{sec:ks_simple_baseline}), and outperforms it by \emph{over an order of magnitude}.

We evaluate CAGE against the state of the art, namely, KS-Simple, Kavosh~\cite{Kavosh}, and FaSE~\cite{FASE}, showing dramatically improved performance and scalability. We also include fast counting approaches~\cite{PGD,escape_algorithm} that follow a more analytical method, showing they are more competitive; however, by design they cannot run for $k>5$, whereas CAGE does not share this limitation. 

In what follows, for a set of nodes $S\subseteq V$, we define the neighborhood of $S$ as $N(S) = \cup_{u\in S}N(u) \setminus S$, the distance-2 neighborhood as $N^2(S) = N( N(S) ) \setminus (S \cup N(S))$ and the distance-3 neighborhood as $N^3(S) = N( N^2(S)) \setminus (S \cup N(S) \cup N^2(S))$.

\section{Failure Leaves Analysis}\label{cage:sec:failure_leaves}
In this section we provide a broader analysis of the failure leaves of the binary partition technique adopted by KS-Simple \cite{SIMPLE_graphlets}. 
Table \ref{tab:failure_leaves-appendix} highlights the trend of the dataset: the percentage of failure leaves relative to the total number of leaves in the recursion tree is impressively small (i.e. less than $9\%$). This is corroborated by the empirical analysis that we conducted on the whole dataset made up by 155 graphs, as this data analysis drove our design of CAGE (see Figure ~\ref{fig:chart_fail_leaves}).
\begin{table}[htb]
\caption{Failure leaves found by KS-Simple for $k=5, 9$. $\dagger$: execution stopped after 30 minutes.}
\label{tab:failure_leaves-appendix}
\centering
\small
\begin{tabular}{|c|c|r|r|} 
\hline
Graph                           & $k$ & \multicolumn{1}{c|}{\#Leaves} & \multicolumn{1}{c|}{\#Failure Leaves (\%)}  \\ 
\hline
\multirow{2}{*}{Brady}          & 5   & 272,286                     & 2,082 (0.76\%)                       \\
                                & 9   & 2,818,545,162               & 6,518,803 (0.23\%)                   \\ 
\hline
\multirow{2}{*}{cti}            & 5   & 7,439,769                   & 274 (0.004\%)                        \\
                                & 9   & 13,756,931,277              & 351,237 (0.003\%)                    \\ 
\hline
\multirow{2}{*}{Roadnet-TX}     & 5   & 26,530,251                  & 1,528,480 (5.76\%)                   \\
                                & 9   & 1,928,340,917               & 42,101,984 (2.18\%)                  \\ 
\hline
\multirow{2}{*}{RoadNet-CA}     & 5   & 38,943,018                  & 2,060,443 (5.29\%)                   \\
                                & 9   & 2,926,916,012               & 59,420,852 (2.03\%)                  \\ 
\hline
\multirow{2}{*}{IMDB$^\dagger$} & 5   & 414,675,826                 & 443 (0.0001\%)                       \\
                                & 9   & 357,607,147                 & 79 (0.00002\%)                       \\ 
\hline
\multirow{2}{*}{Wing}           & 5   & 4,183,189                   & 50,728 (1.21\%)                      \\
                                & 9   & 1,103,510,287               & 1,850,767 (0,17\%)                   \\
\hline
\end{tabular}%

\end{table}

For the sake of completeness, we report a comparative plot of the percentages found in the entire dataset in Figure \ref{fig:chart_fail_leaves}, for $k = 4, 5, 7, 9$, where the x-axis corresponds to the graphs, sorted in decreasing order of percentage, and the y-axis reports the percentage values. 
A clear long-tail distribution emerges from these charts, where we can see that only a small portion of the dataset as a higher number of failure leaves with respect to the majority of the graphs.
Nonetheless, this higher number does not reach even the 10\% threshold, while the rest of the dataset produces way less than 1\% of failure leaves, strengthening our hypothesis that any pruning rule added to this enumeration strategy will increase the computational complexity of the recursive calls in exchange for little to no benefits in terms of the recursion tree size.

\begin{figure}[htb]
    \centering
    \hfill
    \begin{subfigure}[b]{0.49\textwidth}
         \centering
         \includegraphics*[width=\textwidth]{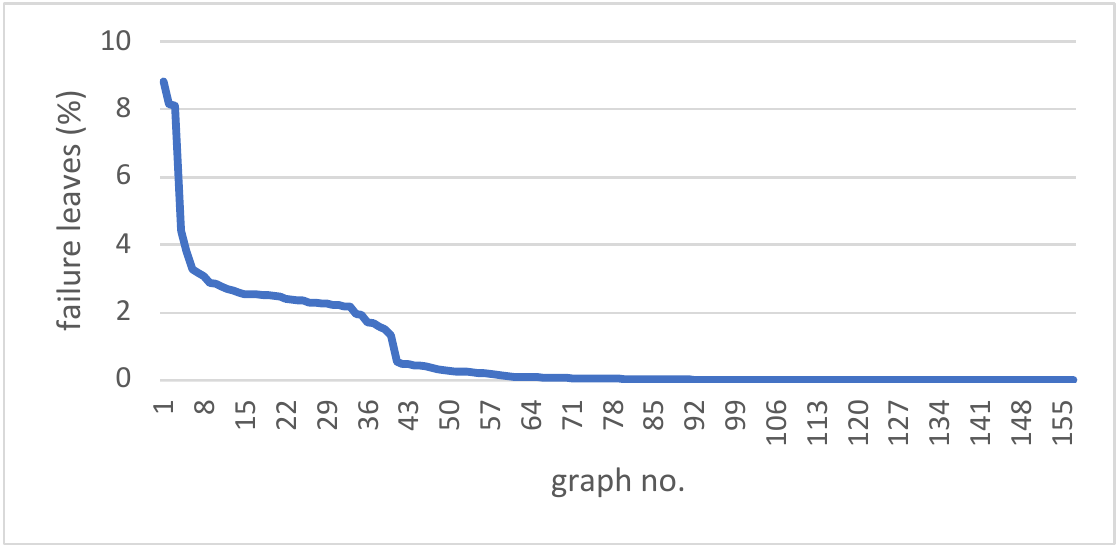}
         \caption{$k=4$}
         \label{fig:fail_leaves_4}
     \end{subfigure}
     \begin{subfigure}[b]{0.49\textwidth}
         \centering
         \includegraphics*[width=\textwidth]{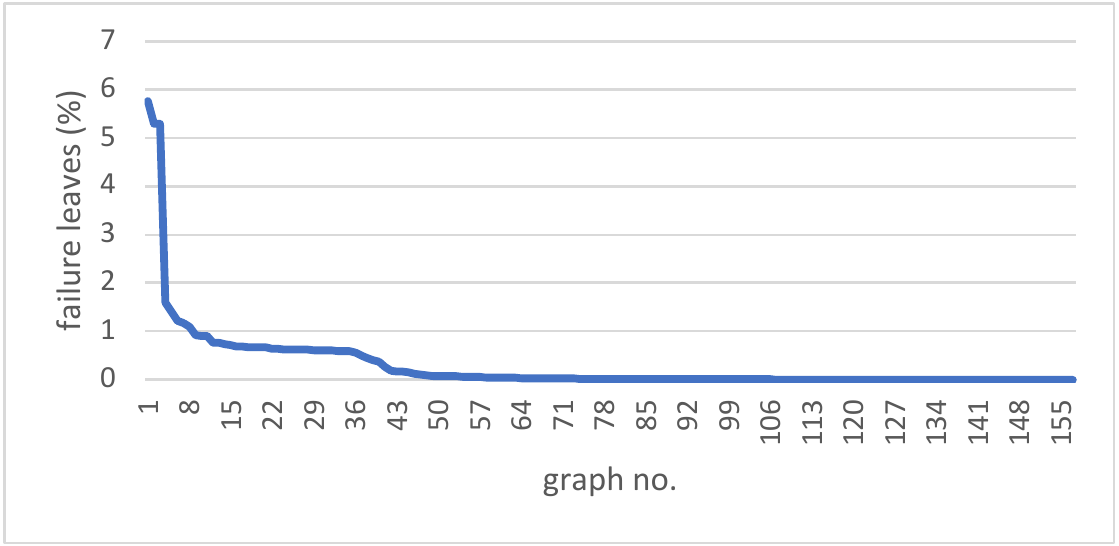}
         \caption{$k = 5$}
         \label{fig:fail_leaves_5}
     \end{subfigure}
    \\
     \vspace*{1mm}
     \begin{subfigure}[b]{0.49\textwidth}
         \centering
         \includegraphics*[width=\textwidth]{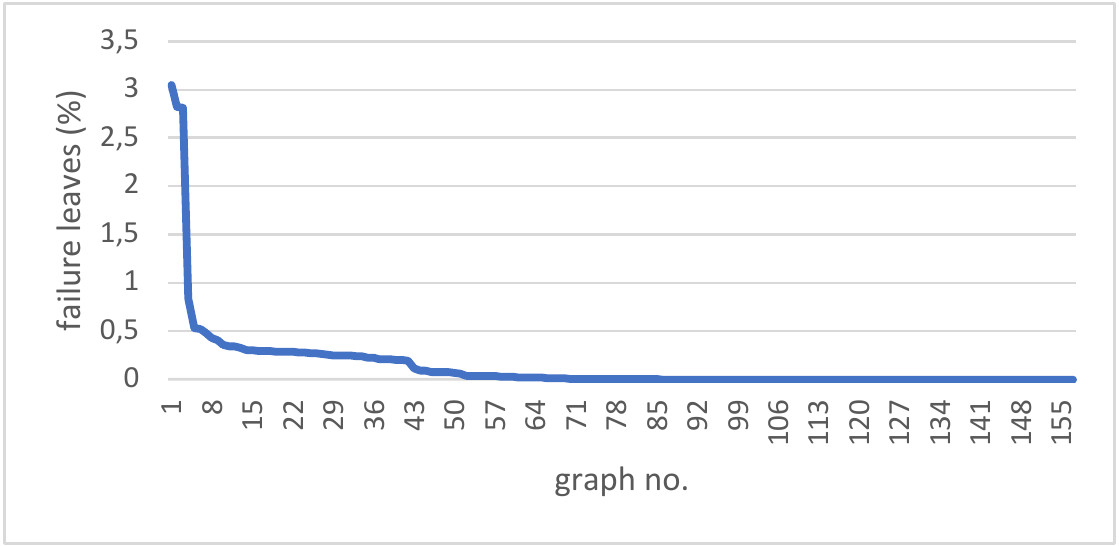}
         \caption{$k = 7$}
         \label{fig:fail_leaves_7}
     \end{subfigure}
     \begin{subfigure}[b]{0.49\textwidth}
         \centering
         \includegraphics*[width=\textwidth]{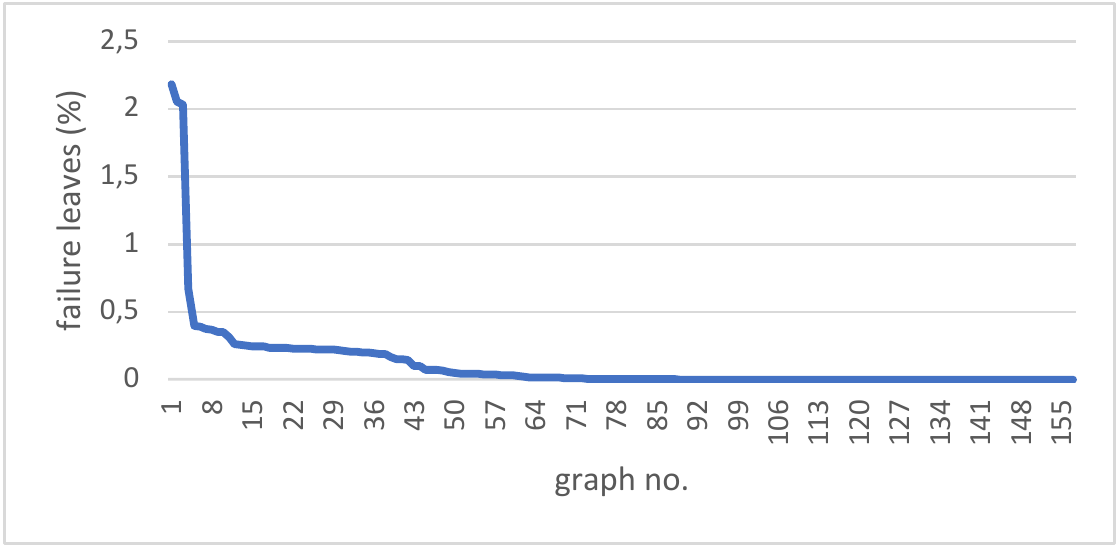}
         \caption{$k=9$} 
         \label{fig:fail_leaves_9}
     \end{subfigure}

        \caption{Percentage of failure leaves over total leaves in the recursion tree, for our whole dataset of 155 graphs and $k = 4, 5, 7 , 9$.}
        \label{fig:chart_fail_leaves}
\end{figure}%

\section{The CAGE Algorithm}
\label{sec:cage}
In this section we present our main contribution, the CAGE algorithm.
We give the pseudocode and discuss its cache-awareness, carefully considering all the scenarios that the algorithm may encounter during its execution, and show how we can effectively compress the last three levels of its recursion tree.

\subsection{Addressing a hard combinatorial limit}
\label{sec:leaves}
Before describing our algorithm, we introduce the principles behind its design. 

First, we opted for a data-driven design after the empirical analysis conducted in in the previous Section.

Second, we wanted to better exploit the cache in the CPU, the main motivation for this work. For example, if we need to check whether an edge $(x,y)$ exists and the adjacency list $N(x)$ is already in cache, whereas we have no such a guarantee for $N(y)$, it is clearly more efficient to test membership $y \in N(x)$ rather than $x \in N(y)$. Therefore, we carefully orchestrated the loading of graph data during execution (see comments in the pseudocode).\footnote{We are not directly controlling the cache, but rather allowing the algorithm to run cache-friendly by making standard assumptions on the associative cache~\cite{FrigoLPR12}.}

Third, once some data is in the cache in the CPU, we would like to exploit it at its best before it is released to load further data. As $|S| < k$ and $k$ is small, it is reasonable to assume that both $S$ and $N(S)$ fit into the cache most of the times, as $|N(S)| < k \Delta$ (exceptions may occur for large $\Delta$). Based on this, we stop the recursion when $|S|=k-3$, and deduce all the possible extensions of $S$ into graphlets by adding three further nodes. This greatly truncates the recursion tree, as each level of calls can be $\Delta$ times more populated than the previous one (potentially reducing the total number of calls by a $O(\Delta^3)$ factor).\footnote{This idea can be generalized to $k-4$, $k-5$, and so on, but there is no payoff going further than $k-3$ in practice as there are too many cases to handle.} 
The last two principles are intimately connected to each other, and give a speedup of at least one order of magnitude in our experiments in Section~\ref{sec:experiments}. 
\begin{figure}[!ht]
  \makebox[\linewidth]{%
  \begin{minipage}{\dimexpr\linewidth+7em}
\begin{algorithm}[H]
\DontPrintSemicolon
  \SetKwProg{myproc}{Function}{}{}
  \SetKwFunction{enum}{ENUM}
  \SetKw{true}{true}
  \SetKw{false}{false}
  \SetKwInput{kwPre}{Hp.}

  \myproc{\enum{$G$,$S$,$k$,$X$}\tcp*[h]{Hp. $k > 3$ (case $k=3$ is simpler)}}{ 
  \tcp{find all graphlets in $G \setminus X$ containing all nodes in $S$}
    \lIf{$|N(S) \setminus X| = 0$}{
      \Return \false\tcp*[h]{failure leaf}    \label{cage:a}
    } 
    \If(\tcp*[h]{exploit the cache to complete $S$}){$|S| = k-3$}{ \label{cage:b}
    \textsl{localcount} $\gets \binom{|N(S) \setminus X|}{3}$ \tcp*[h]{(1) zero if $|N(S) \setminus X| < 3$}  \; \label{cage:c}  
    \ForAll(\tcp*[h]{$N(S)$ loaded in cache}){$u \in N(S) \setminus X$}{ \label{cage:d}
    \textsl{udeg} $\gets$ \textsl{duplicated} $\gets 0$\;
    \ForAll(\tcp*[h]{$z \in N^2(S) \setminus X$}){$z \in (N(u) \setminus (N(S) \cup S)) \setminus X$}{ \label{cage:e}
        \tcp{$N(u)$ and $S$ loaded in cache}
        \textsl{udeg} $\gets \textsl{udeg} + 1$   \tcp*[h]{needed for (2)} \; \label{cage:f}
        \ForAll(\tcp*[h]{(4) see Figure \ref{fig:cases}}){$w \in (N(z) \setminus (N(S) \cup S)) \setminus X$}{ \label{cage:g}
            \tcp{$N(S), N(u)$, $S$ cached, and $N(z)$ loaded} 
            \lIf{$w \not \in N(u)$}{
            \textsl{localcount} $\gets \textsl{localcount} + 1$ \label{cage:h}    
            } 
        } 
        \ForAll(\tcp*[h]{(3) $N(S), N(z)$ cached}){$v \in (N(S) \setminus \{u,z\}) \setminus X$}{ \label{cage:i}
            \eIf(\tcp*[h]{(3a) counted twice, from $v$ and $z$}){$v \in N(z)$ \label{cage:l}}{\textsl{duplicated} $\gets \textsl{duplicated} + 1$  \label{cage:l1}}(\tcp*[h]{(3b)}){\textsl{localcount} $\gets \textsl{localcount} + 1$  \label{cage:l2}}
        } 
    } 
    \textsl{localcount} $\gets \textsl{localcount} + \binom{\textsl{udeg}}{2}$ \tcp*[h]{(2) see Section~\ref{sec:remarks_enum}}\; \label{cage:m}
    \textsl{localcount} $\gets \textsl{localcount} + \textsl{duplicated}/2$~\tcp*[h]{(3a) duplicates counted once} \label{cage:n}\;
    } 
    \textsl{solutions} $\gets \textsl{solutions} + \textsl{localcount}$\; \label{cage:o} 
    \Return $(\textsl{localcount} > 0)$\; \label{cage:p}  
    } 

    
    \textsl{found} $\gets$ \false\; \label{cage:q}
    
    \ForAll(\tcp*[h]{here $|S|<k-3$. Vertices sorted by label}){$u \in N(S) \setminus X$}{ \label{cage:qq}
      \leIf{\enum{$G, S \cup \{u\}, k, X$}}{\textsl{found} $\gets$ \true}{\textbf{break}}

      $X\gets X\cup \{u\}$\;
    }
    
    \Return \textsl{found}\; \label{cage:r}
  
  } 

  \caption{CAGE: cache-efficient ENUM() algorithm, \textsl{solutions} is a global variable.}
  \label{alg:cage}
\end{algorithm}
  \end{minipage}}
\end{figure}

\subsection{Pseudocode}

Algorithm~\ref{alg:cage} implements our principles stated above. It is called Cache-Aware Graphlet Enumeration (CAGE), and is a non-trivial and cache-aware extension of KS-Simple \cite{SIMPLE_graphlets}. 

According to the first principle, we replaced the ENUM() function in KS-Simple with the one in Algorithm~\ref{alg:cage}. We decided to keep the for loop of KS-Simple at lines~\ref{cage:q}--\ref{cage:r}, and fixed the scanning order of the vertices top-level loop for efficient memory usage:
as we only look for the graphlets in which $v$ is the smallest node of the order, we know all vertices smaller than $v$ are in $X$, so we do not need to store them explicitly.
The second and third principles are implemented at lines~\ref{cage:a}--\ref{cage:p}, as the base case. We observe that the base case is no more $|S|=k$ as in KS-Simple, but $|S|=k-3$ as shown in Section~\ref{sub:exploiting_cache}. This is the outcome of several iterative improvements over KS-Simple.

\subsection{Exploiting the cache}
\label{sub:exploiting_cache}

Consider $S$ of size $|S|=k-3$, and let $v$ be the smallest node in $S$ according to an arbitrary ordering (we use label ordering) of the vertices of $G$. As we have to extend $S$ in all possible ways with 3 nodes, we pick them from $N(S)$, $N^2(S)$, and $N^3(S)$, ignoring the nodes belonging to $X$ (or only reachable from $S$ via $X$). Specifically, $X$ is composed of the nodes $v' \in V$ such that $v' < v$, and of the nodes that have already been examined in the loop at lines \ref{cage:q}--\ref{cage:r} in the parent call. We only need to store explicitly the latter nodes.

When picking the 3 nodes, we have to deal with 4 different cases. They are discussed below with an accompanying drawing shown in Figure \ref{fig:cases}.

\begin{enumerate}
\item[\emph{Case 1}:] The 3 nodes come from $N(S) \setminus X$ (see line \ref{cage:c}).

\item[\emph{Case 2}:] A node $u$ is from $N(S) \setminus X$, and the other 2 nodes are chosen from $N^2(S) \setminus X$ so that they are $u$'s neighbors (see lines~\ref{cage:f} and~\ref{cage:m}).

\item[\emph{Case 3}:] Distinct nodes $u,v$ are from $N(S) \setminus X$, and the remaining node $z$ is chosen from $N^2(S)  \setminus X$, so that
    \begin{enumerate}
        \item[(\emph{3a}).] $z$ is a common neighbor of both $u$ and $v$ (see lines~\ref{cage:l1} and~\ref{cage:n}), or
        \item[(\emph{3b}).] $z$ is a neighbor of only $u$ (symmetrically, only $v$) (see line~\ref{cage:l2}).
    \end{enumerate}
    

\item[\emph{Case 4}:] A node $u$ is from $N(S) \setminus X$, a node $z$ from $(N^2(S) \setminus X)\cap N(u)$, and the remaining node $w$ from $N(z) \setminus (N(S) \cup S \cup N(u)) \setminus X$, which comes either from $N^3(S)$ or $N^2(S)\setminus N(u)$ (see line~\ref{cage:h}). 
\end{enumerate}
As a result, CAGE is actually pruning the leaves, their parents, and their grandparents in the recursion tree, each time reducing the total number of calls roughly by a factor of $\Delta$.
\begin{figure}[ht]
        \centering 
        \captionsetup{type=figure}
        \includegraphics[width=.6\linewidth]{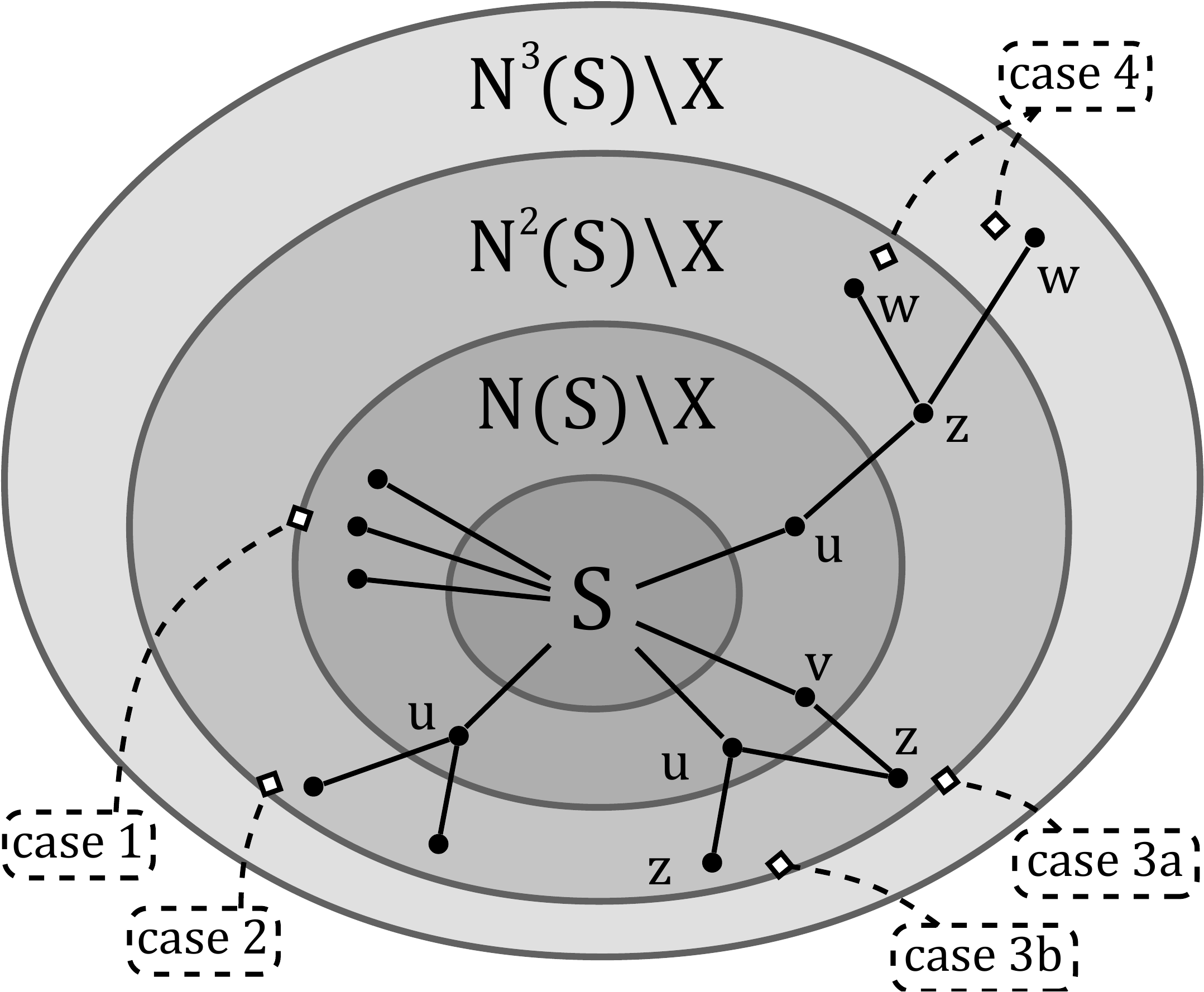}
        \caption{The 4 possible ways of completing a $k-3$ graphlet under construction (contained in set $S$). Case 1: pick 3 vertices at distance 1 from $S$; case 2: pick 1 neighbor of $S$ and 2 of its neighbors; case 3: pick 2 neighbors of $S$ and then 1 of their neighbors, if $z$ is adjacent to both $u$ and $v$ (case 3a) then divide the count by 2; case 4: pick 1 neighbor of $S$, then 1 of its neighbors at distance 2 from $S$ and then one of its neighbors at distance 3 from $S$.}
        \label{fig:cases}
\end{figure}

As each case involves the same or fewer steps than KS-Simple \cite{SIMPLE_graphlets}, CAGE shares its worst-case running time guarantee:
\begin{theorem}
\label{theorem:main}
CAGE reports all $k$-graphlets including the nodes in $S$, but not the nodes in $X$, in $O(k^2 \Delta)$ time per graphlet using $O(\min\{k\Delta,|V|\})$ space.
\end{theorem}
\begin{proof}[Proof (sketch)]
    The proof essentially follows by observing how the three nodes needed to extend $S$ can distribute over $N(S) \setminus X$, $N^2(S) \setminus X$, and $N^3(S) \setminus X$, while preserving connectivity: the only possibilities are 3-0-0, 1-2-0, 2-1-0, 1-1-1, which correspond to cases 1--4, respectively.
    The only concern is case 3a, as shown in Figure \ref{fig:cases}: when $z$ is connected to both $u$ and $v$, we are counting the same graphlet twice, one through $u$ and another through $v$, as we consider unordered pairs of nodes $u$ and $v$ in this case. So we must enumerate this case separately and divide the number of graphlets thus found by 2.
    The time cost per graphlet reported is $O(k^2\Delta)$ since the depth of the recursion is $O(k)$ and each recursive call takes $O(k\Delta)$ time because we are traversing the graph once.
    As for the working space occupation (which does not consider the read-only input graph), the space requirement of each recursive call is dominated by $N(S)$, having worst-case size $O(\min\{k\Delta, |V|\})$ as it is a set of nodes, and at most $\Delta$ nodes are added to them in each of the $k-3$ recursion levels; furthermore, we update the data structures when recurring (and restore them when backtracking) to avoid duplication on the recursion stack, so the working space of the algorithm is indeed $O(\min\{k\Delta, |V|\})$; we remark that this is a worst-case bound, and experimental results show this value to be small in practice (see Section \ref{sub:cache-analysis}).
\end{proof}

In order to make our implementation as cache-friendly as possible, we coupled the pruning of the recursion tree with some cache-friendly data structures so that $S$, $N(S) \setminus X$, and $N(u)$ are stored as vectors, $N(S)$ is stored as a cuckoo hash table \cite{cuckoo_hashing}, and $X$ is not explicitly stored due to vertex ordering. We have:
\begin{lemma}
\label{lemma:cache}
Given an associative CPU cache \cite{associative_cache}, if its size is $\Omega(k \Delta)$ words of memory, then the base case $|S|=k-3$ in CAGE loads $S$, $N(S)$, $N(u)$, and $N(z)$ once for distinct nodes $u$ and $z$, and then keeps them in the cache when they are accessed.
\end{lemma}
Comments in the pseudocode at lines \ref{cage:d}--\ref{cage:n} give indications when $S$, $N(S)$, $N(u)$, and $N(z)$ are loaded into cache, or already cached. This predicted behavior has been carefully planned during our experimental study.
Moreover, the cache words used in practice are likely to be far less than the worst case $O(k\Delta)$, as real-world networks are sparse, with most vertices having very small degree.

\subsection{Remarks}
\label{sec:remarks_enum}

We did not discuss the case $k=3$ as CAGE is mainly aimed at larger values of $k$, and it can be easily obtained as a simplification of the pseudocode.
Moreover, for simplicity, we only report the number of solutions found, and the graphlets are not output. 
All algorithms in the experiments also do not make explicit output, as is usual when comparing enumeration algorithms~\cite{Ruskey03combinatorialgeneration}.
Furthermore, CAGE can produce a compressed output, efficient to both read and write and a common practice in enumeration (see, e.g.,~\cite{tomita,shioura1997optimal}): it suffices to output the $k-3$ nodes in $S$, as well as the sets from which to choose the completions, e.g., for \textit{case 1}, write the contents of $S$, followed by ``plus any 3 elements from'', followed by the contents of $N(S)\setminus X$. When dealing with case 3a, we can decide to output only the choice $u < v$, to avoid the duplication when $z$ is connected to both $u$ and $v$.

\section{Experimental Results}
\label{sec:experiments}
This section provides the details and results of our extensive experimental phase.

\subsection{Environment and dataset}

Our experiments were carried on a dual-processor Intel Xeon Gold 5318Y Icelake @ 2.10GHz machine, with 48 physical cores each and 1TB of shared RAM; private cache L1 per core: 48K, private cache L2 per core: 1.25MB, shared cache L3: 36MB,
running Ubuntu Server 22.04 LTS, Intel C++ compiler \texttt{icpx}, version 2022.1.0, and cache analysis was done with Intel VTune profiler, version 2022.2.0.

The dataset consists of 155 graphs taken from public repositories LASAGNE \cite{lasagne},  Network Repository \cite{networkrepository}, and SNAP \cite{snapnets}.
These files range in size from very small (hundreds of nodes and edges) to very large (millions of nodes and half a billion edges). 
Table \ref{tab:dataset} shows a relevant subset of the dataset, and the 10 graphs mentioned in this section can be downloaded from GitHub\footnote{\scriptsize\label{footnote:mega_url}\url{https://github.com/DavideR95/CAGE}}, as well as the source code of our algorithm and the re-implemented KS-Simple.
\begin{table}[htb]
\caption{A sample from our dataset, sorted by $|E|$}
\label{tab:dataset}
\footnotesize
\centering
\small
\begin{tabular}{|l|l|r|r|r|} \hline
\multicolumn{1}{|c|}{Graph} & \multicolumn{1}{c|}{Type} & \multicolumn{1}{c|}{$|V|$} & \multicolumn{1}{c|}{$|E|$} & \multicolumn{1}{c|}{$\Delta$} \\ \hline
Brady       & Biological   & 1,117      & 1,330       & 28      \\ 
ca-GrQc & Collaboration & 5,242 & 14,484 & 81 \\ 
cti & Mesh & 16,840 & 48,232 & 6 \\ 
Wing & Mesh & 62,032 & 121,544 & 4 \\ 
Roadnet-TX  & Road Network & 1,379,917  & 1,921,660   & 12      \\ 
Roadnet-CA & Road Network & 1,965,206 & 2,766,607 & 12 \\ 
Auto & Mesh & 448,695 & 3,314,611 & 37 \\ 
Hugetrace-00 & DIMACS10    & 4,588,486  & 6,879,133   & 4 \\ 
IMDB        & Movies       & 913,201    & 37,588,613  & 11,941  \\ 
Arabic-2005 & Web Crawl    & 22,744,080 & 553,903,073 & 575,628 \\ \hline 
\end{tabular}%
\end{table} 
\subsection{Implementation and comparison methodology}
\label{sub:comparison}

We implemented CAGE in C++, compiled with the highest optimization flag \texttt{-O3}. The cache-friendly nature discussed in Algorithm~\ref{alg:cage} is enhanced by the \texttt{likely} and \texttt{unlikely} macros to help the compiler with branch prediction, and we do not physically delete and restore nodes between recursive calls by reusing the data structures mentioned in Section~\ref{sec:cage}, e.g., traversing the vector for $N(S) \setminus X$ so that deletion of nodes is modeled by a flexible starting index.

Below we adopt the notation CAGE-1, CAGE-2, and CAGE-3, with the intent of distinguishng among the base cases $|S| = k-1$, $|S| = k-2$, and $|S| = k -3$, respectively. Indeed, CAGE-3 is actually CAGE, and the others can be easily derived by simplifying the pseudocode in Algorithm~\ref{alg:cage}.


We compared our algorithms against several different algorithms from the literature, including recent and well-known approaches from \cite[Table~2]{survey_graphlet}, i.e.,
PGD \cite{PGD}, Kavosh \cite{Kavosh}, FaSE \cite{FASE} and Escape \cite{escape_algorithm} (although we exclude the matrix-based approach \cite{melckenbeeck2019optimising}, unsuitable for large graphs), as well as the more recent KS-Simple \cite{SIMPLE_graphlets}.
The aforementioned methods adopt different strategies to exactly count the number of $k$-graphlets, but they all share an enumerative core, either via implicit counting (i.e., PGD and Escape) or explicit construction of the graphlets (the others).
Since the available code of KS-Simple \cite{SIMPLE_graphlets} is in Python, we re-implemented it in C++ to avoid penalizing it, following the same guidelines as for CAGE.
We then conducted the following experiments:


\begin{enumerate}[(i)]\itemsep0em
    \item Analysis of cache efficiency and hotspots in the code using VTune, setting a timeout of 15 minutes on a subset of the dataset built to represent differently sized graphs,
    \item Analysis of the recursion tree of KS-Simple for $k = 4, 5, 7, 9$ on the entire dataset with a time limit of 30 minutes, showing that the number of failure leaves is typically less than 1\% of the total number of leaves, and never more than 9\% (ref. Section~\ref{cage:sec:failure_leaves},
    \item A stress test for all the algorithms and competitors, for $k \in [4, 10]$ with a timeout of 12 hours.
\end{enumerate}
\begin{table}[htb]
\caption{VTune Profiler cache access statistics for our algorithms and competitors with $k=7$. $\dagger$: execution stopped after 15 minutes. $*$: execution stopped due to a memory allocation problem.}
\label{tab:vtune_misses}
\centering
\resizebox{\textwidth}{!}{
\begin{tabular}{|c|l|r|r|r|r|r|r|r|r|}
\hline
Graph &
  \multicolumn{1}{c|}{Algorithm} &
  \multicolumn{1}{c|}{\makecell{Time\\(s)}} &
  \multicolumn{1}{c|}{\begin{tabular}[c]{@{}c@{}}\#Graphlets Found\\ ($k=7$)\end{tabular}} &
  \multicolumn{1}{c|}{L3 Misses} &
  \multicolumn{1}{c|}{\makecell{L1\\Bound}} &
  \multicolumn{1}{c|}{\makecell{L2\\Bound}} &
  \multicolumn{1}{c|}{\makecell{L3\\Bound}} &
  \multicolumn{1}{c|}{Loads} &
  \multicolumn{1}{c|}{Stores} \\
\hline
\multirow{6}{*}{ca-GrQc} &
  KS-Simple $^\dagger$ &
   $\dagger$ & 
  8,577,821,416 &
  0 &
  7.9\% &
  0.7\% &
  0\% &
  3,531\,E+9 &
  1,394\,E+9 \\
  & Kavosh & $\dagger$ & 884,849,128 & 65,553,660 & 8.2\% & 0.3\% & 0.1\% & 3,075\,E+9 & 1,782\,E+9 \\
  & FaSE & $\dagger$ & 2,448,373,561 & 5,429,820 & 10.2\% & 0.7\% & 0\% & 1,912\,E+9 & 276\,E+9 \\
 & CAGE-1 & 76 & \textbf{15,186,322,814} & 0 & 8.8\%  & 0.6\% & 0\%    & 303\,E+9 & 117\,E+9 \\
 & CAGE-2 & 39 & \textbf{15,186,322,814} & 0 & 11.3\%  & 1.5\% & 0\%    & 116\,E+9 & 7\,E+9   \\
 & CAGE-3 & \textbf{34} & \textbf{15,186,322,814} & 0 & 15.4\% & 2.8\% & 0\%    & 95\,E+9 & 2\,E+9  \\
\hline
\multirow{6}{*}{roadnet-TX} &
  KS-Simple &
   14 & 
  \textbf{203,059,778} &
  0 &
  14.9\% &
  0\% &
  0\% &
  43\,E+9 &
  23\,E+9 \\
 & Kavosh & $*$ & -- & -- & -- & --   & -- & -- & -- \\
 & FaSE & 46 & \textbf{203,059,778} & 0 & 14.2\% & 0.5\%   & 0\% & 128\,E+9 & 37\,E+9 \\
 & CAGE-1 & 5 & \textbf{203,059,778} & 0 & 16.2\% & 0\%   & 0.6\% & 14\,E+9 & 8\,E+9 \\
 & CAGE-2 & 4 & \textbf{203,059,778} & 0 & 22.3\%   & 0\% & 0\% & 6\,E+9 & 2\,E+9 \\
 & CAGE-3 & \textbf{3} & \textbf{203,059,778} & 0 & 28.2\%   & 0\%   & 1.3\%   & 5\,E+9 & 1\,E+9 \\
 \hline
\multirow{6}{*}{auto} & KS-Simple & $\dagger$ & 4,775,173,331 & 0    & 14.2\% & 0.3\% & 0\% & 2,360\,E+9 & 1,051\,E+9 \\
& Kavosh & $*$ & -- & -- & -- & -- & -- & -- & -- \\
& FaSE & $\dagger$ & 3,032,335,810 & 27,202,000 & 12.5\% & 0.7\% & 0.1\% & 2,027\,E+9 & 365\,E+9 \\
                             & CAGE-1   & $\dagger$ & 80,580,776,005     & 0    & 14.2\% & 0.3\% & 0\% & 2,176\,E+9 & 1,014\,E+9    \\
                            & CAGE-2    & $\dagger$ & 135,096,265,408    & 0   & 21.9\% & 0.5\% & 0.1\% & 1,480\,E+9 & 138\,E+9     \\
                             & CAGE-3    & $\dagger$ & \textbf{152,621,021,219}    & 0 & 21.1\% & 1.4\% & 0.1\% & 2,080\,E+9 & 53\,E+9    \\
\hline
\multirow{6}{*}{arabic-2005} & KS-Simple & $\dagger$ & 4,001,309,731      & 27,157,305    & 6\% & 0.7\% & 5.1\% & 3,626\,E+9 & 593\,E+9   \\
& Kavosh & $*$ & -- & -- & -- & --   & -- & -- & -- \\
& FaSE & $*$ & -- & -- & -- & --   & -- & -- & -- \\
                             & CAGE-1    & $\dagger$ & 22,084,290,111,889 & 48,516,858 & 5.9\% & 0.8\%   & 4.9\% & 2,976\,E+9 & 552\,E+9 \\
                             & CAGE-2    & $\dagger$ & 27,208,214,120,342 & 21,697,928    & 7.8\% & 1\% & 4.2\% & 1,820\,E+9 & 11\,E+9     \\
                             & CAGE-3    & $\dagger$ & \textbf{47,718,156,097,277}
                             & 173,529,344   & 1\% & 0.3\% & 0.2\% & 3,329\,E+9 & 23\,E+9    \\
\hline
\end{tabular}%
}

\end{table}

\subsection{Details on Cache-Aware data structures}
We provide here some details of the cache-friendly data structures that we used in our implementation of CAGE.

\begin{enumerate}
    \item Graphs are represented using adjacency lists that in turn are implemented as cuckoo hash tables \cite{cuckoo_hashing} with the possibility of being scanned sequentially, providing fast neighboring checks and cache friendly scans.
    \item $S$ is stored as vector of size $\leq k-3$, employed in lines~\ref{cage:b}, \ref{cage:e}, \ref{cage:g} for cache-friendly sequential scan. 
    \item \label{item:N_S_X} $N(S) \setminus X$ is stored as vector of size $\leq (k-3) \Delta$,  employed in lines~\ref{cage:a}, \ref{cage:c}, \ref{cage:d}, \ref{cage:qq} for cache-friendly sequential scan.
    \item $N(S)$ is stored as a cuckoo hash table with $\leq (k-3) \Delta$ keys, employed in lines~\ref{cage:e}, \ref{cage:g}, \ref{cage:i} for random access or cache-friendly sequential scan.
    \item \label{item:adj} $N(u)$, for each $u \in V$, is stored as a cuckoo hash table with $\deg(u)$ keys, employed in lines~\ref{cage:e}, \ref{cage:g}, \ref{cage:h}, \ref{cage:l} for random access or cache-friendly sequential scan. 
    \item $X$ is not explicitly stored: it is implicitly represented by the nodes $v' \in V$ such that $v' < v$ in the vertex ordering and by the structure described at point \ref{item:N_S_X}. 
\end{enumerate}

The total space (on top of the input graph) is $O(k \Delta)$. As $S$, $N(S)$, $N(u)$, and $N(z)$ are required simultaneously in the cache during the for loop execution at lines~\ref{cage:d}--\ref{cage:n}.

\subsection{Cache analysis}
\label{sub:cache-analysis}

During our implementation of CAGE, we took our design decisions based upon the performance metrics given by Intel VTune Profiler, tweaking the code according to the statistics of cache accesses and misses. 
The results of the VTune analysis for the optimized versions of KS-Simple, CAGE-1, CAGE-2, and CAGE-3 are summarized in Table \ref{tab:vtune_misses}, along with the same data for our competitors Kavosh \cite{Kavosh} and FaSE \cite{FASE}. We chose not to include PGD \cite{PGD} and ESCAPE \cite{escape_algorithm} in this analysis since they only work for $k \leq 5$, while our focus is on larger values of $k$, in order to show the performance scaling to larger working sets (i.e. larger $S$ and $N(S)$). 
The macro rows correspond to four graphs, whose names are given in the first column. We report the cache performance by taking the best out of 5 non-consecutive executions of VTune, where the computation was stopped after 15 minutes for two large graphs\footnote{This timeout is due to the size of the data gathered by VTune, which grows quickly over time.}.
The other columns report: the run time in seconds, number of graphlets found for $k=7$, number of L3 cache misses (i.e. accesses to the RAM), how much L1, L2, L3 cache affected on the percentage of clock ticks where the CPU was stalled waiting on that level of cache (the lower, the better), and the total number of load and store operations.

From the results on our algorithms, it clearly emerges that L2 and L3 bounds are very small, whereas L1 is larger for small graphs. This is a sign of good cache usage, as the L2 and L3 cache misses are definitely more expensive than the L1 cache misses (both data and instructions). For the small graphs \emph{ca-GrQc} and \emph{roadnet-TX} and mid-size graph \emph{auto}, we have zero L3 misses (as they probably fit in the L3); for large graph \emph{arabic-2005}, the number of L3 cache misses increases going through the rows for KS-Simple, CAGE-1, CAGE-2, and CAGE-3. This may not be obvious as there is a timeout of 15 minutes for this large graph: consequently, KS-Simple produces very few solutions compared to the others, and CAGE-3 more solutions (the apparent anomaly for CAGE-3 on \emph{arabic-2005} can be explained by dividing the numbers by 2 in its last row in the table, so that we roughly get the same number of solutions as CAGE-1 and CAGE-2, we have also similar numbers in the other columns, except for L3 misses, which we discuss in a while). The number of loads/stores follows a similar pattern to that for L3 misses. The number of L3 cache misses is negligible with respect to the total number of load and store instructions issued, allowing the percentage of L2- and L3-bound to stay always within 5\%.
On the other hand, our competitors are able to achieve similar results in terms of L1, L2, and L3 bound, but the number of L3 misses is higher even on \emph{ca-GrQc}. FaSE is able to compute all the solutions with zero cache misses on \emph{roadnet-TX}, but with a much higher time requirement compared to CAGE.
Additionally, the Kavosh algorithm started having bad memory allocation\footnote{i.e. the runtime raised a \texttt{bad\_alloc} exception or a segmentation fault while reading the input graph.} issues already with \emph{roadNet-TX}, while FaSE fails later on the largest graph.

We also evaluated the size of the data structures 
for $N(S)$, $N(u)$, and $N(z)$, examined during the for loop of Algorithm~\ref{alg:cage} along with the average degree of the networks, as shown in Table~\ref{tab:nsnunz}, for $k=7$. They clearly fit into the cache most of the times, recalling that each node identifier is a 32-bit integer and that the L2 cache in each core is 1.25MB on our machine. For \emph{arabic-2005}, even if its maximum degree $\Delta$ is large, its average degree is small, as well as the median values for $N(\cdot)$. Lemma~\ref{lemma:cache} indicates that CAGE-3 is cache-friendly under this condition on the average.

\begin{table}[htb]
    \caption{Median size of $N(S)$, $N(u)$ and $N(z)$ for a subset of the dataset ($k=7$).
    }
    \label{tab:nsnunz}
    \centering
    \begin{tabular}{|c|r|r|r|r|r|}
    \hline
        Graph & $\Delta$ & avg. degree & $N(S)$ & $N(u)$ & $N(z)$ \\
    \hline
        ca-GrQc & 81 & 5.5 & 52 & 13 & 9 \\
        roadnet-TX & 12 & 2.8 & 3 & 3 & 8 \\
        auto & 37 & 14.7 & 35 & 15 & 15 \\
        arabic-2005 & 575,628 & 48.7 & 115 & 36 & 49 \\
    \hline
    \end{tabular}
%
\end{table}
According to these results, we believe that CAGE-3 is preferable to CAGE-1 and CAGE-2 as it finds more solutions for the given timeout, even though in some cases CAGE-2 could perform as well as CAGE-3.
%
%
Finally, we remark that CAGE-4 (i.e., adopting $|S|=k-4$ as base case) might not be faster than CAGE-3, as a recursive call needs to address 8 scenarios individually and explore $N^4(S)$ instead of $N^3(S)$: this increases code complexity, requires extensive fine-tuning and could reduce cache friendliness (adding nested for-loops and more load operations); as such we leave this study for future work.

\begin{figure}[htb]
    \centering
    \begin{subfigure}[b]{.49\textwidth}
         \centering
         \includegraphics*[width=\textwidth]{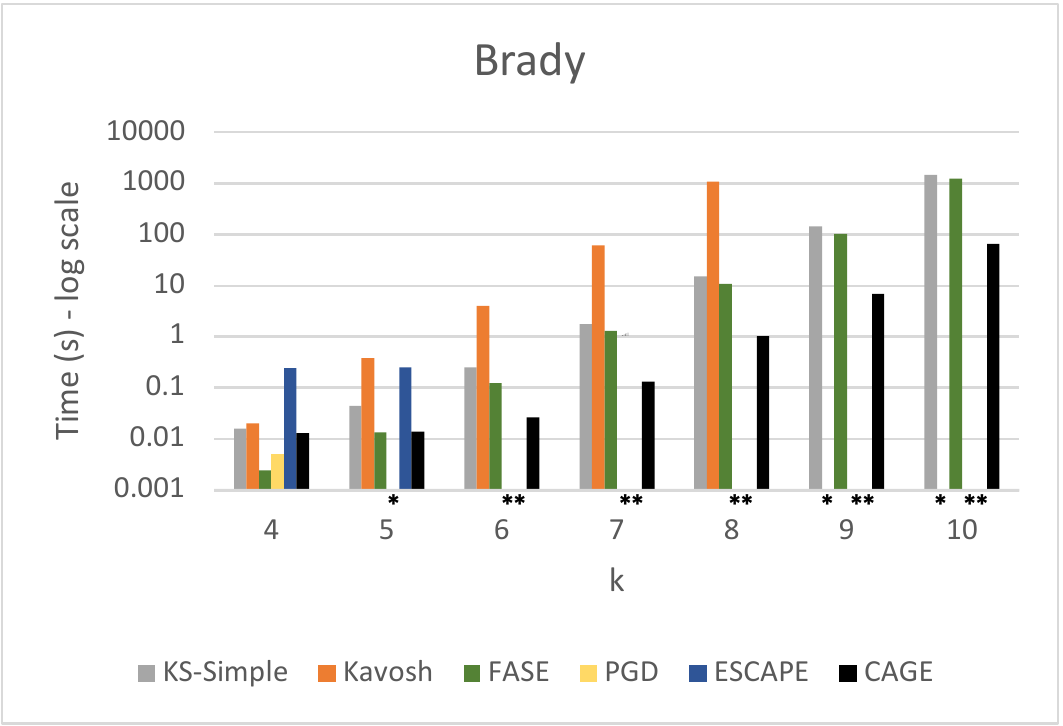}
         \caption{}
         \label{fig:cage_time_brady}
     \end{subfigure}
     \begin{subfigure}[b]{.49\textwidth}
         \centering
         \includegraphics*[width=.99\textwidth]{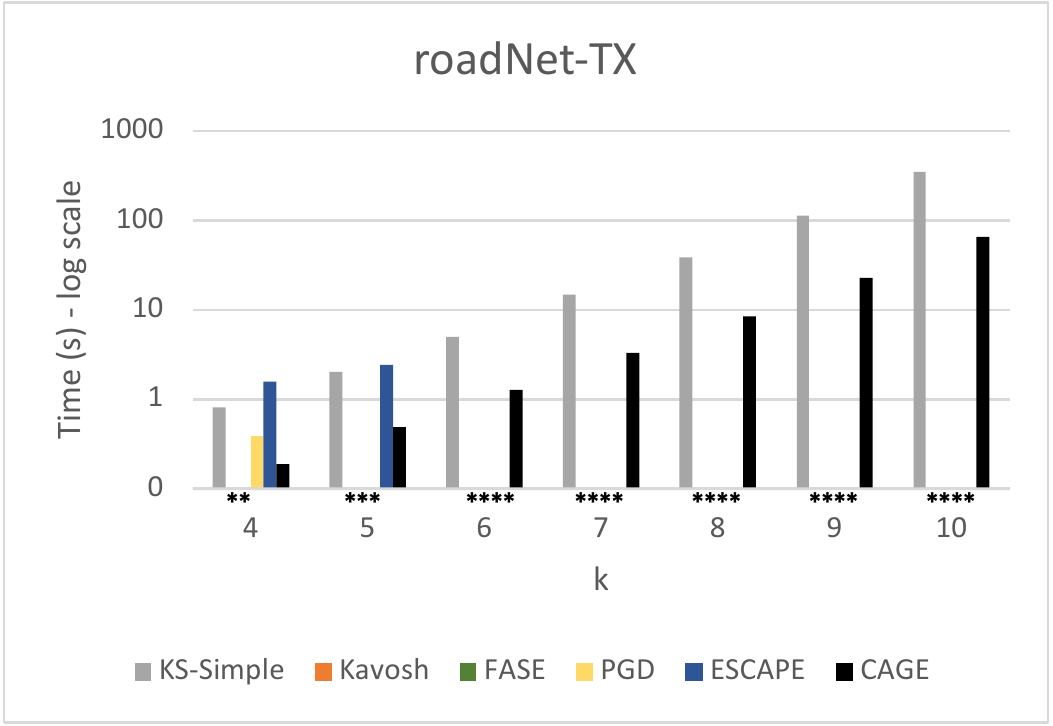}
         \caption{}
         \label{fig:cage_time_roadnet}
     \end{subfigure}\\
     \vspace*{1mm}
     \begin{subfigure}[b]{.49\textwidth}
         \centering
         \includegraphics*[width=\textwidth]{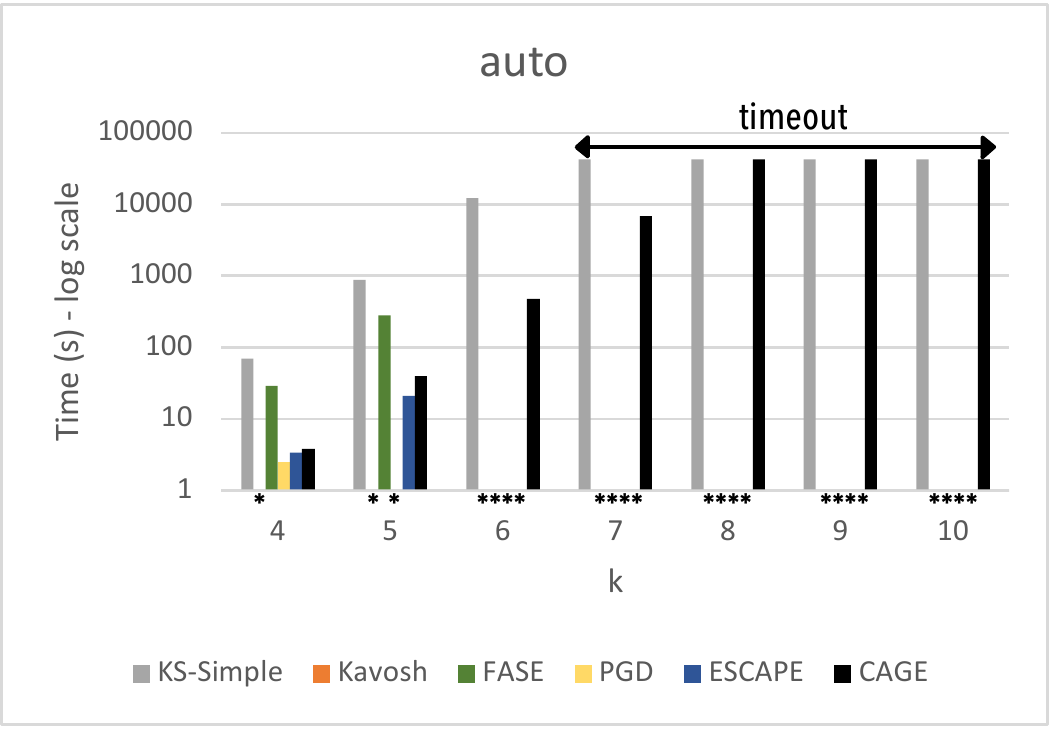}
         \caption{timeout: 12h}
         \label{fig:cage_time_auto}
     \end{subfigure}
     \begin{subfigure}[b]{.49\textwidth}
         \centering
         \includegraphics*[width=\textwidth]{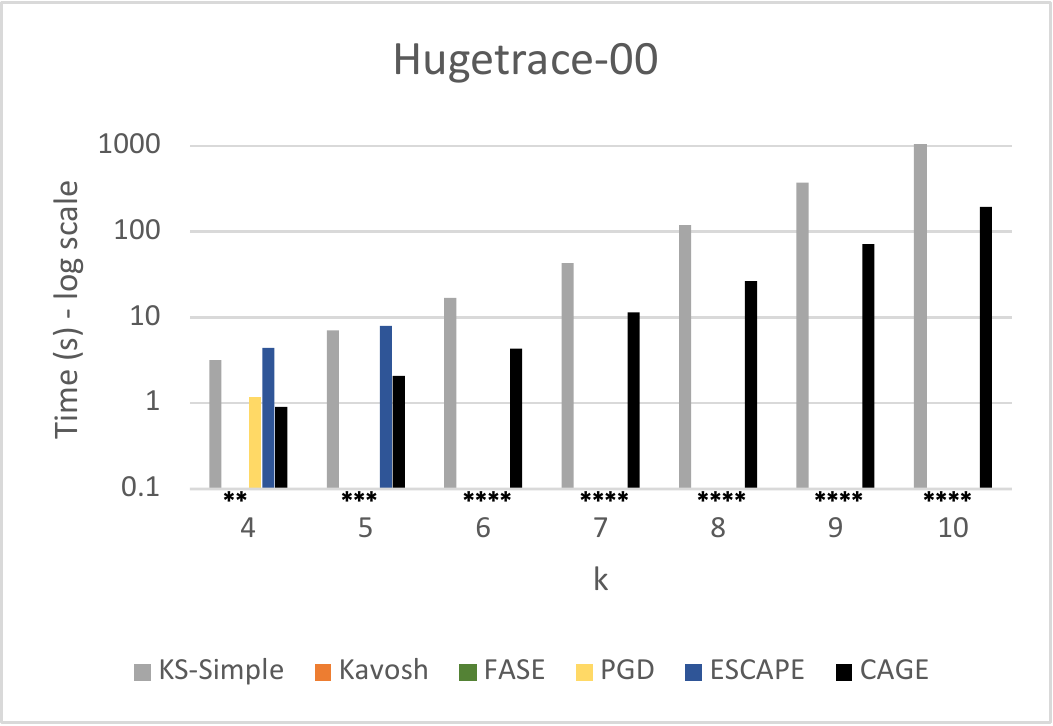}
         \caption{} 
         \label{fig:cage_time_hugetrace}
     \end{subfigure}%
        
    \caption{Running time of our competitors vs. CAGE on four significant input graphs. *: an execution issue occurred.}
    \label{fig:cage_time}
\end{figure}%

\subsection{Running time analysis}
\label{sub:time-analysis}

Armed with our fine-tuned implementation of CAGE (i.e.~CAGE-3), we compared it to the other methods mentioned in Section \ref{sub:comparison}. We performed an extensive validation phase, measuring both the running time and the number of graphlets found, with a timeout of 12 hours. Some methods could not terminate the execution for some issues, e.g. they could not handle a graph so large, ran into memory issues, or were killed by the operating system. 

Figure \ref{fig:cage_time} shows a summary of the results obtained in four graphs of increasing size, Brady, roadNet-TX, auto, and Hugetrace-0000 (see Table \ref{tab:dataset} for their characteristics). On the y-axis, we reported the running time in seconds, in logarithmic scale; on the x-axis, we report increasing values of $k \in [4,10]$.
For each value of $k$, each bar corresponds to an algorithm as specified in the legend, where lower bars mean better performances. The executions that had the issues mentioned above, have an asterisk `*' in place of the corresponding bar. For graph auto, the executions for $k = 7, 8,9,10$ were in timeout (indicated with a bar).

We can observe that KS-Simple and CAGE perform quite well, with CAGE running faster than KS-Simple by roughly an order of magnitude (for $k=7$ on auto, KS-Simple was in timeout whereas CAGE was the only one to terminate its execution). Both compare significantly better with the other methods, as the latter ones are either too slow or cannot execute that instance of the graph. An exception is the PGD algorithm~\cite{PGD}, as it achieves faster runnning time than CAGE on Brady and auto, but it does not scale well for the other two graphs; moreover, PGD is explicitly designed to work only with $k = 3, 4$. These results, combined with the code profiling statistics of Section~\ref{sub:cache-analysis}, confirm the benefits of our design principles for CAGE.

  \chapter{\texorpdfstring{$k$}{k}-cores for Temporal Graphs}\label{chap:temporal}

This Chapter shifts our focus from static graphs to temporal graphs, as we provide a case study on the usefulness of $k$-cores for community analysis on temporal data.

In fact, in a world where cliques and their relaxations are computationally intensive to list (see Chapter \ref{chap:background}), a holy grail of graph analysis is the $k$-core \cite{montresor2013distributed}: a structure which accurately identifies communities and connectivity in graphs \cite{santofortunato, giatsidis_evaluating, KONG20191}; yet the time to compute all $k$-cores in a network is only linear in its size. 
The $k$-core model has been remarkably successful for a variety of tasks, such as identifying communities~\cite{santofortunato}, as well as important nodes~\cite{faloutsos_corescope}, producing graphical embeddings of large graphs~\cite{alvarez2005large,nguyen2017k}, and even speeding up graph algorithms via parametrization~\cite{conte2020sublinear} or pruning~\cite{conte2017fast}.
However, how do we adapt decades of graph-based concepts and algorithms to \changed{the temporal graph model}?
So far, significant effort has been devoted to find \changed{parallels}: direct adaptations of classical graph concepts that feel natural on temporal graphs, such as cliques~\cite{himmel2016enumerating}, Eulerian walks~\cite{marino2021konigsberg}, $k$-plexes~\cite{bentert2019listing}, 
and indeed $k$-cores~\cite{bonchi2020_span_cores,wu2015core,li2018core_union}.
While these adaptations are elegant and present interesting properties, the focus has always been on what definition \changed{feels more natural, i.e., intuitive and clutterless,} and not on what information we can actually get out of them, and this is the direction we take in this Chapter.

Specifically, we consider the following questions: 
\begin{enumerate}\itemsep0em
    \item Can proposed adaptations of $k$-cores extract meaningful data from temporal graphs? 
    \item \changed{Can we compute such adaptations efficiently?}
    \item How can this information be visualized?
\end{enumerate}
We indeed find positive answers. We firstly unify some of the proposed generalizations of \changed{temporal $k$-core} with a unique general notation, showing they can be computed efficiently: for a temporal graph with $n$ nodes, $m$ edges and $\tau$ temporal \changed{\emph{snapshots}}, we show how to query $k$-cores from any temporal window in $O(m \log{\tau})$ time with simple data structures, rather than the naive $O(m\tau)$ time approach. 

Secondly, we use temporal $k$-cores to analyze several real-world temporal networks, varying the available parameters and considering both \changed{instant} and aggregated data, drawing inspiration from techniques for static graph analysis, such as \cite{faloutsos_corescope}, to produce suitable visualizations. We show how this can help to highlight key information about the roles of nodes in a temporal network, as well as the dynamics of the network as a whole.
\section{Basic Concepts and Models}\label{temporal:sec:prelim}
Our work focuses on undirected temporal graphs, where we assume $n=O(m)$, as nodes never incident to any edge are not of interest.
For convenience, we define $E_{[a,b],h}$ as the set of (static) edges $\{u,v\}$ that appear \textit{at least $h$ times} in the interval $[a,b]$, i.e., $E_{[a,b],h} = \{ \{u,v\} : |\{ t : a\le t \le b \land \{u,v,t\}\in E_\tau\}| \ge h \}$. This will help understand how existing definitions are related.

A $k$-core in a static graph $G = (V, E)$ is defined as an inclusion-maximal set of vertices $C \subseteq V$ in which all the vertices have at least $k$ neighbors in $C$\footnote{Alternatively, $C \subseteq V$ is an inclusion-maximal set of vertices such that all vertices in the induced subgraph $G[C]$ have degree $\geq k$. Inclusion-maximal means there is no $C'\supset C$ with the same properties on $G[C']$.}. 
Since there is only one $k$-core for each $k$ and for each connected component of a graph, $k$-cores can be organized in a chain of inclusions, i.e., $C_n \subseteq \dots \subseteq C_1 \subseteq C_0 = V$ where $C_i$ is the $i$-core of the graph.
Thanks to this structure, we can compute the \emph{core decomposition} of a graph, which is the set $\{C_1, C_2, \dots C_n\}$ and sort the vertices according to the core to which they belong. This, in turn, allows us to compute the \emph{coreness} of each vertex:
\begin{de}[Coreness]\label{temporal:def:coreness}
The \emph{coreness} (or core number) of a vertex $v \in V$ of a static graph $G = (V, E)$ is the largest $k$ for which $v$ belongs to the $k$-core of $G$.
\end{de}
We obtain the coreness of each node directly from the $k$-core decomposition of a graph, with well-known linear-time algorithms~\cite{k_core_survey2020}. 
In static graphs, the coreness captures the community structure of a network \cite{Malvestio2020, santofortunato, BAE2014549}, and is a relevant importance score for its nodes \cite{kumar2020identifying,wu2015core}, alongside $h$-index and degree. 
The importance of coreness is also shown in \cite{faloutsos_corescope} as its variation, compared to degree, accurately highlights outlier nodes.

Throughout this Chapter we will thus rely on coreness; \changed{however, since we are dealing with temporal graphs, we have to adapt the concept of $k$-core to the temporal setting.} It turns out that different generalizations are possible, depending on how temporal information is aggregated; we give an overview of existing models, and show a convenient unified representation that is able to capture each of them by tweaking its parameters.

\subsection{Temporal k-core models}

In recent years different formalizations of the concept of temporal $k$-core have been proposed, and we briefly recap them below.\footnote{For self-containedness, we paraphrase all definitions with a uniform notation.}
The approach adopted by Galimberti et al. \cite{bonchi2020_span_cores} is
the one that looks for $k$-cores in what we call the \emph{intersection} of temporal snapshots, i.e., given a time window, any pair of vertices has to interact in every snapshot covered by the window. 
This leads us to the formal definition of \emph{span cores}. 
\begin{de}[Span Core \cite{bonchi2020_span_cores}]
    Given a temporal graph $G_\tau = (V, E_\tau)$, $k \in \mathbb{N}$, and a window $\Delta = [a, b]$, a $(k, \Delta)$-core or \emph{span core} is a maximal set $\emptyset \neq C_{k,\Delta} \subseteq V$ such that $C_{k,\Delta}$ is a $k$-core of the static graph $G' = (V, E_{[a,b], b-a+1})$.
\end{de}

According to this definition, the authors of \cite{bonchi2020_span_cores} consider valid only $k$-cores that exist in all the snapshots contained in the given interval $\Delta$, in a logical \changed{\emph{and}} fashion, i.e., every edge belonging to a core has to appear in \emph{every} timestamp inside the window (see Figure~\ref{temporal:subfig:span-6} and Figure~\ref{temporal:subfig:span-3}). 
This is especially useful in social networks analysis, where usually the links between any given pair of people last for long periods without interruptions.
Naturally, the opposite \changed{\emph{or}} fashion can also be considered, this time requiring only 1 interaction in any given time window $\Delta$.

In particular, Li et al. \cite{li2018core_union} proposed the following definition:
\begin{de}[$(\theta, k)$-persistent core \cite{li2018core_union}]
Let $G_\tau = (V, E_\tau)$ be a temporal graph, $k, \theta \in \mathbb{N}$, and $[t_s, t_e]$ with $t_e-t_s \geq \theta$ be an interval. A $(\theta, k)$-persistent core in the interval $[t_s, t_e]$ is a set of nodes that is a $k$-core of the static graph $G' = (V, E_{[i, i+\theta],1})$
for all $i\in [t_s, t_e-\theta]$, and the same does not hold for any interval strictly containing $[t_s,t_e]$.
\end{de}

This definition is useful in call-log networks and email exchanges, or contexts that see sporadic interactions between nodes, \changed{as it considers the union of temporal edges within intervals of $\theta$ snapshots.}

A definition capturing the actual notion of \emph{union} is the following:
\begin{de}[$\mathcal{T}^k_{[t_s,t_e]}$\cite{yang2023scalable}]\label{temporal:def:union_core}
    A temporal $k$-core in the interval $[t_s, t_e]$ is a $k$-core of the static graph $G'= (V,E_{[t_s,t_e],1})$.    
\end{de}
While this definition allows control over the temporal interval, it does not consider how many times two nodes interact in a time interval, one interaction is always enough (see Figure~\ref{temporal:subfig:temporal-union}). 

The last definition that we consider here is given by Wu et al. \cite{wu2015core}, which considers a temporal graph as a multigraph (ref. Chapter \ref{chap:background}):
\begin{de}[$(k, h)$-core~\cite{wu2015core}] 
In a temporal graph $G_\tau$, a $(k, h)$-core is a $k$-core of the static graph $G' = (V, E_{[1,\tau],h})$.
\end{de}

The authors of \cite{wu2015core} introduce the parameter $h$, useful for modeling the strength of the connection, however their model does not allow to consider sub-intervals of the timeline or even the sequentiality of temporal information (see Figure~\ref{temporal:subfig:tau-3}).

We believe that all these models are meaningful and well grounded, \changed{with each model considering different aspects of the data and extracting different information.}

For this reason, we formalize a definition that encompasses all of these variables while staying reasonably simple, the $(k, h, \Delta)$-core: 
\begin{de}[$(k, h, \Delta)$-core]\label{temporal:def:temporal_core}
    Given a temporal graph $G_\tau = (V, E_\tau)$ two integers $k, h \in \mathbb{N}$, and a window length $1 \leq \Delta \leq \tau$, a set of nodes $C\subseteq V$ is a $(k, h, \Delta)$-core of $G_\tau$ if there exist $[a,b]$, with $b-a+1 = \Delta$, such that $C$ is a $k$-core of the static graph $G' = (V, E_{[a,b],h})$. 

\end{de}

Observe how Definition~\ref{temporal:def:temporal_core} is able to generalize the previously defined models: the intersection model of~\cite{bonchi2020_span_cores} is a $(k, \Delta, \Delta)$-core, 
and while the model of~\cite{li2018core_union} is not captured exactly (due to the interval maximality), we can see how the $(\theta,k)$-core implies the existence of a sequence of $(k, 1, \theta)$-cores in $G_\tau$, so those communities are captured by Definition~\ref{temporal:def:temporal_core} as well.
Finally, the $(k, h)$-core of~\cite{wu2015core} is also captured by Definition~\ref{temporal:def:temporal_core} as a $(k,h,\tau)$-core (where, we recall, $\tau$ is the lifespan of the temporal graph), and Definition~\ref{temporal:def:union_core} corresponds to the $(k, 1, \Delta)$-core. We remark that we do \emph{not} claim that Definition~\ref{temporal:def:temporal_core} is superior to the others, but simply that by varying its parameters we can identify the same communities.

Finally recall, as previously motivated, that we will not focus on $k$-cores themselves but rather on the \emph{coreness} of each vertex.

\begin{figure}[!htbp]
    \centering

    \begin{subfigure}[b]{\textwidth}
        \centering
        \includegraphics[width=.85\linewidth]{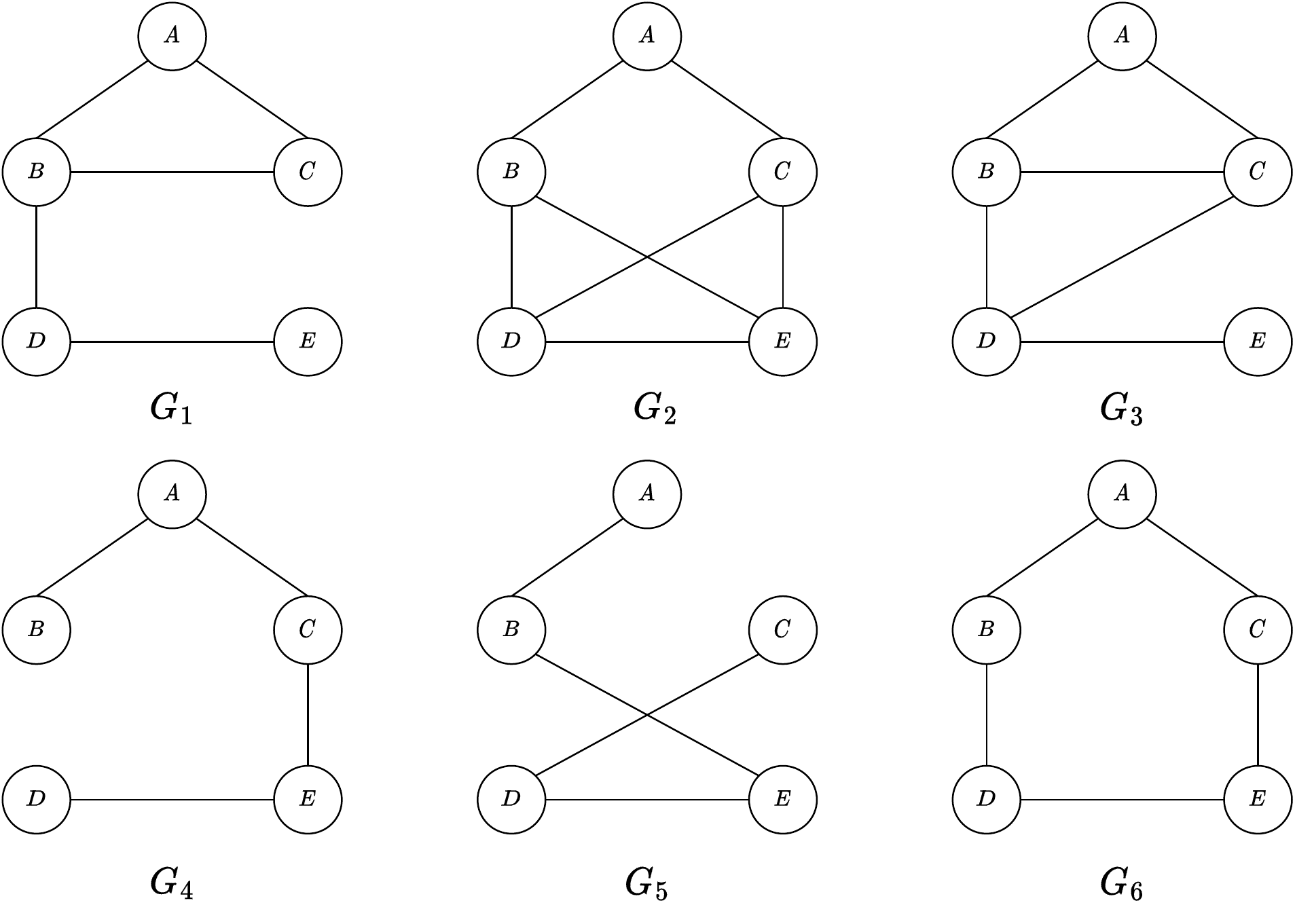}
        \caption{}
        \label{temporal:subfig:starting-graph}
    \end{subfigure}
    \vfill
    \begin{subfigure}[b]{.48\textwidth}
        \centering
        \includegraphics[width=.65\linewidth]{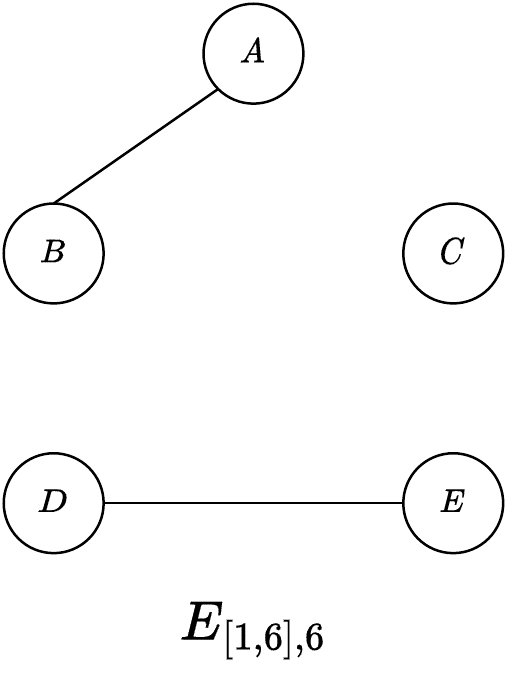}
        \caption{}
        \label{temporal:subfig:span-6}
    \end{subfigure}
    \begin{subfigure}[b]{.48\textwidth}
        \centering
        \includegraphics[width=.65\linewidth]{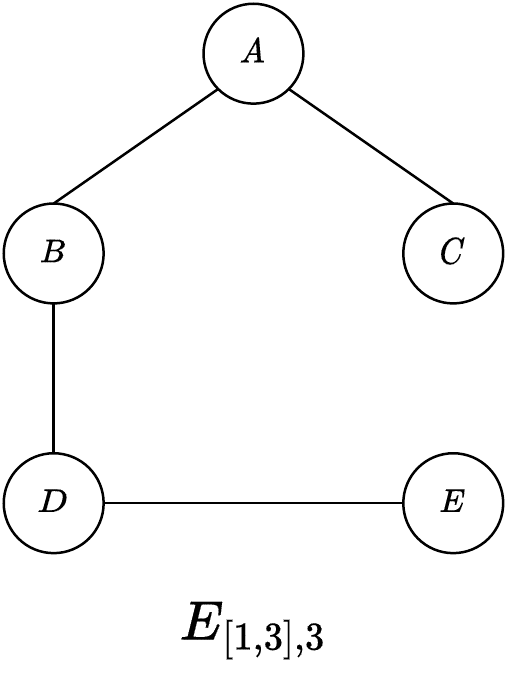}
        \caption{}
        \label{temporal:subfig:span-3}
    \end{subfigure}
    \vfill
    \begin{subfigure}[b]{.48\textwidth}
        \centering
        \includegraphics[width=.65\linewidth]{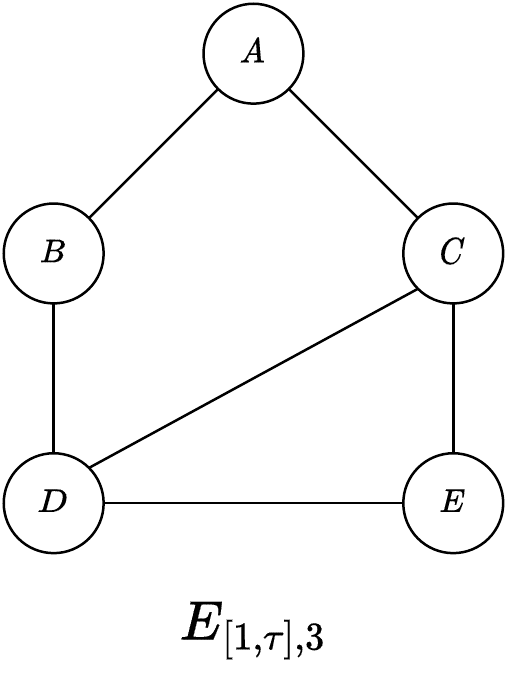}
        \caption{}
        \label{temporal:subfig:tau-3}
    \end{subfigure}
    \begin{subfigure}[b]{.48\textwidth}
        \centering
        \includegraphics[width=.65\linewidth]{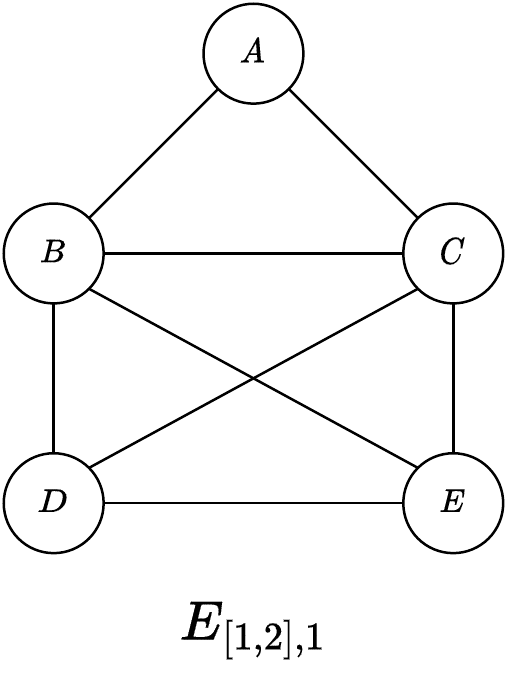}
        \caption{}
        \label{temporal:subfig:temporal-union}
    \end{subfigure}

    \caption{(a) An example of temporal graph $G_\tau$ with $\tau = 6$ snapshots $G_1, \dots, G_6$.  (b, c) Edge grouping according to the span-core \cite{bonchi2020_span_cores} definition with different intervals. (d) Edge grouping with interval $[1, \tau]$ and 3 interaction requested between nodes. (e) $\mathcal{T}^k_{[1,2],1}$, as per Definition 5.4.}
    \label{temporal:fig:temporal-examples}
\end{figure}

\section{Methodology}
\subsection{Algorithmic methods}\label{temporal:sec:algorithms}

\changed{In order to compute the coreness of each node in a given temporal graph, we first have to group together the snapshots according to some criteria: this task would cost $O(m\tau)$ time if done on an on-demaned basis.}
Therefore we exploit a heap-like tree data structure, where each leaf contains a snapshot, and every internal node contains the intermediate result of the grouping operation (union, intersection, or $h$-union) of its two children (we allow the existence of unary nodes containing just a copy of their child).
The tree is built in a bottom-up fashion: we represent it implicitly with an array, of length linear in $\tau$, that will be traversed according to the usual heap operations ($left$, $right$, $parent$) {\cite[Section 6.1]{cormen}}.
First we put each single snapshot in the second half of the array, so to fill the last level of the tree; then, starting from the middle of the array we perform the grouping operation on the two children of the node we are currently on.
Notice that the number of snapshots may not be a power of two and thus the tree would not be constructed in a proper way; to accomodate this we add some padding to the array in order to have a power of two as the number of leaves: this will possibly create some \emph{null} leaves that will have to be handled accordingly, but it will greatly simplify the subsequent traversal of the tree.
Eventually, the tree $T$ will contain the following elements: $T_0 = G_1 \circ G_2 \circ \dots \circ G_\tau$, $T_1 = G_1 \circ G_2 \circ \dots \circ G_{\tau/2}$, $T_2 = G_{\tau/2} \circ G_{\tau/2 + 1} \circ \dots \circ G_\tau$, where $\circ$ is one of $\cap$ or $\cup$\footnote{We do not compute a tree specific to \unionh, instead we use the $\cup$ tree and then filter the edges that appear less than $h$ times.}, and then
$T_{2^{\lceil\log_2 \tau \rceil - 1}} = G_0$ through to $T_{2^{\lceil\log_2 \tau \rceil - 1} + \tau} = G_\tau$.
Leaves from $T_{2^{\lceil\log_2 \tau \rceil - 1} + \tau + 1}$ to $T_{2^{\lceil\log_2 \tau \rceil + 1}}$ will contain an empty snapshot denoted by $\emptyset$.
Building this tree and storing it will require $O(m \tau)$ time and space, and we provide a small graphical example of the resulting data structure in Figure~\ref{fig:datastructure}.
\begin{figure}[hbt]
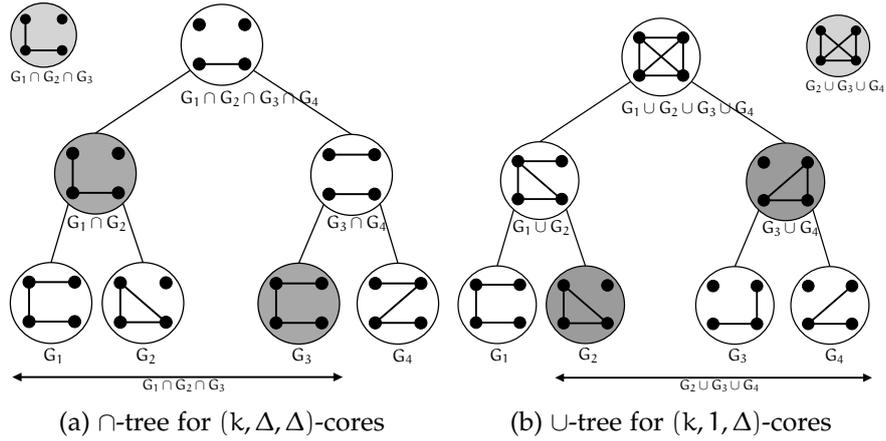

\centering
\begin{subfigure}[t]{.49\linewidth}
    \resizebox{\textwidth}{!}{
        \input{chapters/intersect_tree}%
    }
    \caption{$\cap$-tree for $(k, \Delta, \Delta)$-cores }
\end{subfigure}
\begin{subfigure}[t]{.49\linewidth}
    \resizebox{\textwidth}{!}{
        \input{chapters/union_tree}
    }
    \caption{$\cup$-tree for $(k, 1, \Delta)$-cores}
\end{subfigure}
\caption{Examples of the tree data structure on a toy $G_\tau$ (bottom row). Each tree node contains the intersection (a) or union (b) of the edges of children graphs. Highlighted nodes in (a) are the ones used to perform the intersection $G_1 \cap G_2 \cap G_3$. $k$-cores in the resulting graph correspond to the $(k, [1,3])$-span cores of \cite{bonchi2020_span_cores}.
Highlighted nodes in (b) are the ones used to perform the union $G_2 \cup G_3 \cup G_4$. The two resulting graphs are \changed{shown} in the top left and \changed{top} right corners.}
\label{fig:datastructure}
\end{figure}

We are now ready to query the tree for obtaining any \changed{interval $[a,b]$} of the snapshots combined with the $\circ$ operation of our choice, in $O(m\log\tau)$ time. 
To do so we use the procedure described in Algorithm \ref{alg:compute_interval}: we iteratively traverse the tree from \changed{$a$} to the root until we find an internal node covering the largest possible interval portion $[a, b'\le b]$, then repeat the process for $[b'+1,b]$ until all $[a,b]$ is covered; it is easily shown that at most two nodes from each level are selected, so as the tree is balanced the number of selected nodes is $\le 2\log\tau$. 
In particular, the \texttt{covered} function in Algorithm~\ref{alg:compute_interval} computes the interval of leaves that descend from the internal node $n$, using the formula shown in the algorithm. 
The formula uses $D$, the depth of the whole tree $T$, and $\left\lfloor \log_2 n\right\rfloor$ as the depth of node $n$.
Then, the algorithm tries to traverse the tree upwards until the starting interval of leaves $[a, b]$ is fully covered by the formula above. 
If only one node satisfies the condition above, it is returned; otherwise it is put in a list called \texttt{indexes} together with the other nodes that cover the remainder part of the interval. 
Notice that the maximum number of tree nodes that can be returned by Algorithm~\ref{alg:compute_interval} is 3. 
Indeed, it may happen that if the number of leaves requested is odd, i.e. $b-a+1$ is odd, then we have to add the leaf $b$ to the solutions.
In the end, we simply need to perform the $\circ$ operation between the $O(\log\tau)$ nodes selected this way, taking $O(m \log\tau)$ total time.
\begin{algorithm}[htb]
\DontPrintSemicolon
  \KwIn{An implicit binary tree $T$, two integers $a, b$}
    covered(n) $\equiv$ $[n \cdot 2^{D-2^{\lfloor\log_2 n\rfloor}}, (n+1) \cdot 2^{D-2^{\lfloor\log_2 n\rfloor}} -1 ]$\;
  
  indexes $\gets $ empty list\;
  $[s, t] \gets [a, b]$\;
  \While{$s \leq t$}{
    $n \gets s,~~D \gets depth(T)$\; 

    \lWhile{$covered(parent(n)) \subseteq [a, b]$}{ 
        $n \gets parent(n)$
    }

    Append $n$ to indexes\;
    $[s, t] \gets [s, t]\, \setminus $ covered(n)\;
  }

  \Return indexes\;
  
  \caption{Computing the topmost nodes of the tree~$T$ covering the leaves indicated by interval $[a, b]$.}
  \label{alg:compute_interval}
\end{algorithm}

Having the merged graph we can then use the classical algorithm to compute the coreness of each vertex: we maintain a priority queue with all the vertices, then we repeatedly extract the vertex with the lowest degree and remove it from the graph, updating the degrees of its neighbors accordingly, until the queue is empty \cite{core_decomposition_batagelj_zaversnik}. 
The degree of a vertex at time of their removal is its coreness.

\subsection{Analytic methods}\label{temporal:sec:analytic}

An influential work from Shin, Eliassi-Rad and Faloutsos \cite{faloutsos_corescope} showed how the coreness of a vertex is strongly correlated to its degree, i.e., if the degree of a vertex is high, then so is its coreness value.
Whenever a node deviates from this expectation, we can safely assume that the involved node is an outlier and should be further analysed in order to discover what is going on in the network. 
This can be useful for many different reasons, like highlighting a follower exchange circle on Twitter, or spotting a copy-pasted bibliography in a network of citations and more. 
The method proposed by Shin et al. is straightforward: plot the coreness of every node against their degree and see the result; in a normal situation every point in the plot will find its spot near the bisecting line $y=x$, resembling a stepped line chart (see Figure~\ref{fig:faloutsos_mirror_pattern} for an example, the authors of \cite{faloutsos_corescope} call this kind of plots as the \emph{Mirror Pattern}). 
\begin{figure}[htb]
    \centering
    \begin{subfigure}[t]{.48\linewidth}
        \centering
         \includegraphics[width=.86\linewidth]{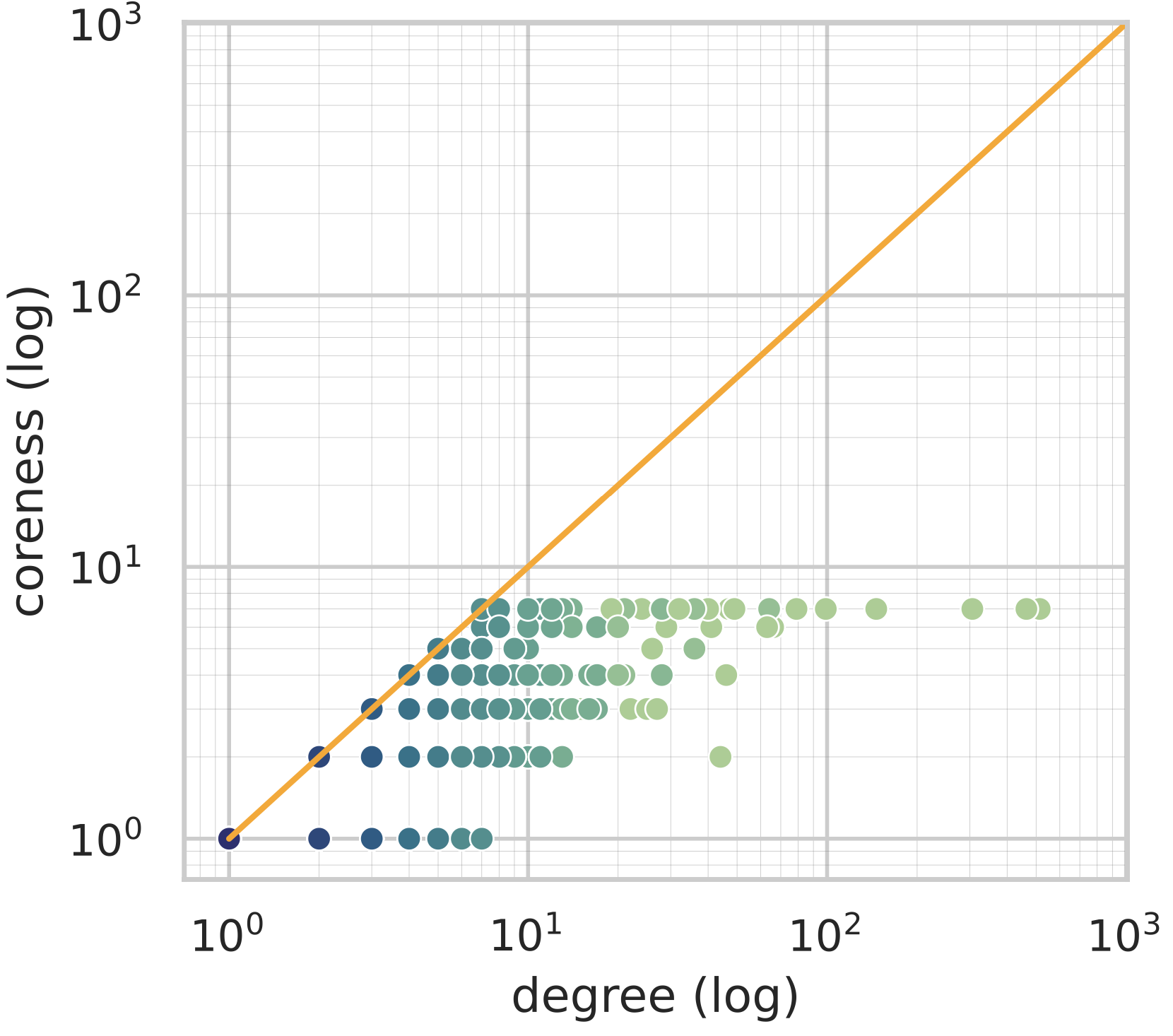}
        \caption{\footnotesize AS-733, $\Delta = 32, h = \Delta$}
    \end{subfigure}
    \begin{subfigure}[t]{.5\linewidth}
        \centering
        \includegraphics[width=\linewidth]{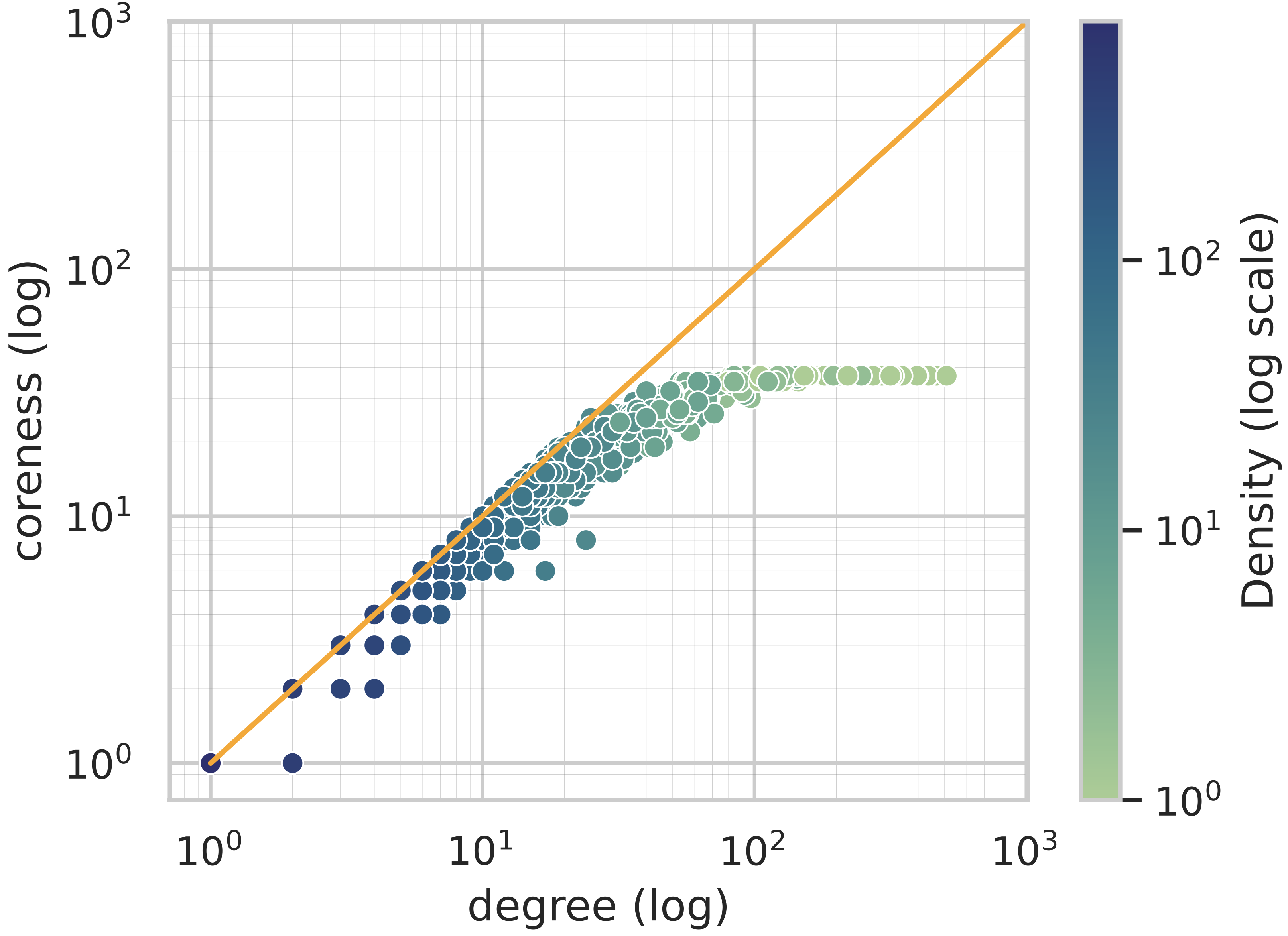}
        \caption{\footnotesize Mathoverflow, $\Delta = 32, h = 1$}
    \end{subfigure}
    \caption{Coreness-vs-degree plots (\changed{Mirror Pattern} in \cite{faloutsos_corescope}) in two of our datasets for a fixed time interval and fixed $h$.}
    \label{fig:faloutsos_mirror_pattern}
\end{figure}
\begin{figure}[htb] 
    \centering
    \begin{subfigure}[t]{.49\textwidth}
        \centering
        \includegraphics[width=\textwidth]{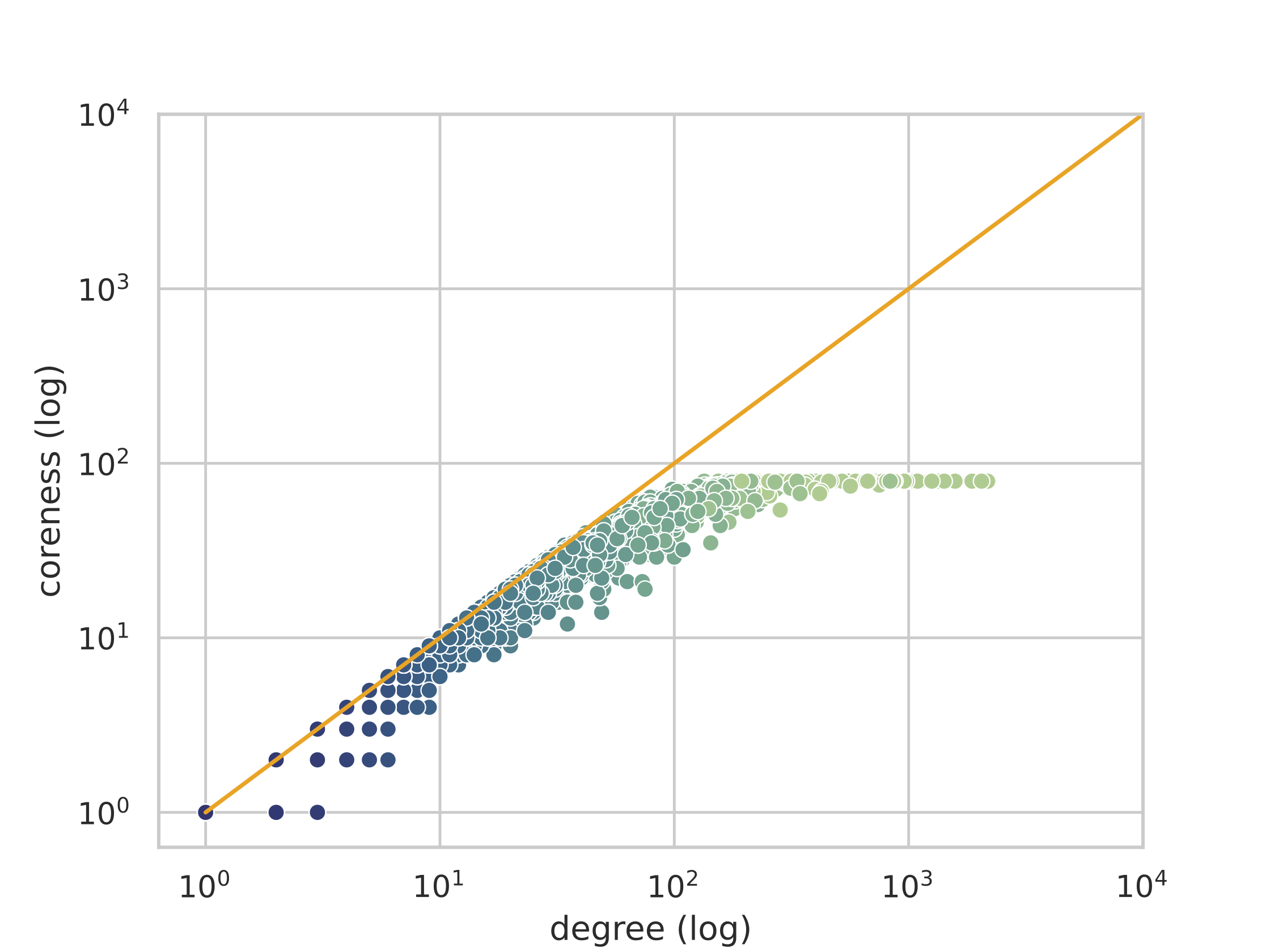}
        \caption{Mathoverflow}
    \end{subfigure}
    \begin{subfigure}[t]{.49\textwidth}
        \centering
        \includegraphics[width=\textwidth]{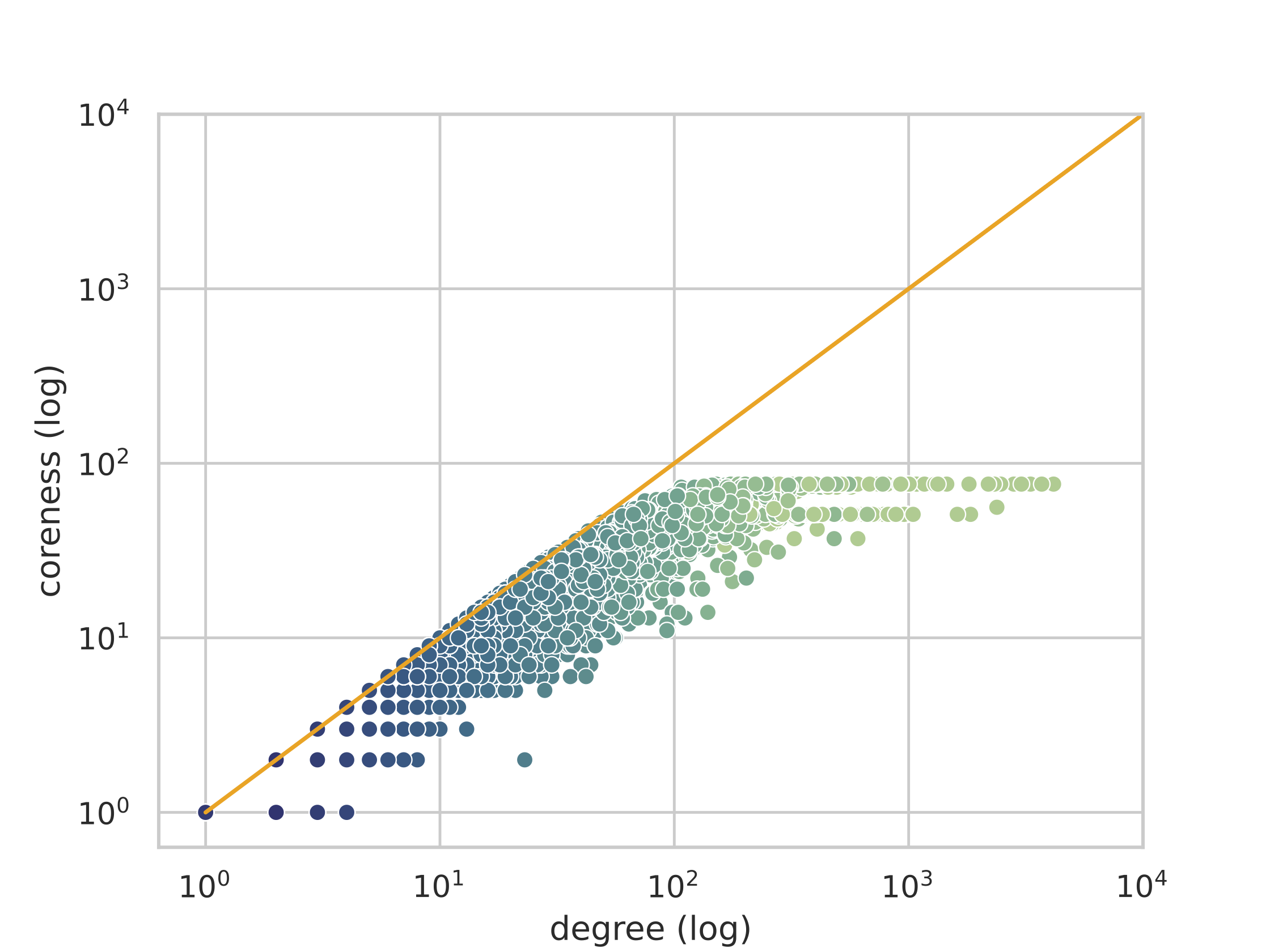}
        \caption{Reddit}
    \end{subfigure}
    \vfill
    \begin{subfigure}[t]{.49\textwidth}
        \centering
        \includegraphics[width=\textwidth]{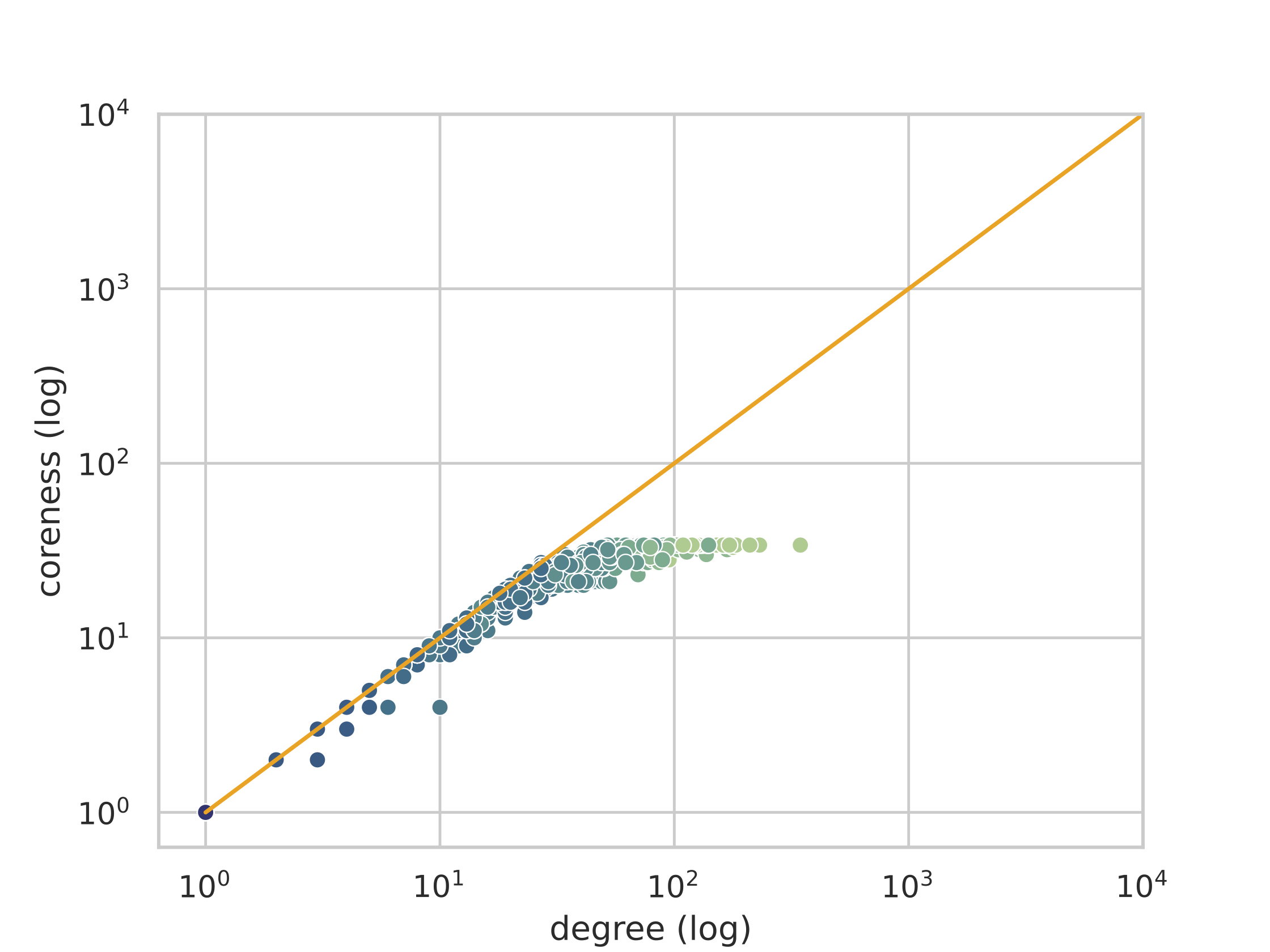}
        \caption{email-Eu-core}
    \end{subfigure}
    \begin{subfigure}[t]{.49\textwidth}
        \centering
        \includegraphics[width=1.1\textwidth]{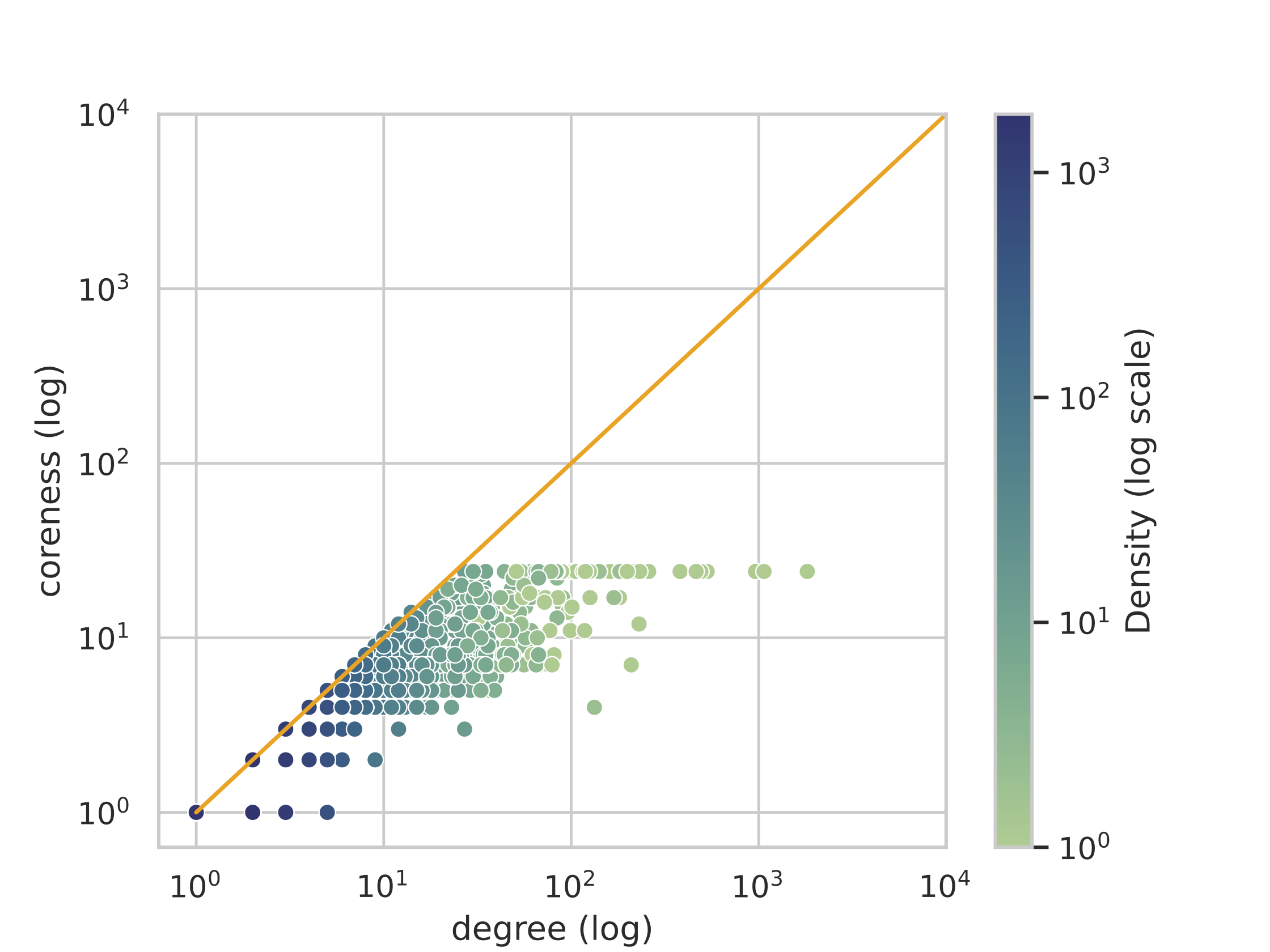}
        \caption{AS-733}
        \label{subfig:as_733_mirror_pattern}
    \end{subfigure}%
    \caption{Degree-vs-Coreness plots as in~\cite{faloutsos_corescope}, ignoring temporal information (i.e., $h = 1, \Delta = \tau$ so every edge is considered).}
    \label{fig:mirror_pattern_all}
\end{figure}
Whenever a node deviates from this, we are probably facing an abnormal event and if we have access to the metadata of the network, we can easily understand what is going on and why that specific node does not follow the expected trend of all other nodes.
Our main goal is therefore the extension of this method to the context of temporal graphs, and to do so we adopted the following methodology:
\begin{enumerate}\itemsep0em
    \item Group snapshots according to Def.~\ref{temporal:def:temporal_core} with different $h$,
    \item Compute degree\footnote{When performing the union of the snapshots, we do not allow for multiedges, so that each edge counts as one towards the final degree.} and coreness of each vertex, for each window of fixed size $\Delta$,
    \item Plot and analyse the results.
\end{enumerate}
While with static graphs it is sufficient to plot the coreness against the degree of every node, with temporal graphs we would have to look at $\tau-\Delta+1$ plots and spot any node behaving in an unexpected way, therefore we adopted different kinds of plots, including:
\begin{enumerate}[(a)]\itemsep0em
    \item Degree for each node, by $\Delta$-window,
    \item Coreness for each node, by $\Delta$-window,
    \item Coreness-vs-degree for each node, by $\Delta$-window,
    \item $\sqrt{\mathit{coreness}\cdot \mathit{degree}}$ (\textbf{RCD}) for each node, in each window of fixed size $\Delta$,
    \item Average of $\sqrt{\mathit{coreness}\cdot \mathit{degree}}$ (\textbf{ARCD}) among all windows of size $\Delta$ for each node, plotting the values for varying window size\footnote{As compromise between generality and detail, we consider powers of 2 up to $\tau -1$.}. This will be the most valuable tool in our analysis, referred to as the \emph{Temporal Resilience Plot}.
\end{enumerate}
We show examples of various plots, but we focus more on the last type as \changed{it is able to summarize the behavior of several parameters along the entire timeline.} 
Furthermore, we will also focus mainly of \textit{structural dynamics} of the networks, and use these values to extract information from the network as a whole, rather than analyzing single nodes: this is due to our choice of dataset, the SNAP repository \cite{snapnets}, which we consider for its popularity but that does not feature extensive semantic annotations for temporal graphs; such analysis is left for future work.
We believe that plotting the ARCD for each node over all the window sizes is meaningful because it provides a \changed{trajectory} view of the trend of coreness and degree for any given node, as it will highlight any unusual behavior of either coreness or degree (like a heavy fall or a steep rise) while also preserving information on the time-based nature of the graphs.

\section{Experimental analysis}\label{temporal:sec:experiments}
\subsection{Dataset and Environment}
We performed our analysis on 6 temporal graphs of different sizes and types, from the SNAP repository \cite{snapnets}; Table \ref{temporal:tab:dataset} shows the total number of nodes and temporal edges, $\tau$, and the maximum node degree and coreness among all snapshots.

The graphs are provided as edge lists, each edge marked with a timestamp. To create a uniform environment among all graphs, we grouped edges by week, i.e., each $G_i$ for each graph corresponds to the edges appearing in a consecutive 7-day interval; thus each graph spans $7\cdot \tau$ days.
Our experiments were carried on a dual-processor Intel Xeon Gold 5318Y Icelake @ 2.10GHz machine, with 48 physical cores each and 1TB of shared RAM, running Ubuntu Server 22.04 LTS, Intel C++ compiler \texttt{icpx}, version 2022.2.1.
The results were gathered in the csv format later processed with Python 3.11 using common external libraries like Pandas, Numpy and Matplotlib.

\begin{table}[ht]
\caption{Summary of our dataset.}
\label{temporal:tab:dataset}
        \centering
\resizebox{\columnwidth}{!}{
\begin{tabular}{|l|r|r|l|r|r|r|}
\hline 
 Name & Nodes & Edges & Context & $\tau$ & $d_{max}$ & $k_{max}$\\
\hline 
 sx-stackoverflow & 2,601,977
 & 63,497,050
 & Comments, questions, answers & 397 & 440  & 11  \\
\hline 
 sx-mathoverflow & 24,818
 & 506,550
 & Comments, questions, answers & 336 & 96 & 9 \\
\hline 
 soc-RedditHyperlinks & 55,863
 & 858,490

 & Hyperlinks between subreddits & 174 & 412 & 14 \\
\hline 
 email-Eu-core-temporal & 986

 & 332,334

 & Email exchanges & 76 & 89 & 9 \\
\hline 
 soc-sign-bitcoin-otc & 5,881

 & 35,592

 & Bitcoin trust network & 261 & 144 & 14 \\
\hline 
 soc-sign-bitcoin-alpha & 3,783

 & 24,186

 & Bitcoin trust network & 262 & 43 & 6 \\
 \hline
  as-733 & 
  6,474
 & 
  11,968,465
 & Autonomous System Network & 105 & 1460 & 14 \\
 \hline
\end{tabular}%
}

\end{table}

\subsection{Running time}

While optimizing our implementation is outside the scope of this work, Table \ref{temporal:tab:running_times} gives an overview of their efficiency, with the time required to build the data structure described in Section~\ref{temporal:sec:algorithms} for two aggregation functions ($\cap$ and $\cup$), as well as the average time for computing a $(k,h,\Delta)$-core decomposition. We can see how the one-time computing time for the tree is typically a second or less, yet it provides an important speedup to the computation of $(k,h,\Delta)$-cores, which becomes significant in the analysis where, for each graph, we need to compute hundreds or thousands of different windows. 

\begin{table}[htb]
    \caption{Running times (milliseconds) for computing our data structures, and average times for $(k,h,\Delta)$-core decompositions with and without our data structures.}
    \label{temporal:tab:running_times}
    \centering
    \resizebox{\columnwidth}{!}{%
    \begin{tabular}{|l|r|r|r|r|r|r|r|}
        \hline
         \multirow{2}{*}{Dataset} & \multirow{2}{*}{\makecell{Building\\ $\cup$ tree}} & \multirow{2}{*}{\makecell{Building\\ $\cap$ tree}} & \multicolumn{2}{c|}{\makecell{Average time for\\ $h=1, \Delta = 64$}} & \multicolumn{2}{c|}{\makecell{Average time for\\ $h=64, \Delta = 64$}} \\ \cline{4-7}
         & & &with&without&with&without \\ \hline
         sx-stackoverflow & 89\,940 & 76\,720 & 3\,387 & 11\,997.5 & 827 & 5\,426.6\\ \hline
         sx-mathoverflow & 1\,270 & 1\,070 & 14.9 & 88.5 & 5.5 & 25.4 \\ \hline
         soc-RedditHyperlinks & 870 & 540 & 55.5 & 260.0 & 4.9 & 16.5 \\ \hline
         email-Eu-core-temporal & 40 & 10 & 8.4 & 44.5 & 0.2 & 1.5 \\ \hline
         soc-sign-bitcoin-otc & 80 & 60 & 3.3 & 12.8 & 0.5 & 1.8 \\ \hline
         soc-sign-bitcoin-alpha & 100 & 80 & 2.5 & 9.7 & 0.6 & 2.3 \\ \hline
         AS-733 & 490 & 370 & 34.5 & 172.8 & 5.5 & 38.1 \\ 
        \hline
    \end{tabular}%
    }
\end{table}

\subsection{Experimental Setup}
We conducted all our experiments in a \emph{sliding window} fashion, i.e., we fix a value for $\Delta$ and then perform the analysis scanning the lifespan of graphs with windows of size $\Delta$, exploiting the data structure described in Section~\ref{temporal:sec:algorithms}.
Specifically, we consider for $\Delta$ the values 
$\Delta = 2^0, 2^1, 2^2, \dots, 2^{\log_2\lfloor \tau -1\rfloor}, \tau -1$.
We also tested multiple values for the parameter $h$, namely $h = 1, \frac{\Delta}{4}, \frac{\Delta}{2}, \frac{3\Delta}{4}, \Delta$; in what follows, we show results for $h = 1,  \frac{\Delta}{2}, \Delta$ as it concisely shows its range, from the plain union ($h=1$) to the intersection of edges ($h=\Delta$), with the middle ground $\frac{\Delta}{2}$.

First, we recreated the mirror pattern \cite{faloutsos_corescope} plots as a baseline static-graph-based method for comparison, then we analyzed all degree and coreness plots to spot any major events and trends, and finally we plotted the RCD and ARCD of the datasets to summarize their behavior over time at different temporal resolutions.
We cannot provide many degree/coreness plots of single snapshots as each network can have hundreds of them,
therefore in what follows \changed{ARCD plots, i.e., the Temporal Resilience Plots defined in Section \ref{temporal:sec:analytic},} will be our main tool of analysis, as they condense a vast amount of information on the temporal evolution of a network at different resolutions simultaneously, allowing us to concisely show useful insights.

\subsection{Results and Discussion}
We first show the results obtained from a broader degree-vs-coreness analysis conducted according to the same principles of the Mirror Pattern analysis in \cite{faloutsos_corescope}, with parameters $\Delta = \tau$ and $h=1$, which corresponds to treating the whole temporal graph as a static graph, taking all the temporal edges (ignoring multiedges).
Figure \ref{fig:mirror_pattern_all} shows the plots for $4$ graphs in our dataset: from these \changed{figures} we cannot extract proper information about the communities, nor about the behavior of the nodes during the lifespan of the graphs, except for the AS-733 graph (Figure \ref{subfig:as_733_mirror_pattern}) where we can see that a small number of nodes stand out for their higher degree and relatively high core number. 
This static analysis provides us with a first insight into the AS-733 dataset that may highlight important nodes. However, the loss of temporal information does not allow us to understand fluctuations of behaviour, e.g., if communities are cohesive or just the result of aggregating different moments in time, or if the importance of a node is long- or short-lived.

Next, we discuss each graph in our dataset separately.

\begin{figure}[hbt]
    \centering
    \begin{subfigure}[t]{.49\linewidth}
        \centering
        \includegraphics[width=\textwidth]{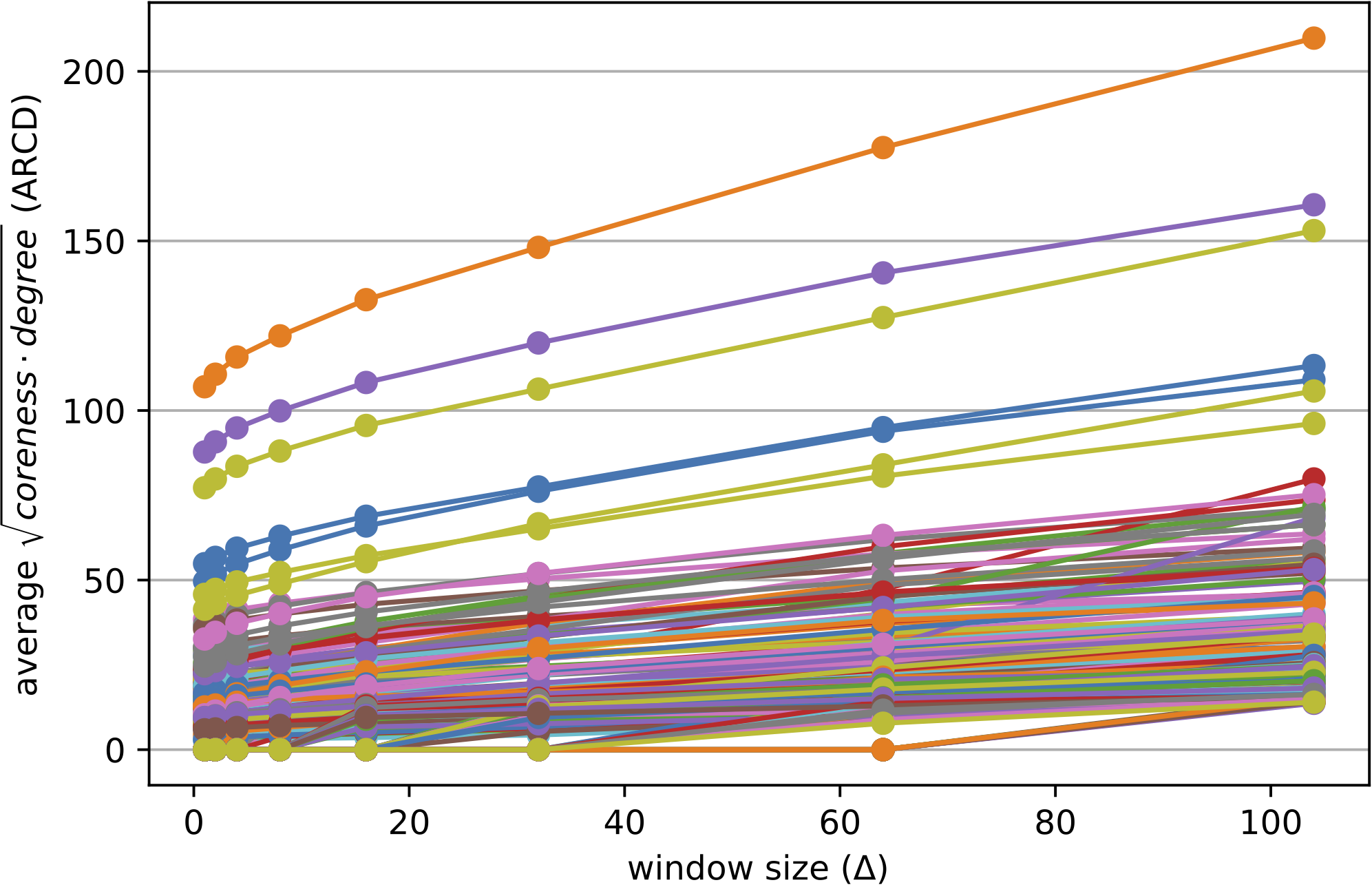}
        \caption{$h = 1$}\label{fig:as-h1}
    \end{subfigure}
    \begin{subfigure}[t]{.49\linewidth}
        \centering
        \includegraphics[width=\textwidth]{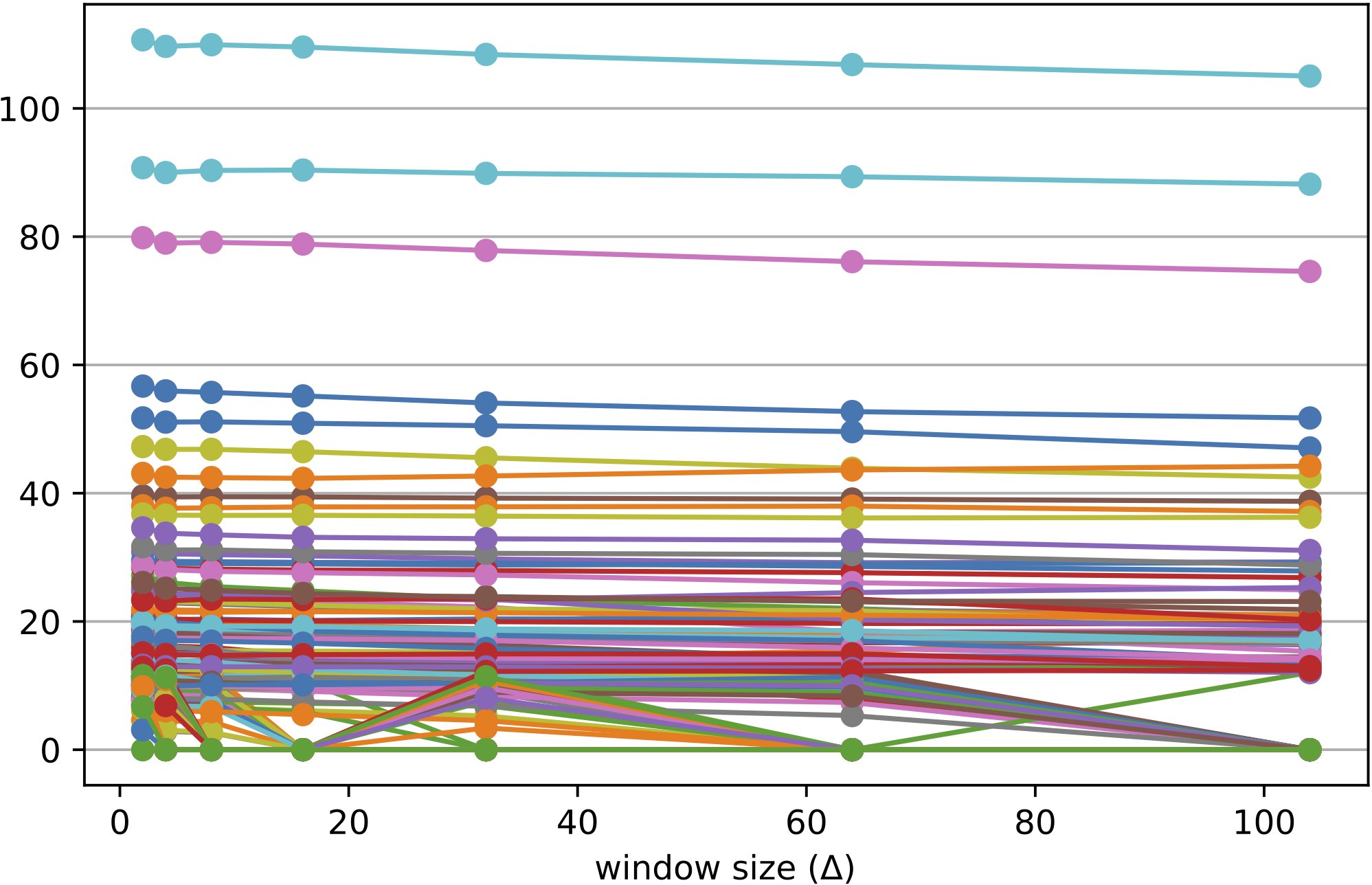}
        \caption{$h = \Delta/2$}
    \end{subfigure}
    \begin{subfigure}[t]{.49\linewidth}
        \centering
        \includegraphics[width=\textwidth]{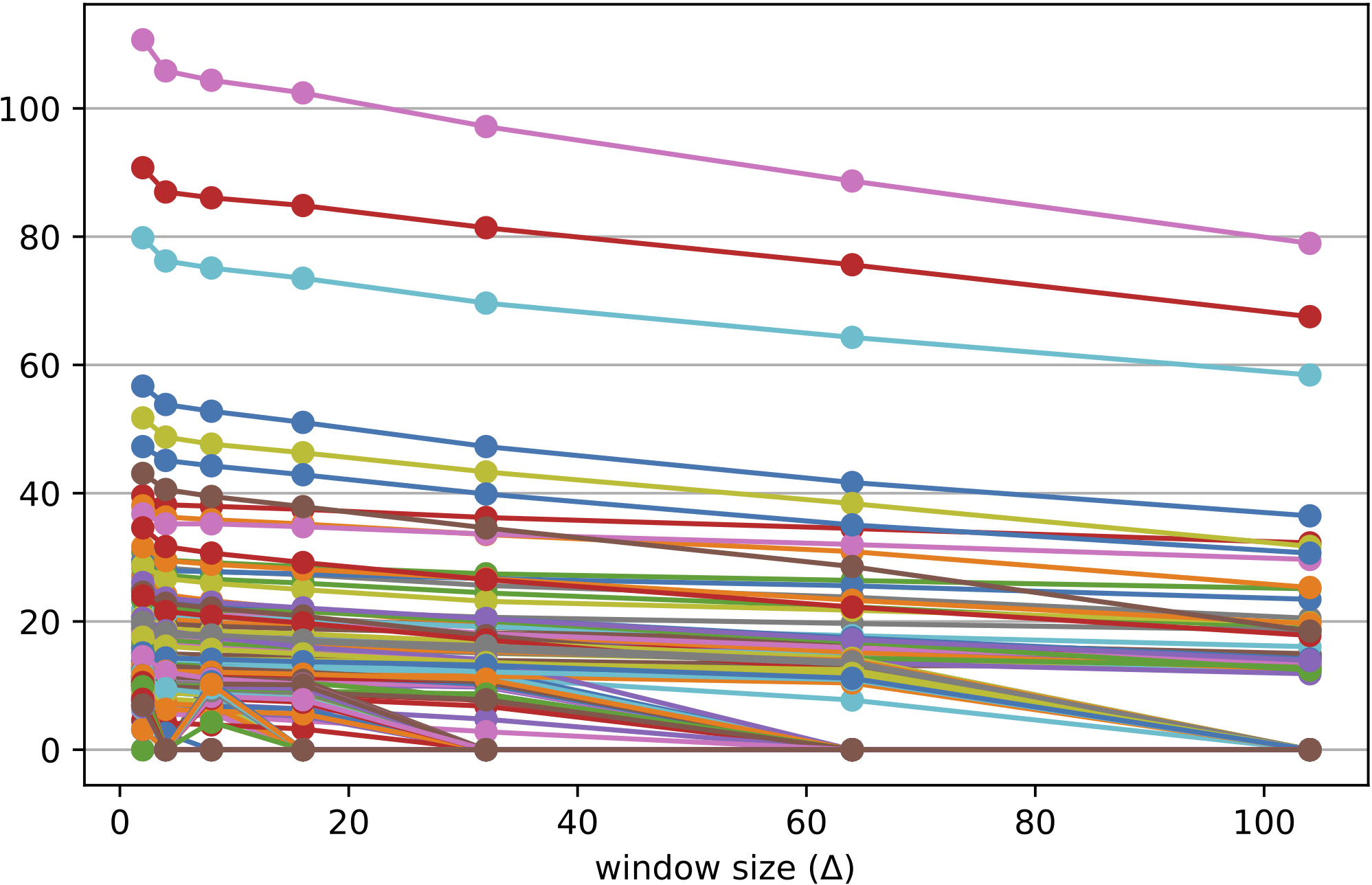}
        \caption{$h = 3\Delta/4$}
    \end{subfigure}
    \begin{subfigure}[t]{.49\linewidth}
        \centering
        \includegraphics[width=\textwidth]{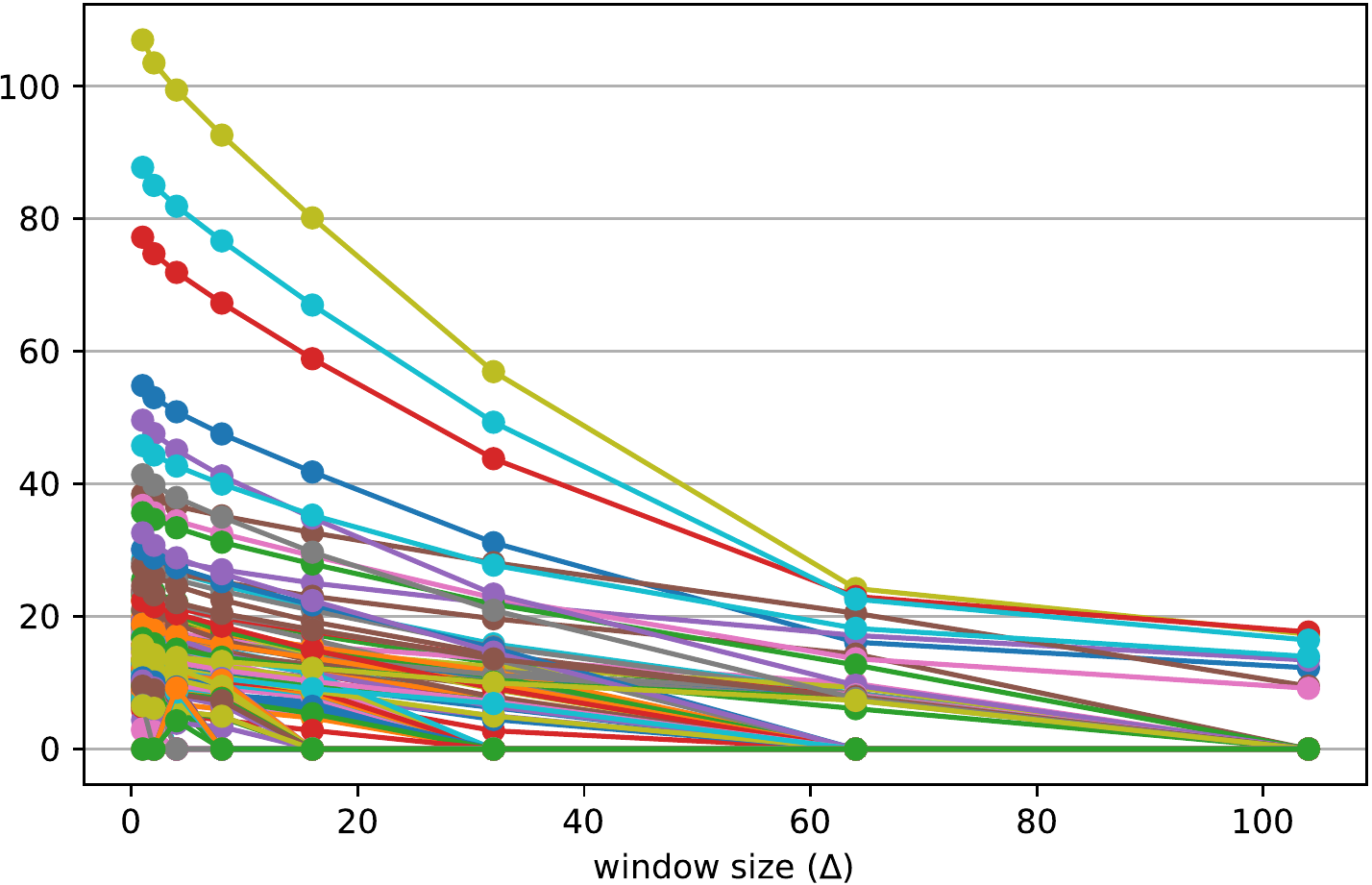}
        \caption{$h = \Delta$}\label{fig:as-complete-d}
    \end{subfigure}%
    \caption{Temporal Resilience plots for the AS-733 dataset for different values of $h$.}
    \label{fig:as_733_complete}
\end{figure}
\paragraph{AS-733.}
This graph represents connections between autonomous systems from November 8, 1997 to January 2, 2000. We first plotted the degree for each node with a fixed window size of $\Delta = 2$ and $h = \Delta$, i.e., performing the intersection of the edges in the windows (Figure~\ref{fig:degree_autonomous_system}). We firstly observe how the degrees seem to grow in time, consistently with the rapid growth of the Internet around the turn of the century. It is also immediately evident that around window number 60 there is a brief but dramatic degree drop for most nodes, especially the ones of highest degree. From the data provided, this \changed{loss of connectivity} happened around Christmas of 1998, and while we are not able to pinpoint a specific event that may have caused it, it is a clear example of temporal behavior that would not emerge from a static view of the graph. 

On the other hand, this drop does not show as much in the coreness of the nodes in the same graph grouping, as shown in Figure~\ref{fig:coreness_autonomous_system}: while it is possible to distinguish the anomaly around window 60, there is no clear \changed{gap} as several node maintain coreness values comparable to the rest of the timeline. Furthermore, Figure \ref{fig:coreness_autonomous_system} does not seem to highlight any particularly meaningful pattern.

\begin{figure}[htb]
    \centering
    \begin{subfigure}[t]{\linewidth}
    \includegraphics[width=.95\linewidth]{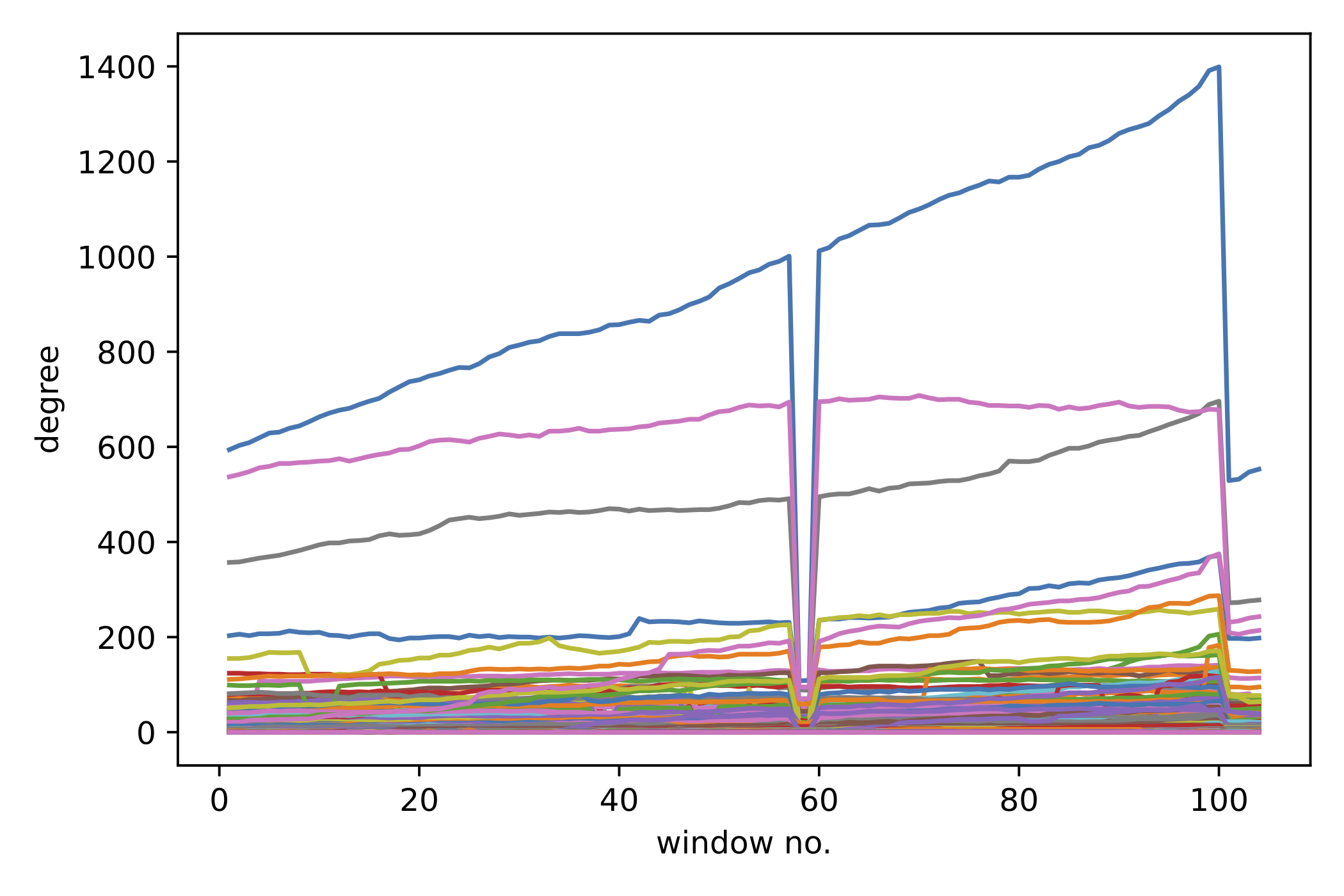}
    \caption{}
    \label{fig:degree_autonomous_system}
    \end{subfigure}
    \begin{subfigure}[t]{.9\linewidth}
    \centering
    \includegraphics[scale=.16]{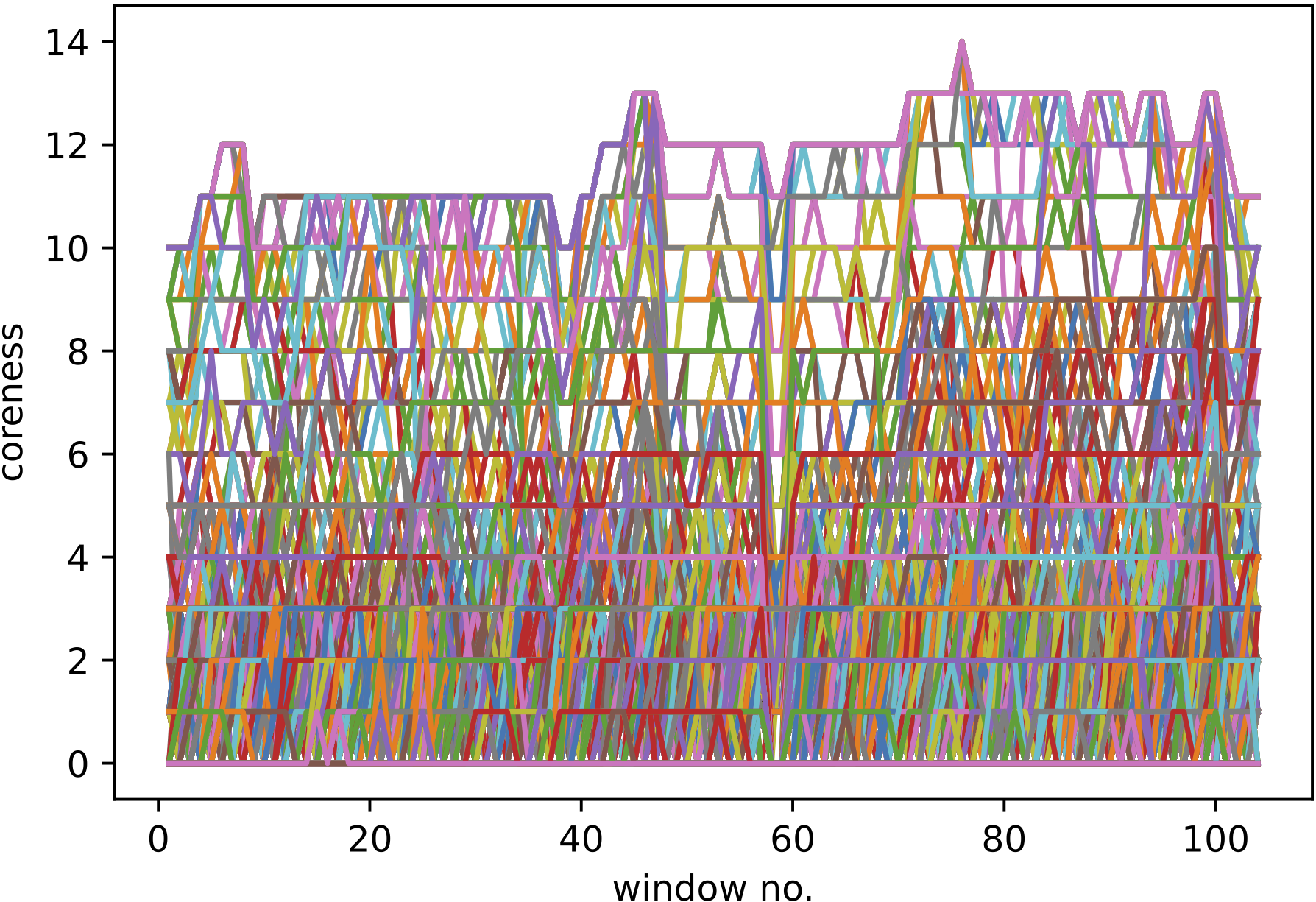}
    \caption{}
    \label{fig:coreness_autonomous_system}
    \end{subfigure}

    \caption{Degree and coreness of all nodes on the Autonomous System dataset AS-733 for $\Delta = 2, h = \Delta$.}
    \label{fig:degree_and_coreness_as733}
\end{figure}

Let us thus consider the \emph{Temporal Resilience plots} in Figure~\ref{fig:as_733_complete}: each line corresponds to a node, and for each window size $\Delta$ ($x$ axis) we plot the average RCD among all windows of size $\Delta$. The plots consider different values of $h$ for computing the $(k,h,\Delta)$-cores. 

We can clearly identify three vertices standing out for their ARCD always above all of the other nodes: these nodes must act like a backbone to the autonomous system network and are resilient to the changes that happen, like the Christmas \changed{event}, even if their degree dropped significantly. 

This is made even clearer by looking at the sequence of plots given by the values $h = 1, \frac{\Delta}{2}, \frac{3\Delta}{4}, \Delta$, where we can see that the three topmost lines ``survive'' the increasing constraint given by the different $h$ values, while the others start approaching zero as soon as $h = \Delta/2$.
This also visualizes the intuition and expectation that the connectivity of the nodes in any temporal graph decreases \changed{as we} increase $h$, i.e., when we require more and more interactions between any pair of nodes for a given window size. This also highlights something remarkable: while the high degrees in Figure~\ref{fig:degree_autonomous_system} could already suggest important nodes, the temporal resilience of Figure~\ref{fig:as-complete-d} (where $h=\Delta$ means intersection is performed) highlights the fact that these node form consistently unbroken communities, as their coreness remains significant even when intersecting \textit{all} windows, i.e., when $\Delta=\tau$ (right-most points in the lines).

With this temporal resilience plot we can also classify the nodes according to their \changed{falling point} in the plot. Specifically, we assign classes based on the time their ARCD value drops to zero, as we enlarge the window size.
We colored the ARCD plot according to the above classes obtaining Figure~\ref{fig:as_733_classes}, from which we can see how the classes are well separated from each other, validating our hypothesis that some nodes are resilient to \changed{the} process of enlarging the window size while keeping the value of $h$ high.
It is worth noting that the classes become more distinguishable as the window size increases, whereas for small values of $\Delta$ it may happen that less-resilient nodes (i.e. their falling point comes earlier in the x-axis) have ARCD values higher than more-resilient ones, 
which we may expect due to smaller values being more susceptible to fluctuations.
We can then conclude that in the AS-733 we can identify backbone nodes, that play a fundamental role in the network and are extremely robust in their mutual connections. Even more remarkably, Figure~\ref{fig:as_733_classes} shows that some of the most resilient nodes (colored green) can only be recognized over very wide windows: when considering small windows, their ARCD is surpassed by nodes of lower classes (orange and purple). This shows that important insight is missed if we don't analyse the dataset at different temporal resolutions. 

Finally, observe the coreness/degree plot for the same dataset in Figure~\ref{subfig:as_733_mirror_pattern}: while we may identify important nodes, all of the other dynamics highlighted are not evident from it, which underlines the value of the analysis shown.

\begin{figure}[htb]
    \centering
    \includegraphics[width=.85\columnwidth]{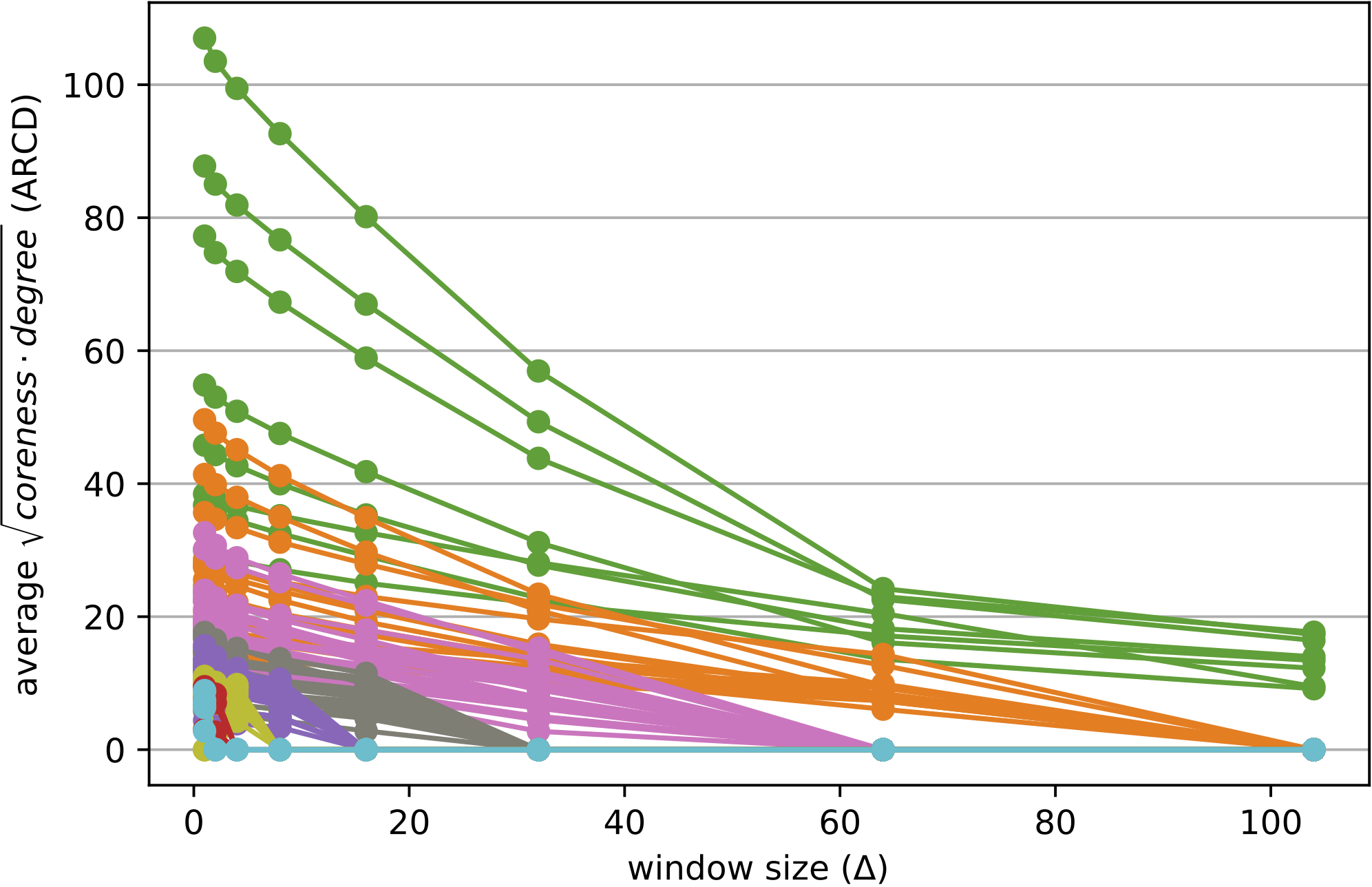}
    \caption{Coloring of the nodes of AS-733, according to their falling points in the Temporal Resilience plot.}
    \label{fig:as_733_classes}
\end{figure}

\paragraph{Math- and Stack-Overflow.} 
Nodes in these graph are users from the StackExchange Q\&A network, and edges are drawn when users respond to questions or comment on questions and responses.

\begin{figure}[htb]
    \centering
    \begin{subfigure}[t]{.48\columnwidth}
        \centering
        \includegraphics[width=\linewidth]{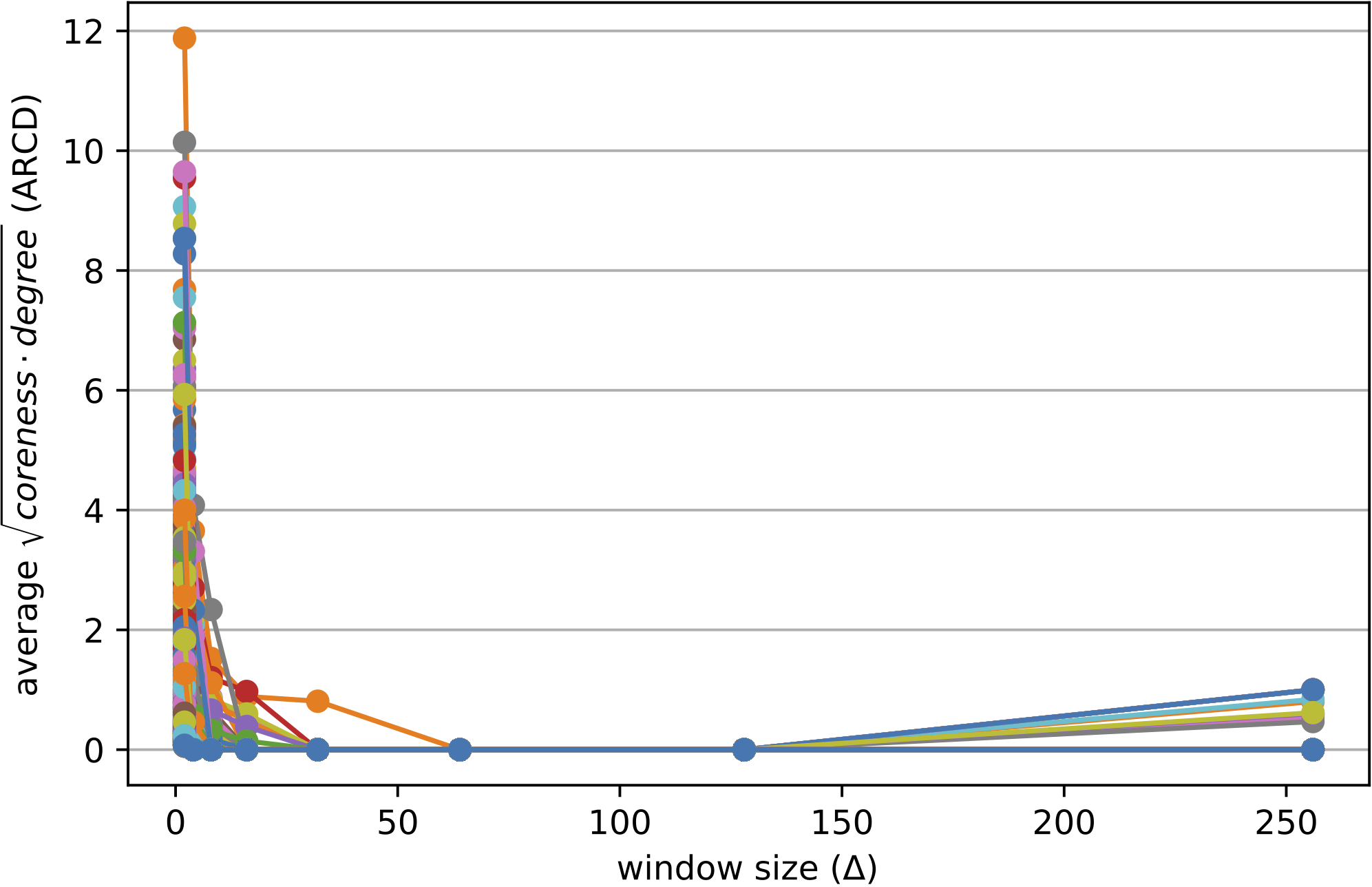}
        \caption{Mathoverflow}
    \end{subfigure}
    \begin{subfigure}[t]{.48\columnwidth}
        \centering
        \includegraphics[width=\linewidth]{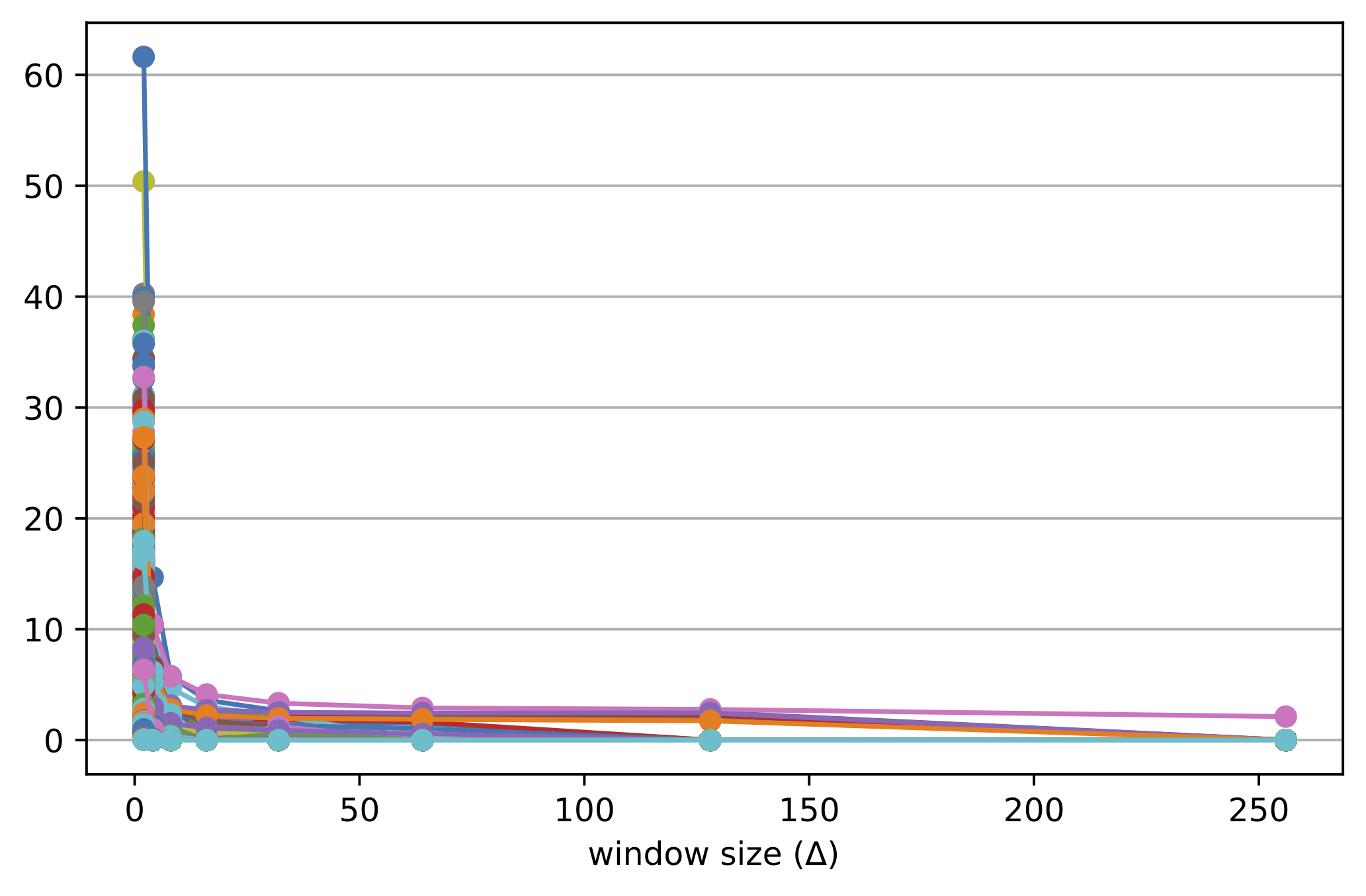}
        \caption{Stackoverflow}
    \end{subfigure}
    \caption{Temporal Resilience plots for the mathoverflow and stackoverflow datasets for $h = \Delta/2$.}
    \label{fig:avg_arcd_mathoverflow_delta_over_two}
\end{figure}

The temporal resilience plot in Figure~\ref{fig:avg_arcd_mathoverflow_delta_over_two} shows how most of the nodes are not able to form lasting communities, even if we require them to interact just half of the time ($h = \Delta/2$).
We do observe positive values of ARCD for $\Delta=1$, meaning connections and communities are created, but as soon as we enlarge $\Delta$ they quickly drop to $0$ or almost 0 (the drop would of course be even steeper with higher values of $h$). 
This behaviour is opposite to that of AS-733, and shows us that communities are not at all resilient but ephemeral. Remarkably, this is consistent with the functioning dynamics of the StackExchange system: questions are short-lived and recommended randomly, and there are no ``friendship'' or ``follow'' features, so we can expect random users to group briefly on a question and then not interact again. We can indeed see how the resilience plot captures this dynamic, and discriminates between this network and AS-733 better than the static plots of Figure~\ref{fig:mirror_pattern_all}.

To further highlight the volatility of these communities we can look at the Temporal Resilience plot for $h = 1$ (i.e., union), in Figure~\ref{fig:avg_arcd_mathoverflow_union}. 
Here the ARCD values grow uniformly, so both the coreness and the degree increase as we include more snapshots into the windows, meaning people continuously interact with new \changed{random} ones (differently, e.g., from the same plot for AS-733 in Figure~\ref{fig:as-h1}, where more stable interactions with the same nodes cause a slower growth).

\begin{figure}[htb]
    \centering
    \begin{subfigure}[t]{.48\textwidth}
    \includegraphics[width=\textwidth]{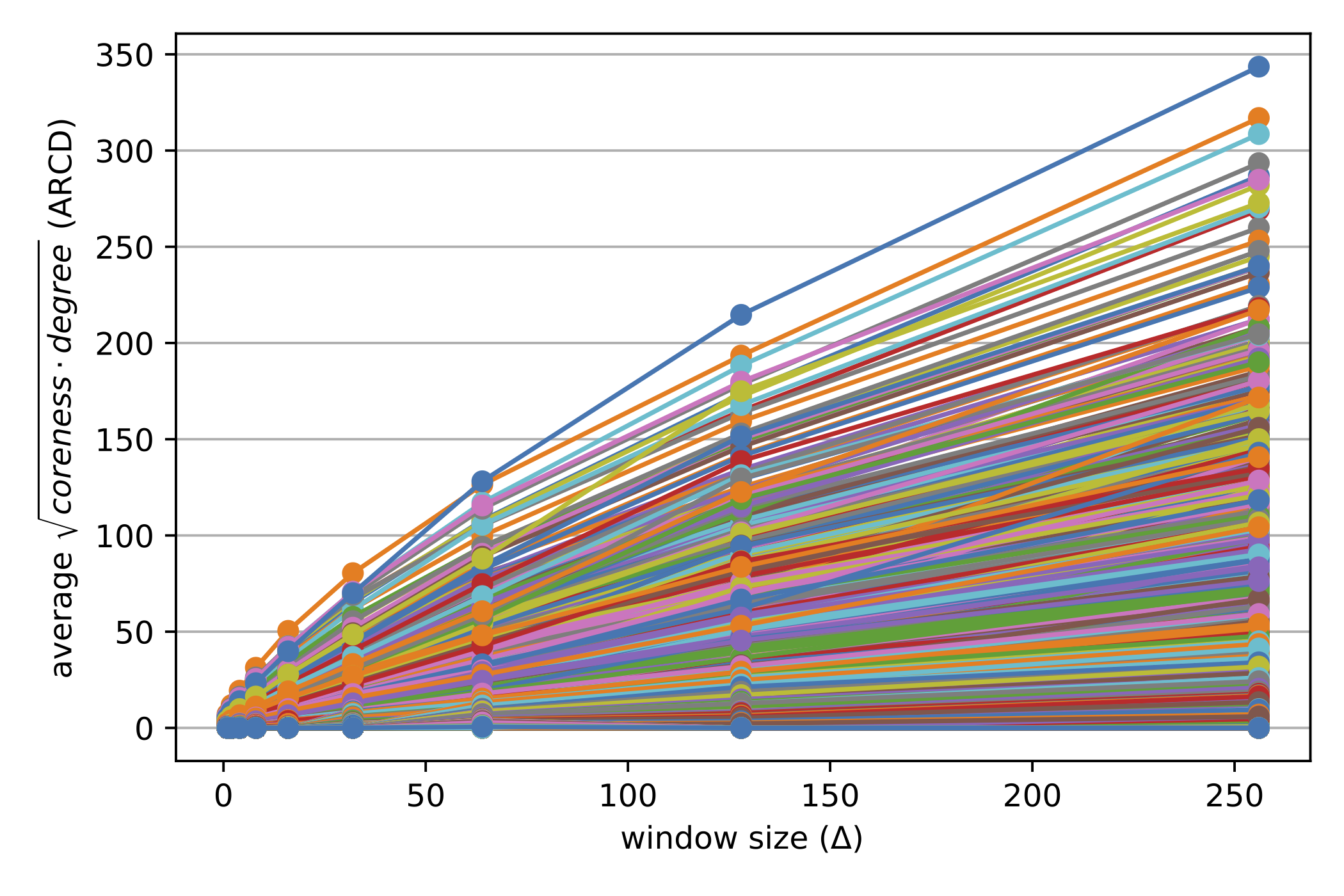}
    \caption{$h = 1$}
    \end{subfigure}
    \begin{subfigure}[t]{.48\textwidth}
    \includegraphics[width=\textwidth]{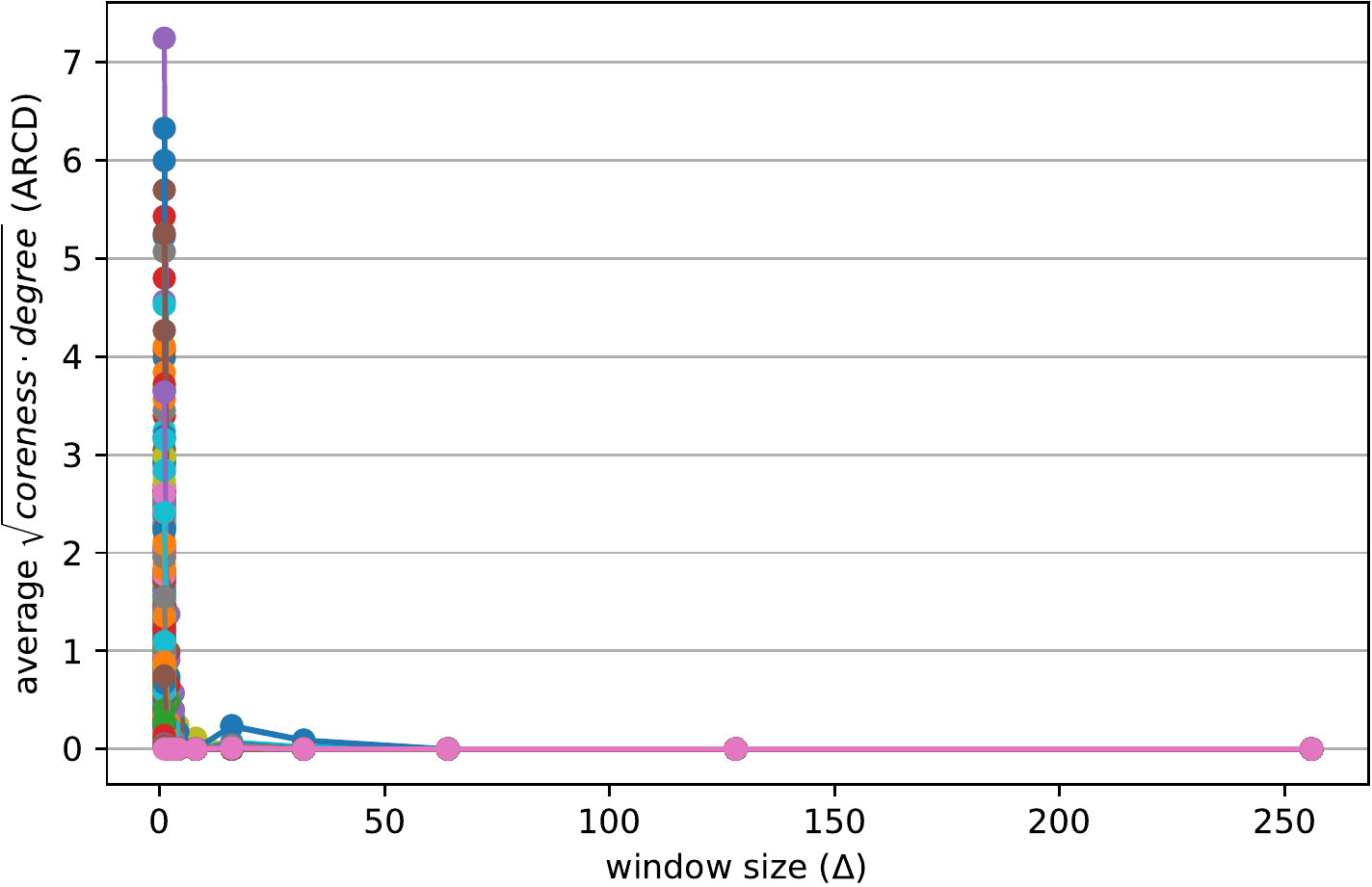}
    \caption{$h = \Delta$}
    \end{subfigure}
    \caption{Temporal Resilience plot for the Mathoverflow dataset for $h=1$ and $h = \Delta$.}
    \label{fig:avg_arcd_mathoverflow_union}
\end{figure}

\paragraph{Bitcoin Networks.}
Bitcoin-alpha and bitcoin-otc are ``who-trusts-whom'' network among Bitcoin users.
We can observe here a volatile behaviour similar to StackExchange networks: looking at the ARCD distribution in Figure~\ref{fig:arcd_bitcoin_intersection} we already see low maximum absolute values, that reach 0 very quickly as the window size increases. The ephemerality of communities is further underlined by the plot staying essentially the same even if we lower $h$ to $\Delta/2$.
However, setting $h=1$ (union), shown in Figure~\ref{fig:arcd_bitcoin_union}, tells us something interesting: here we can observe two classes of nodes for both bitcoin networks, reminiscent of AS-733; few nodes starkly above all others, akin to ``hubs'' in terms of connectivity, then most nodes in the middle on a continuous scale of higher to lower connectivity.
This analysis highlights a dynamic of the bitcoin trust networks that sees few very influential entities, while most people have several interactions over time but typically with new people, hence the absence of persistent communities.

\begin{figure}[htb]
    \centering
    \begin{subfigure}[t]{.48\columnwidth}
        \centering
        \includegraphics[width=\linewidth]{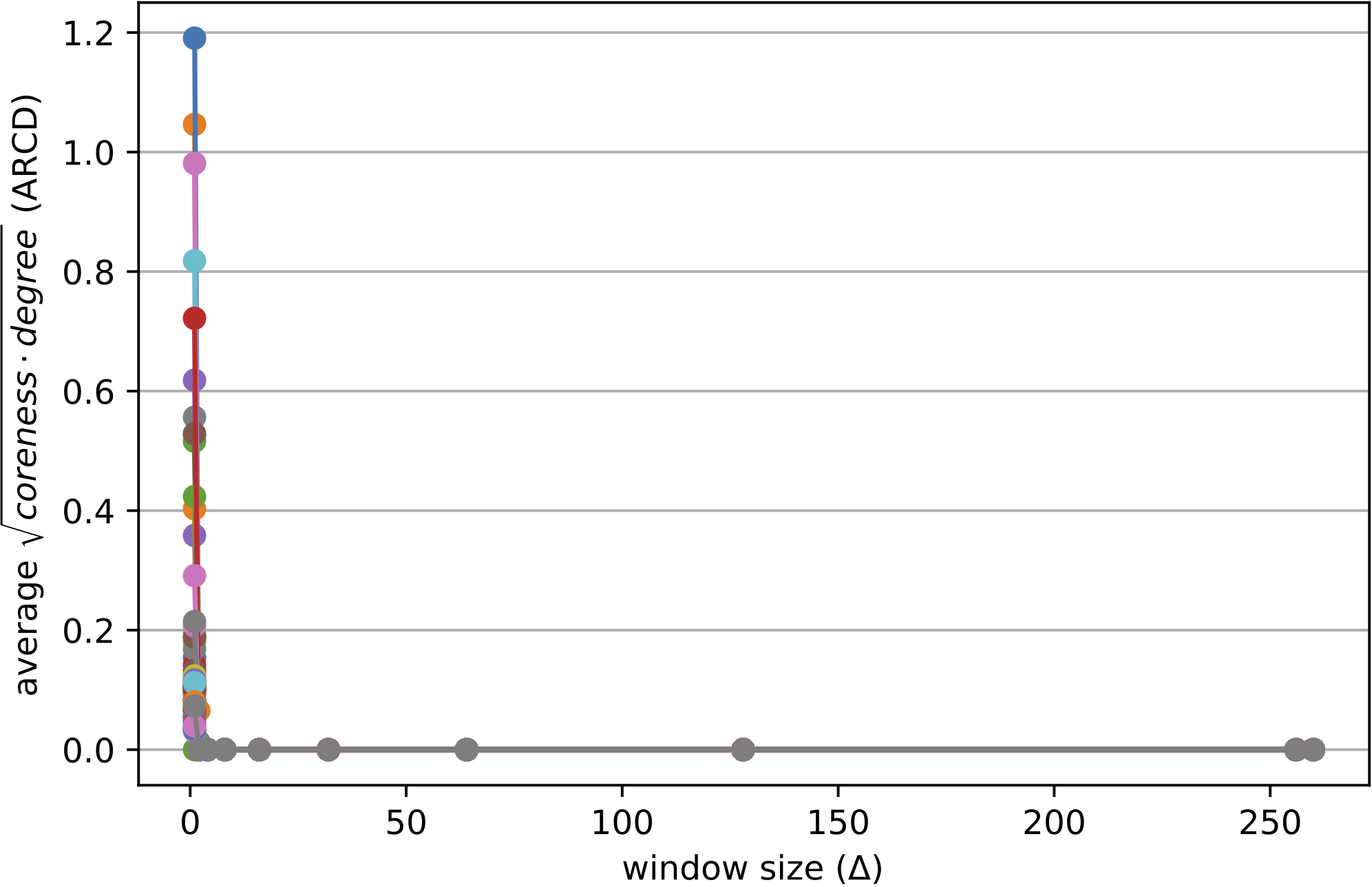}
        \caption{Bitcoin-otc}
    \end{subfigure}
    \begin{subfigure}[t]{.48\columnwidth}
        \centering
        \includegraphics[width=\linewidth]{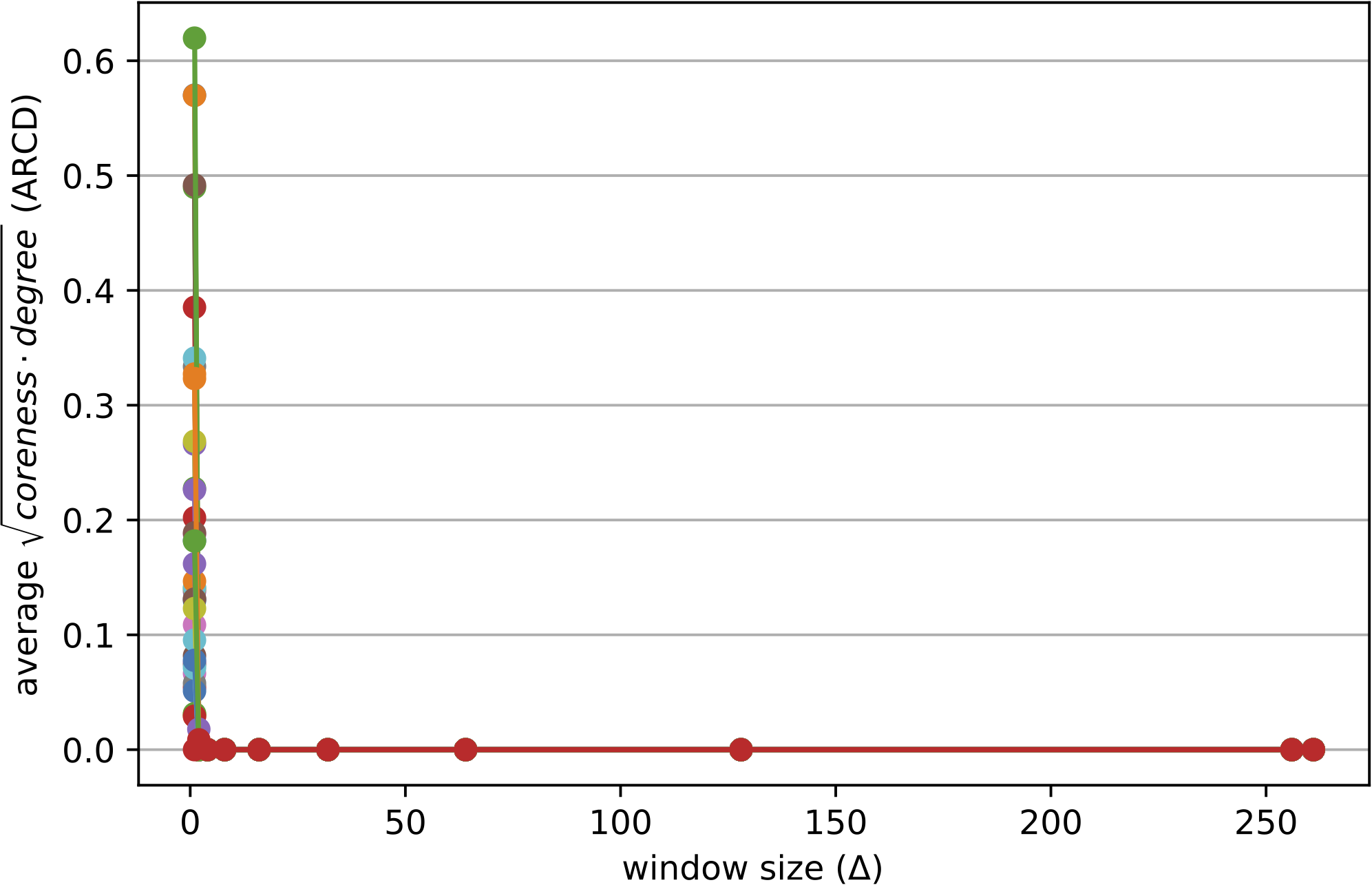}
        \caption{Bitcoin-alpha}
    \end{subfigure}
    \caption{Temporal Resilience plots for the bitcoin-otc and bitcoin-alpha datasets for $h = \Delta$ (intersection).}
    \label{fig:arcd_bitcoin_intersection}
\end{figure}

\begin{figure}[htb]
    \centering
    \begin{subfigure}[t]{.48\columnwidth}
        \centering
        \includegraphics[width=\linewidth]{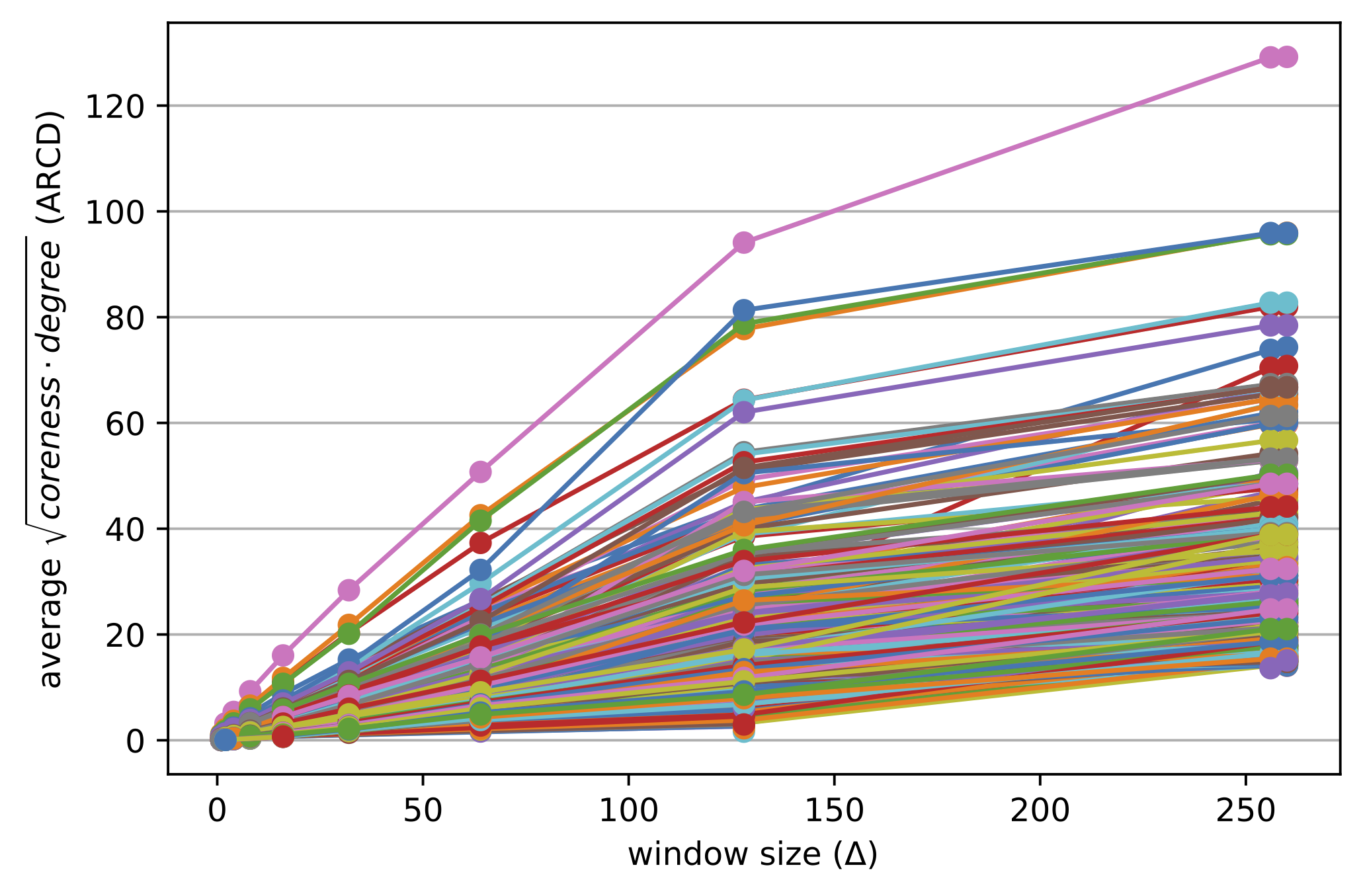}
        \caption{Bitcoin-otc}
    \end{subfigure}
    \begin{subfigure}[t]{.48\columnwidth}
        \centering
        \includegraphics[width=\linewidth]{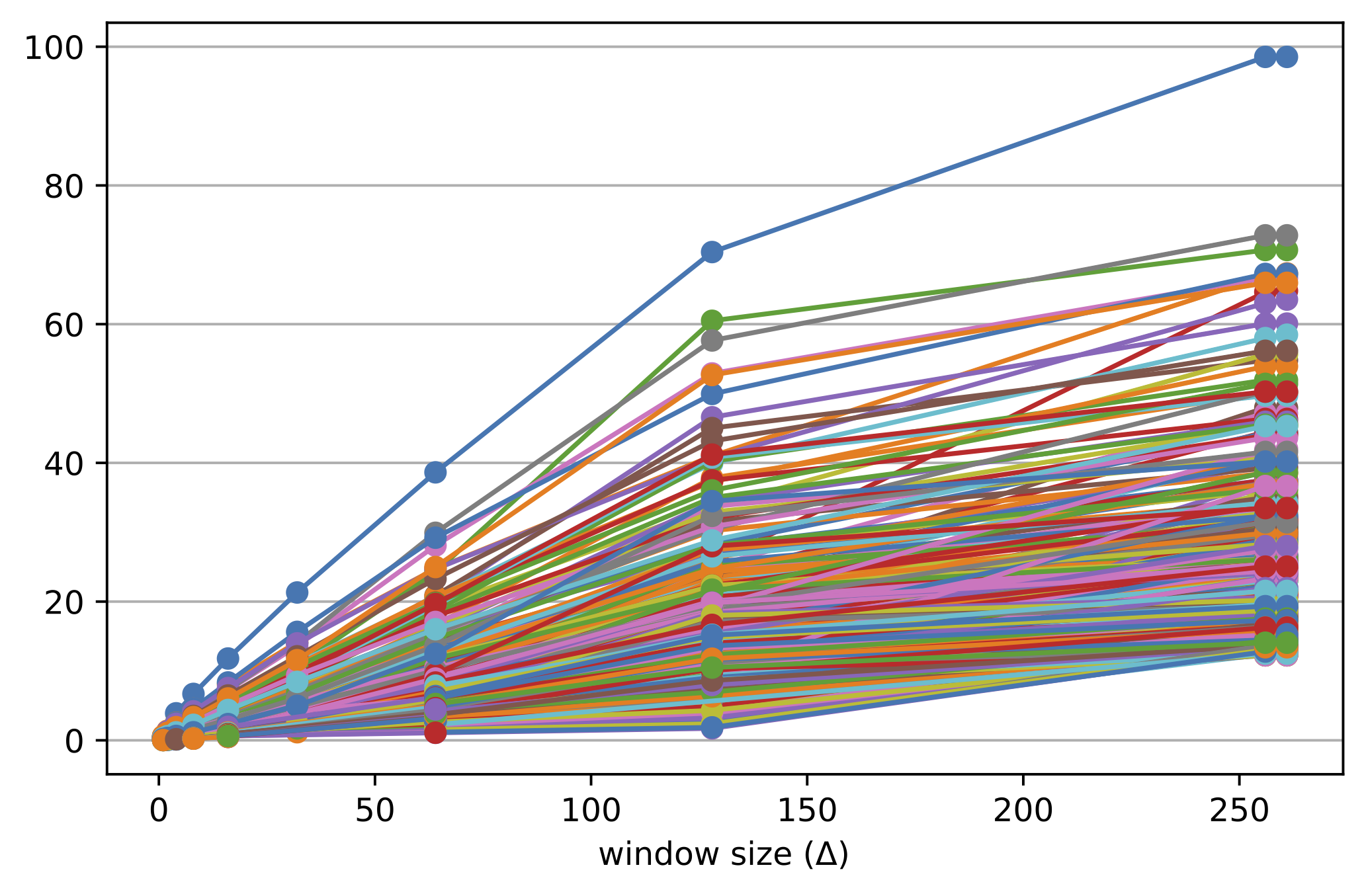}
        \caption{Bitcoin-alpha}
    \end{subfigure}
    \caption{Temporal Resilience plots for the bitcoin-otc and bitcoin-alpha datasets for $h = 1$ (union).}
    \label{fig:arcd_bitcoin_union}
\end{figure}

\paragraph{Email-EU-core.}

This network shows email exchanges between members of a large European research institution.

The temporal resilience plot, shown for $h=\Delta/2$ in Figure~\ref{fig:email_classes}, again uncovers interesting information: several nodes maintain high ARCD for the whole plot, thus forming a persistent community, but the others are not all immediately broken; it appears the network contains a mix of short-lasting, long-lasting and medium-lasting communities, as for each falling point in the graph we see a reasonably well-sized class of nodes. Although some nodes may simply correspond to employees on a short-term contract, this possibly suggests a dynamic structure of the research institute with frequent creation of projects of different size and duration.

\begin{figure}[htb]
    \centering
    \begin{subfigure}[b]{.48\textwidth}
        \includegraphics[width=\textwidth]{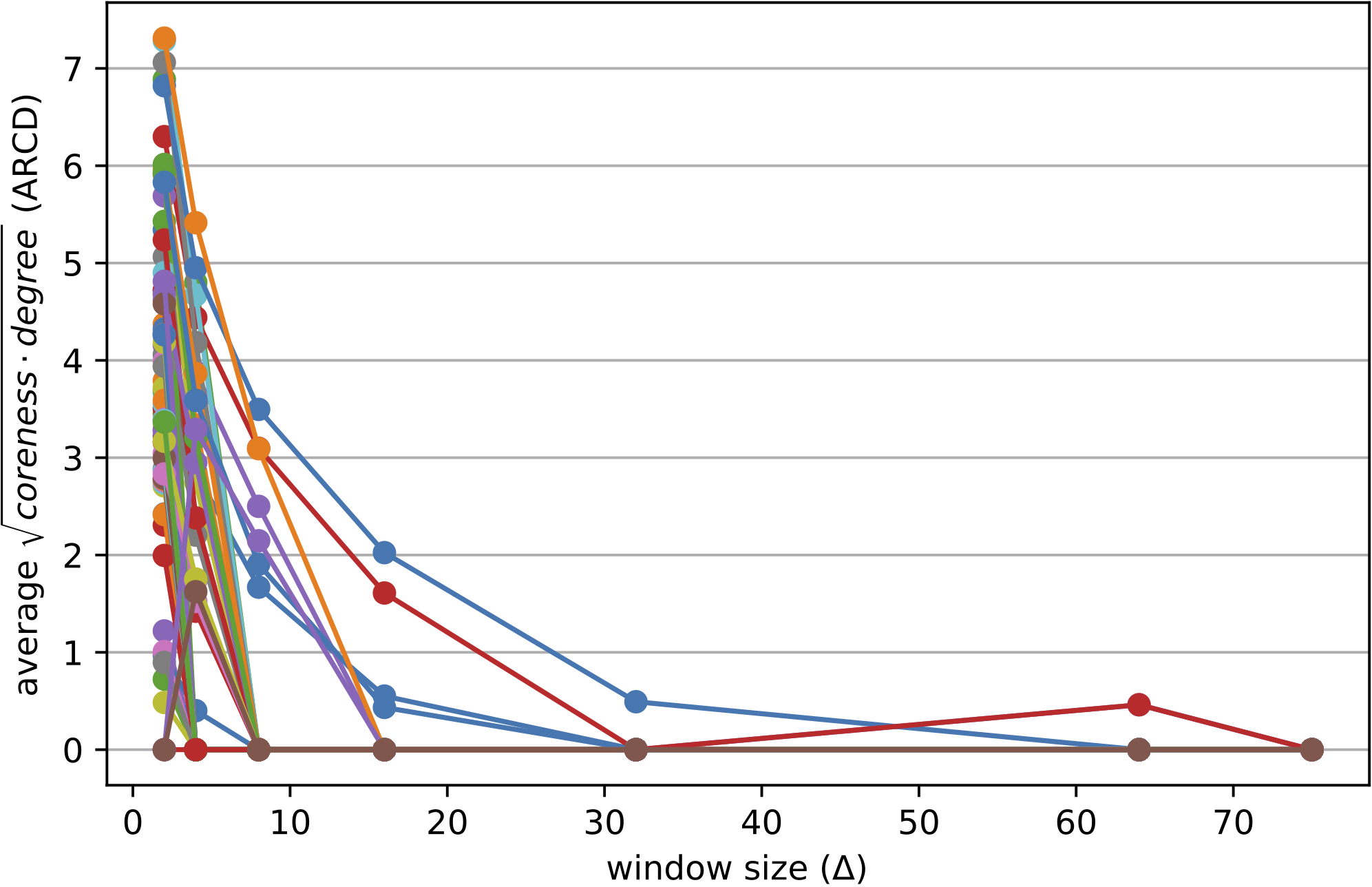}
        \caption{$h = \Delta$}
    \end{subfigure}
    \begin{subfigure}[b]{.48\textwidth}
        \includegraphics[width=\textwidth]{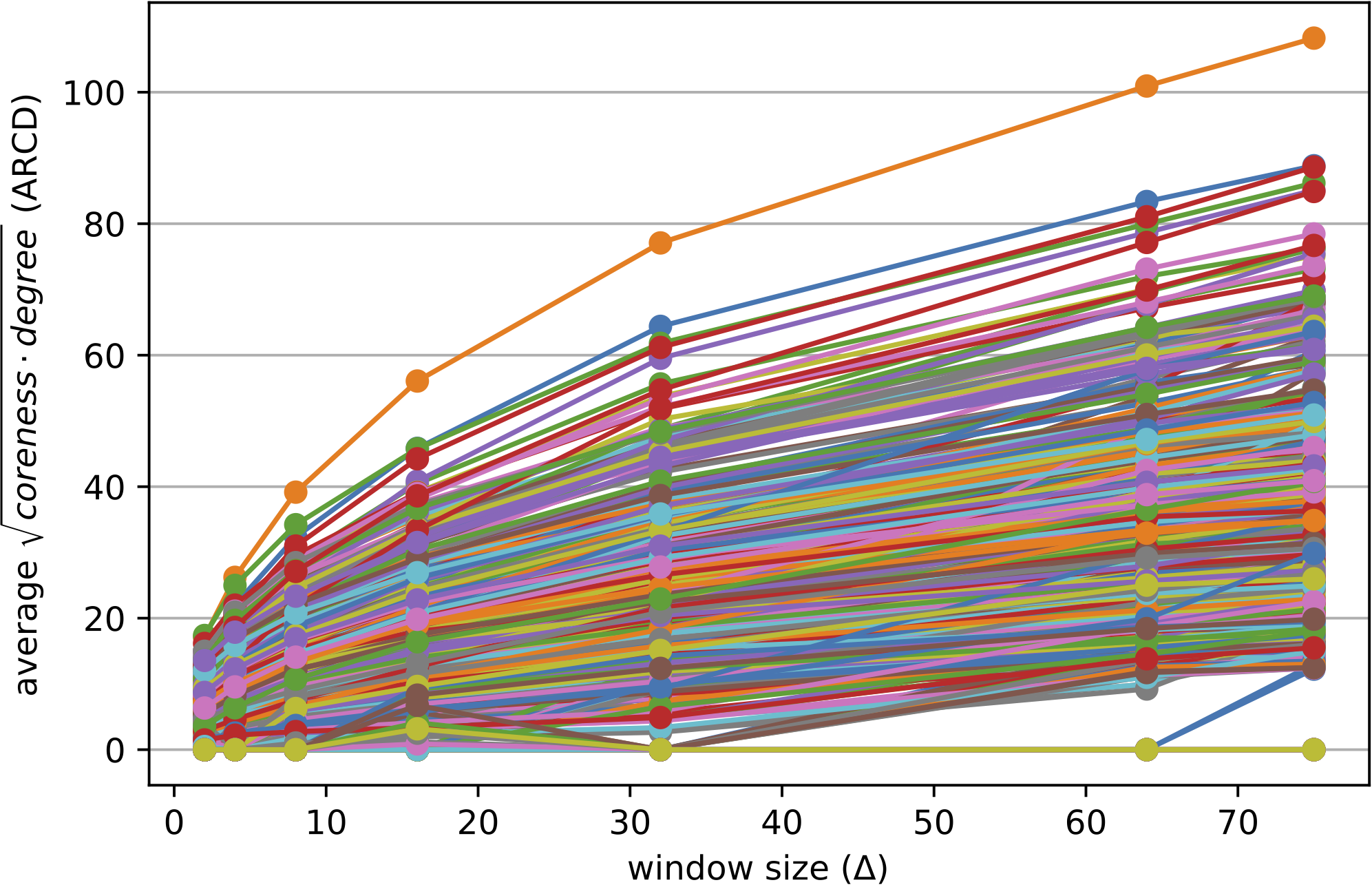}
        \caption{$h = 1$}
    \end{subfigure}
    \caption{Temporal Resilience plots for the Email-EU datasets for intersection and union operations. }
    \label{fig:enter-label}
\end{figure}

Again we observe how this trend could not be seen without careful consideration of temporal information (e.g., from Figure~\ref{fig:mirror_pattern_all}), and we remark the importance of the choice of $h$.

Finally, we remark again how the falling points in the temporal resilience plot can provide a crisp clustering into classes of higher and lower connectivity nodes, marked by color in Figure~\ref{fig:email_classes} similarly to what we did for AS-733 (see Figure~\ref{fig:as_733_classes}).
\begin{figure}[htb]
    \centering
    \includegraphics[width=.85\columnwidth]{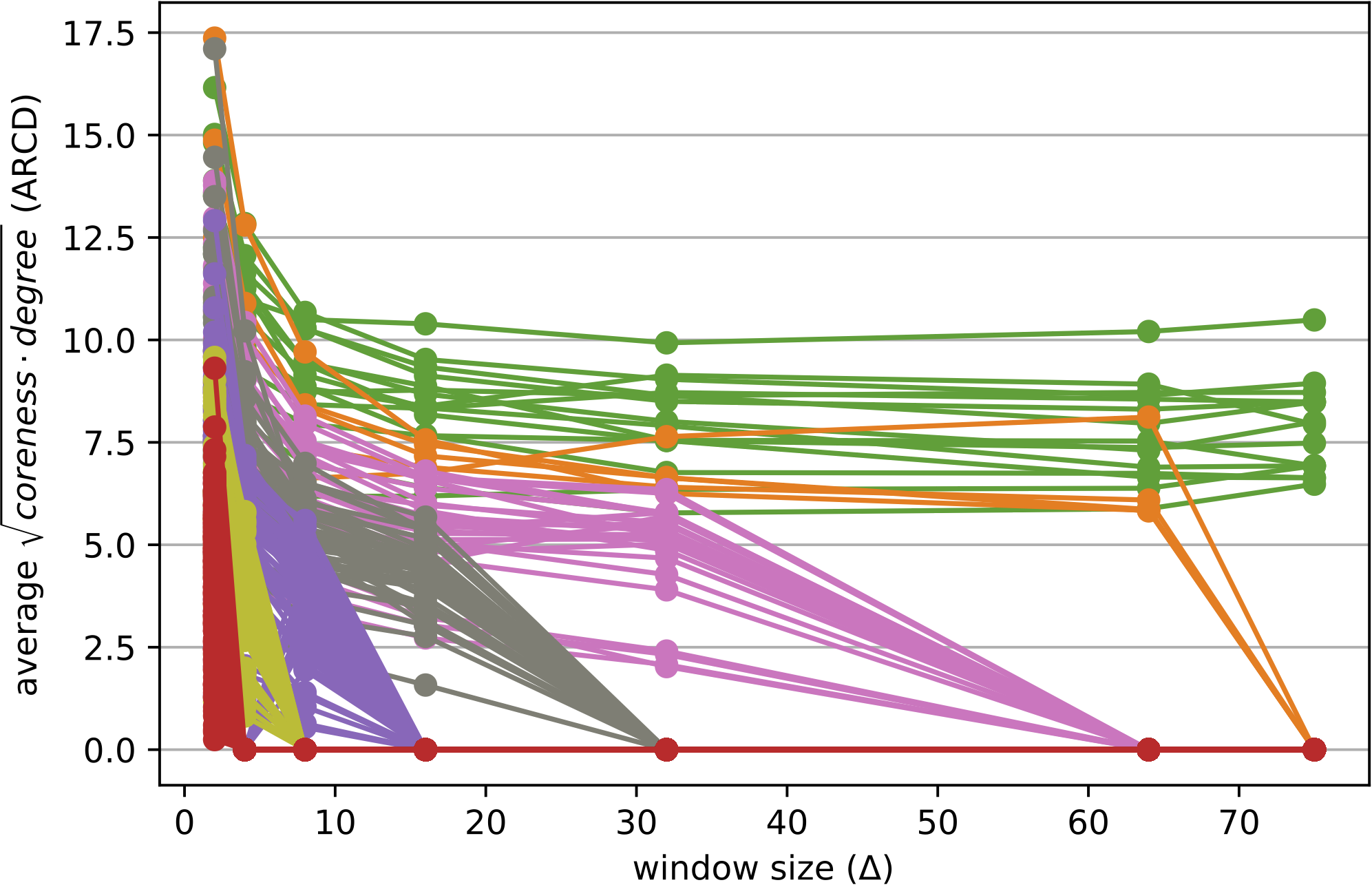}
    \caption{Coloring of the nodes on the Email-EU dataset according to their connectivity class given by the falling point in the corresponding Temporal Resilience plot ($h = \Delta / 2 $).}
    \label{fig:email_classes}
\end{figure}
\paragraph{Reddit Hyperlinks.} 
Nodes in this graph are ``sub-reddits'' and edges are links to one posted on another, collected from January 2014 to April 2017.
Looking at the Temporal Resilience plot for $h=\Delta$ (Figure~\ref{fig:reddit_plots}), the graph resembles the Bitcoin and StackExchange networks for $h=\Delta$, with quickly dwindling ARCD values. On the other hand, for $h=\Delta/2$ we can see similarities with the Email-EU-core graph instead.
We observe that there are medium- and long-lasting communities here, with interactions that are consistent in time; however, the communities are not so tightly knit, i.e., they are not observed if we require frequent interactions, e.g., in every window, which intuitively makes sense when considering that small sub-reddits may not receive comments and links everyday.
Therefore, we can conclude that this graph is somewhat similar in structure to StackExchange networks, but more densely interconnected, and with ``weak but persistent'' communities.

\subsection{Key takeaways}

\begin{figure}[bt]
    \centering
    \begin{subfigure}[t]{.45\textwidth}
        \centering
        \includegraphics[width=\textwidth]{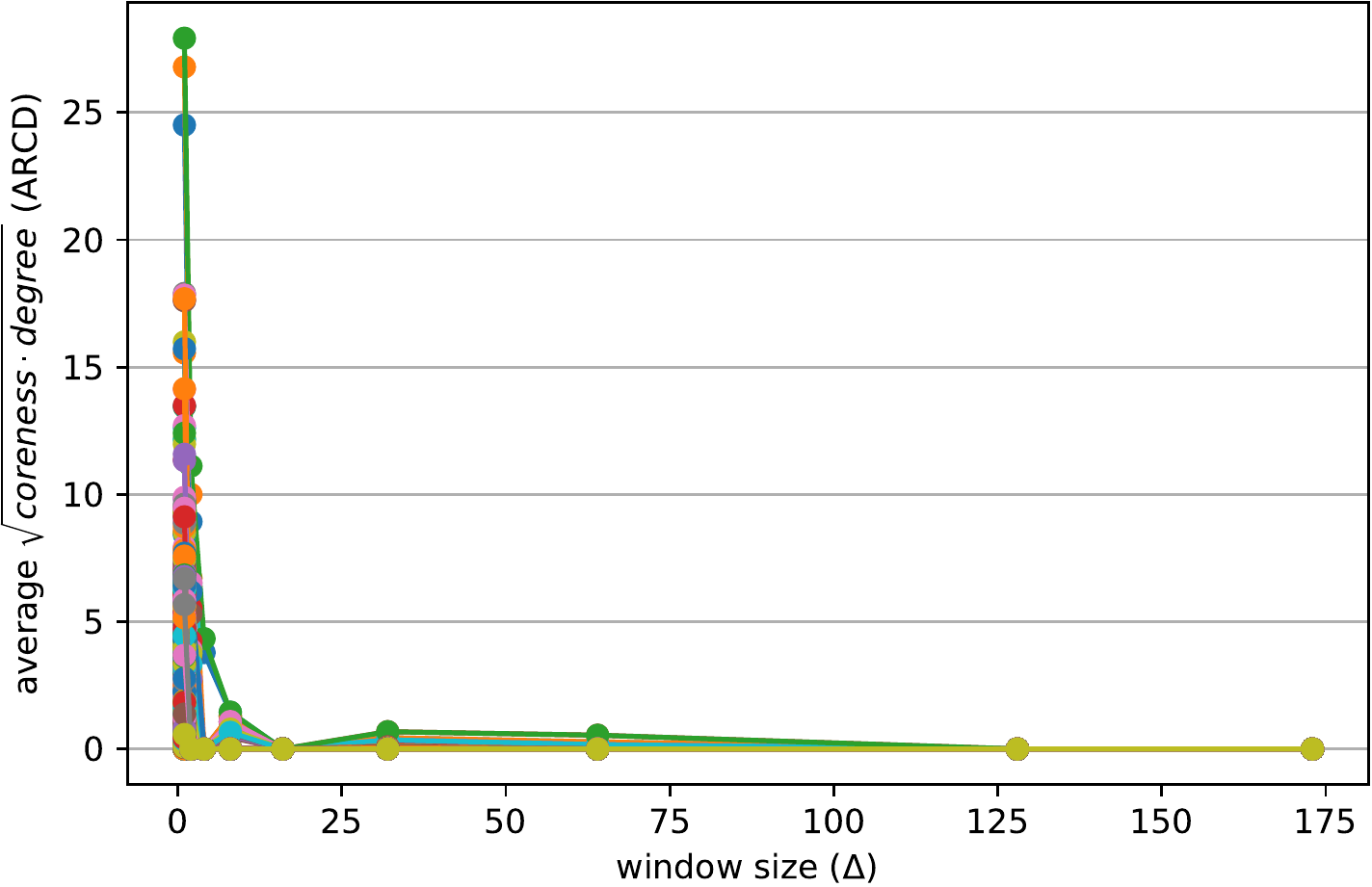}
        \caption{$h = \Delta$}
    \end{subfigure}
    \begin{subfigure}[t]{.45\textwidth}
        \centering
        \includegraphics[width=\textwidth]{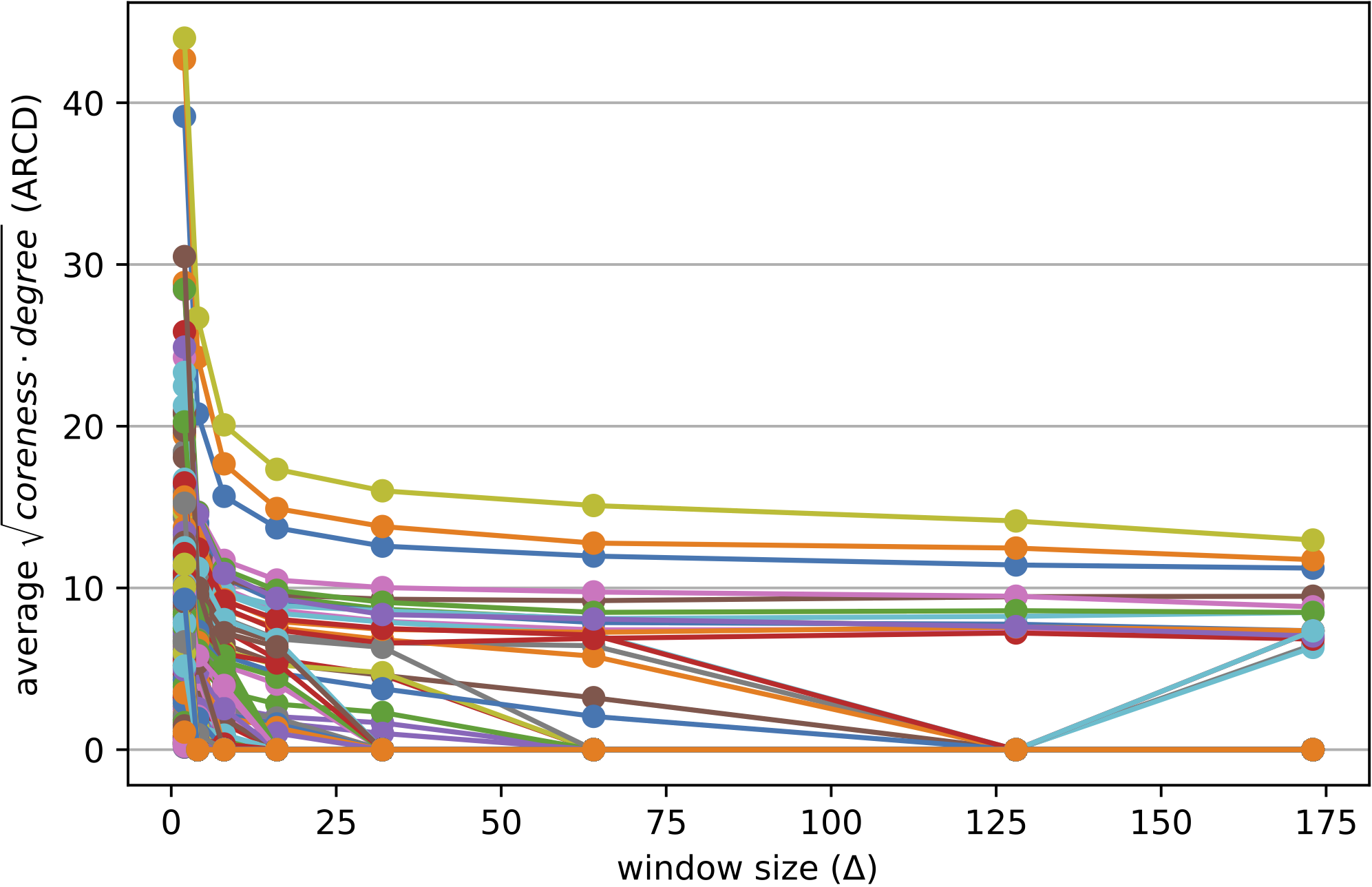}
        \caption{$h = \Delta/2$}
    \end{subfigure}
    \vfill
    \begin{subfigure}[t]{.95\textwidth}
        \centering 
        \includegraphics[width=.65\textwidth]{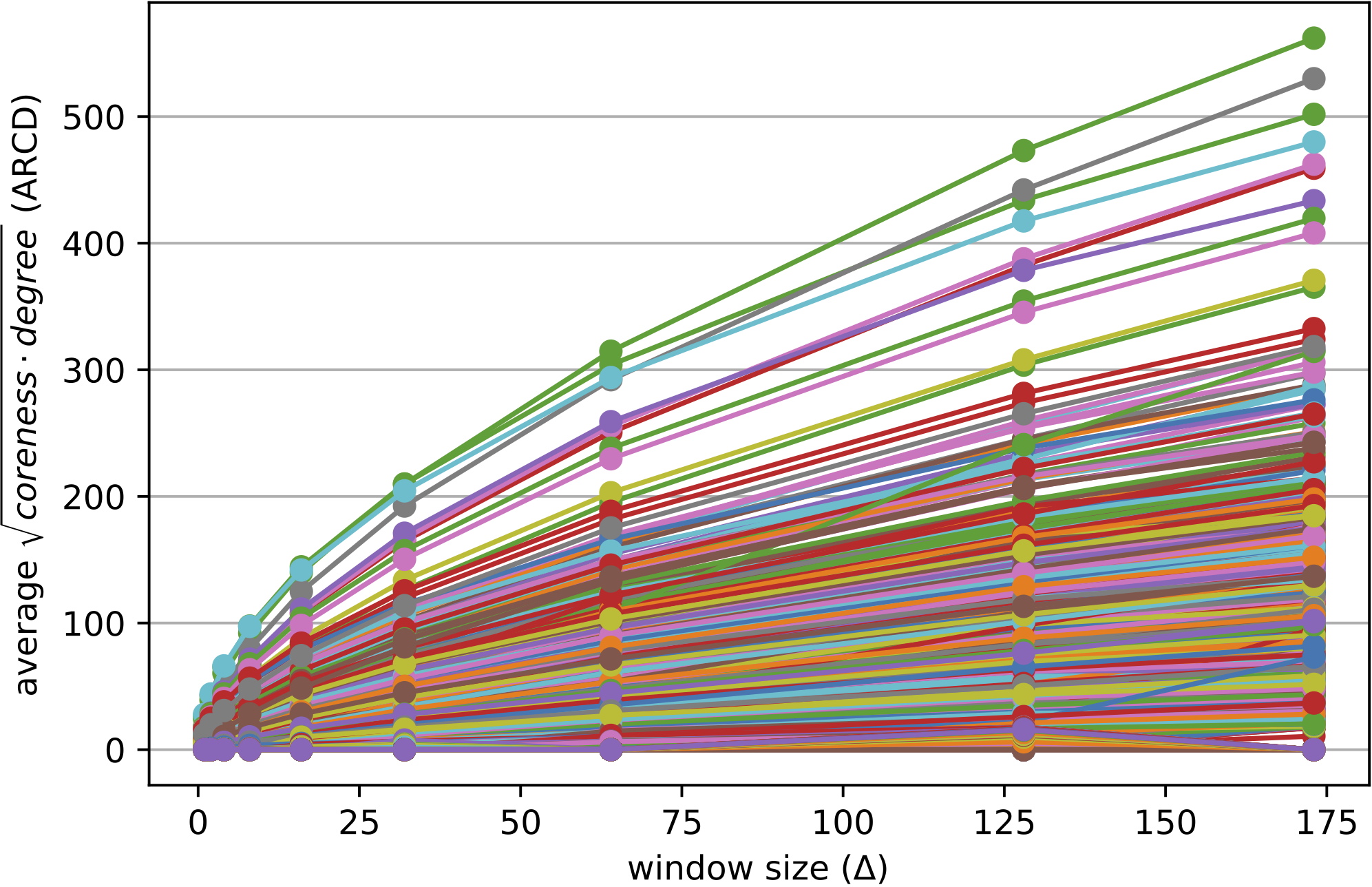}
        \caption{$h = 1$}
    \end{subfigure}
    \caption{Temporal Resilience plots for the RedditHyperlinks dataset.}
    \label{fig:reddit_plots}
\end{figure}%

Our experimental analysis showed that using current definitions of temporal $k$-cores, supported by node degree, is a valid method to uncover trends and dynamics of temporal graphs.
Moreover, this kind of analysis is able to provide a \changed{preliminary} classification of the nodes for further inspection.
Another key aspect is that we are actually able to highlight some aspects of the communities' strenght, resilience and character, as our method was able to visualize some behaviors expected from the domain knowledge of each dataset, such as ephemerality of community structure in Q\&A sites or reliability of backbone routers in AS networks.


For these reasons we believe that temporal $k$-cores are able to highlight information that would not otherwise be captured.

  \chapter{Dense Subgraphs via Dynamic Graph Orientation}\label{chap:denmark}

So far we have seen techniques and results about the extraction and usage of predetermined graph patterns (or motifs) in static and temporal graphs, i.e., we first fixed the class of subgraph to find before proceeding with the community analysis.
Now, we explore a different approach: instead of fixing \emph{apriori} the kind of subgraph to extract, we want to find any subgraph $H$ that maximizes a given measure, in our case the \emph{density} $\rho(G) = |E|/|V|$.
The densest subgraph problem has been heavily studied in the literature, as it obviously serves as the foundation for community detection -- and for this thesis. 
Thus, we point the interested reader to two more surveys on the topic, which can be found in \cite{lanciano2023survey, fang2022survey_densest}.

In this Chapter we exploit dynamic graphs for obtaining an out-orientation that minimizes the maximum outdegree starting from an undirected, simple graph $G$.
This will, in turn, allow us to construct a sequence of decreasingly dense subgraphs that can be later used for our community detection goals, and in particular the first subgraph obtained will indeed the densest subgraph contained in $G$.
All of this comes approximately, as the way in which we obtain the density measure of the graph is a $(1+\epsilon)$ approximation of the exact density.
Of course, in exchange we get very fast time bounds for the update operations that we are going to perform on the dynamic data structures that we need.

The methods and results presented in this Chapter are due to Chekuri et al. \cite{chekuri2023adaptive}, and our contribution lies in, to the best of our knowledge, the first public \textsc{C++} implementation of such strategies, available on GitHub\footnote{\url{https://github.com/DavideR95/densest-subgraph}}.
Similar works in the state of the art for this approach include \cite{16-danimarca, 31-danimarca, 45-danimarca, Berglin2020_8-danimarca, 34-danimarca}.
\section{Background}
The general statement of the problem that we tackle is the following.
\begin{de}[Densest Subgraph Problem]
    Given a simple, undirected, graph $G = (V, E)$, find a subgraph $H$ of $G$ such that $H = $ arg $\max_{H \subseteq G} \rho(H)$.
\end{de}

Throughout this dissertation we viewed this problem from the combinatorial perspective (e.g. by finding pre-defined subgraphs such as $k$-graphlets or $k$-cores), but in the literature there exist several other formulations.
In particular, in this chapter, we focus on the linear programming formalization of the densest subgraph, appeared for the first time in \cite{densest_formalization_LP} and defined as follows:

\begin{de}[Densest Subgraph Linear Program (\texttt{DS}) \cite{densest_formalization_LP}]
    \begin{align*}
        \text{maximize} & \sum\limits_{\{u, v\} \in E}^{} y_{u,v} \qquad\text{s.t} \\
        x_u, x_v & \geq y_{u, v} \qquad \forall u, v \in V, \{u, v\} \in E \\
        \sum\limits_{v \in V}^{} x_v & \leq 1 \\
        x, y & \geq 0
    \end{align*}
\end{de}

Here, $\omega : E \rightarrow \mathbb{R}$ is a function that assigns a \emph{weight} to each edge.
By the duality property of linear programs \cite{matouvsek2007understanding}, we obtain an interesting problem that is equivalent to \texttt{DS}:

\begin{de}[Fractional Orientation Problem (\texttt{FO})]
    \begin{align*}
        \text{minimize}\ \rho & \qquad \text{s.t} \\
        \omega\left((u, v)\right) + \omega\left((v, u)\right) &= 1 \qquad \forall \{u, v\} \in E \\
        \rho \geq d^+(u) &= \sum\limits_{v \in V}^{} \omega((u, v)) \qquad \forall u \in V \\
        \rho, \omega\left((u, v)\right), \omega\left((v, u)\right) \geq 0
    \end{align*}
\end{de}

A \emph{fractional} orientation assigns to every edge $(u, v)$ a weight $\omega((u, v))$ such that $\omega((u, v)) + \omega((v, u)) = 1$.
An \emph{orientation} is a fractional orientation where $\omega((u, v))$ is either 0 or 1. 
In what follows we may assume that an orientation includes the edge $(u, v)$ whenever $\omega((u, v)) = 1$. 
We will not enter the details of these two linear programs, as we only need to notice that the duality between the two problems allows us to get an approximation of the maximum subgraph density, by directing the edges of $G$ in a suitable way that minimizes $\rho$.
Therefore, in what follows, we will discuss a way for maintaining a fractional orientation of $G$.
Obviously, it is not possible to direct an edge ``fractionally'', and we will adopt a graph model in which the edges are duplicated an appropriate amount of times so to mimic the the theoretical fractional orientation.
Note, however, that from a fractional orientation we can obtain an orientation by ``rounding'' every edge by putting $\omega((u, v)) = 1$ whenever $\omega((u, v)) \geq \omega((v, u))$.

\subsection{Outdegree-lowering paths}

The strategy adopted by Chekuri et al. \cite{chekuri2023adaptive} to minimize the maximum outdegree is to exploit this simple observation:
\begin{obs}
    If we take a directed path from a high-outdegree vertex to a low-outdegree vertex and reverse it, we only lower the high outdegree of the first vertex and add one outdegree to the low-outdegree vertex.
\end{obs}

This observation is key: in fact, by applying it to the vertices of maximum outdegree, we effectively decrease the maximum outdegree of the whole graph.
But how long could such paths be? Theoretically, they can be up to $O(|V|)$ long, but this is where approximation comes in.
In particular, the authors of \cite{chekuri2023adaptive} show that maintaining a local, multiplicative, condition of 
\begin{align}
    d^+(u) \leq (1+a)\cdot d^+(v) \label{eq:checkuri}
\end{align}
for any pair of vertices $u, v$ and a chosen value for $a$ allows us to both achieve short paths of length $O(\epsilon^{-2}\log n \log \rho)$ while also maintaining a good approximation of the maximum subgraph density (and in turn, of the densest subgraph).
This local condition is inspired by previous works in the literature such as \cite{34-danimarca, 45-danimarca}.

\section{Dynamic Graphs to maintain orientations}
We will make use of dynamic graphs with their \texttt{Insert} and \texttt{Remove} utility functions to organize and maintain an orientation that meets the requirement of Equation \ref{eq:checkuri}.
In particular, we will insert, one at a time, $b \in \mathbb{N}^+$ copies of each edge of $G$ into a dynamic graph $G^b$ where we maintain an orientation called $\directedG{}^b$ in a way such that the following invariant holds. 

\begin{de}[{\cite[Invariant $\theta$]{chekuri2023adaptive}}]\label{def:invariant_zero}
    For every edge $(u, v)$ in $\overrightarrow{G}^b$, $d^+(u) \leq (1 + \eta \cdot b^{-1}) \cdot d^+(v) + 2\theta$, where $\eta \in \mathbb{R}$.
\end{de}

The authors propose two possible values for $\theta = 0, 1$. 
Here we only show results for $\theta = 0$, and indeed our implementation focuses on this specific value of $\theta$, which will also achieve better approximation for the density value $\rho$.
However, it can happen that a graph may not admit an orientation such that the aforementioned invariant holds, hence the need of a slight relaxation that does not alter the validity of Invariant $\theta$.

\begin{de}[{\cite[Invariant $\theta'$]{chekuri2023adaptive}}]
    For every edge $(u, v)$ in $\overrightarrow{G}^b$,
    
    \begin{align}
        d^+(u) \leq \max \left\{ (1 + \eta \cdot b^{-1}) \cdot d^+(v) + 2\theta , \left\lfloor\frac{b}{2}\right\rfloor \right\}
    \end{align}
    where $\eta \in \mathbb{R}$.
\end{de}
\begin{lemma}[{\cite[Lemma 2.1]{chekuri2023adaptive}}]
    Let $G^b$ be a graph obtained by duplicating each $e \in E$ exactly $b$ times.
    For all $\theta \geq 0$, any orientation $\directedG{}^b$ of $G^b$ that satisfies Invariant $\theta'$ also satisfies Invariant $\theta$.
\end{lemma}

Now we have all the ingredients needed for the main result, a structural theorem that establish the connection between maintaining Invariant $\theta$ in $\directedG{}^b$ and the estimate of the density of $G$, in particular the $(1+\epsilon)$-approximation of $\rho(G)$.

\begin{theorem}[{\cite[Theorem 3.1]{chekuri2023adaptive}}]
    Let $G$ be a graph and let $G^b$ be $G$ with every edge duplicated exactly $b$ times. Let $\rho_b$ be the maximum subgraph density of $G^b$. Let $\overrightarrow{G}^b$ be any orientation of $G^b$ having the following invariant: for some $c \geq 0$, every directed edge $(u, v)$ satisfies $d^+(u) \leq (1 + \eta \cdot b^{-1}) \cdot d^+(v) + c$.

    Then, for any $\gamma > 0$ there exists a value $k_{max} \leq \log_{1+\gamma} n$ for which:
    \[
        (1 + \eta \cdot b^{-1})^{-k_{max}} \Delta(\overrightarrow{G}^b) \leq (1+\gamma) \rho_b + c(\eta ^ {-1}\cdot b + 1).
    \]
\end{theorem}

This Theorem is the main, and the most general, result of \cite{chekuri2023adaptive}. 
The following corollary specifies the consequences of maintaining Invariant $\theta = 0$.

\begin{corollary}[{\cite[Corollary 3.1]{chekuri2023adaptive}}]
    Denote by $\rho$ the density of $G$.
    For any $\eta$ and $b$ such that $\eta b^{-1} = O(1/\log n)$ we have that Invariant $\theta = 0$ for the graph $G^b$ implies $\Delta(\overrightarrow{G}^b) = O(b\rho)$, thus $\Delta(\overrightarrow{G}) = O(\rho)$.
\end{corollary}

We now know that the maximum outdegree of $\directedG{}^b$ is only $b$ times away the value of $\rho(G)$, and the maximum outdegree $\Delta$ of $\directedG{} = O(\rho)$: this means that the (approximated) value of the density of $G$ is implicitly stored in the maximum outdegree of $\directedG{}$.

\subsection{Algorithm Details}
We have established a connection between the density of $G$, our starting graph, and the maximum outdegree of $\directedG{}$, through the process of directing edges in a dynamic graph $\directedG{}^b$.
Now, we show the summarized pseudocode of the algorithm(s) proposed in \cite{chekuri2023adaptive} to maintain a suitable orientation of $\directedG{}^b$ following the rules of Invariant $\theta' = 0$.

\begin{figure}[!ht]
  \makebox[\linewidth]{%
  \begin{minipage}{\dimexpr\linewidth+7em}
\begin{algorithm}[H]
  \DontPrintSemicolon
  \KwIn{A graph $G$, the endpoints $u, v$ of the edge to be inserted $\eta, b$ are fixed and $\lambda = \eta b^{-1} / 64$.}
  \SetKwProg{myproc}{Function}{}{}
  \SetKwFunction{insort}{Insert}
  \SetKwFunction{insertdir}{InsertDirected}
  
    \For{$i = 0$ to $b$}{
      \eIf{$d^+(u) \leq d^+(v)$}{\insertdir{$\directedG{}, u, v$}}{\insertdir{$\directedG{}, v, u$}}
    }

  \myproc{\insertdir{$G, u, v$}}{%
    $d^+(u) \gets d^+(u) + 1$\;
    \lIf{$\nexists (u, v) $ in $\directedG{}^b$}{Add $u$ to $N^-(v)$ and $v$ to $N^+(u)$}
    Add one edge $(u, v)$ to $\directedG{}^b$\;
    $x \gets $ arg $\min \{d^+(w)\,\mid\, w \in N^+(u)\}$\;
    \eIf(\tcp*[f]{$\lambda = \eta b^{-1}/64$}){$d^+(u) > \max \{ (1 + \lambda) \cdot d^+(x), \lfloor\frac{b}{4}\rfloor\}$}{%
      $d^+(u) \gets d^+(u) - 1$\;
      Remove the edge $(u, x)$ from $\directedG{}^b$\;
      \lIf{$\nexists (u, x) $ in $\directedG{}^b$}{Remove $u$ from $N^-(x)$ and $x$ from $N^+(u)$}
      \insertdir{$G, x, u$}
    }{%
      \ForAll{$w \in N^+(u)$}{%
        Update $d^+(u)$ in Buckets$(N^-(w))$\;
      }
    }
  }
  \caption{Adding a new edge to $G$}
  \label{alg:denmark}
\end{algorithm}
  \end{minipage}%
  }
\end{figure}

In practice, we will insert the edges of the (static) graph $G$ one at a time using Algorithm \ref{alg:denmark}, which will take care of duplicating it $b$ times and keeping the internal orientation of $\directedG{}^b$ organized.
When the edge list of $G$ has been consumed, the maximum outdegree of $\directedG{}^b$, divided by $b$, will be our approximation of $\rho(G)$.

\paragraph{Data Structures.} The data structures used in Algorithm \ref{alg:denmark} can vary, and Chekuri et al.\cite{chekuri2023adaptive} show indeed that a suitable choice of the data structure allow for even more improved time complexity bounds. 
For the scope of this thesis and our implementation, we only show the basic data structures needed. 
In particular, for each vertex $u \in V$ we need to store:
\begin{enumerate}[(a)]\itemsep0em
    \item the value $d^+(u)$ in the current orientation $\directedG{}^b$,
    \item the set $N^+(u)$ of outgoing neighbors in any arbitrary order,
    \item the set $N^-(u)$ of incoming neighbors in a \emph{sorted} doubly-linked list of \emph{buckets} $B_j(u)$:
    \begin{itemize}[-]
        \item $B_j(u)$ contains a doubly linked list with all $w \in N^-(u)$, with $j = \lfloor \log_{(1+\lambda)} d^+(w) \rfloor$;
        \item vertex $u$ has a pointer to the bucket $B_i(u)$ with $i = \left\lfloor \log_{(1+\lambda)} \max \{(1+\lambda)d^+(u), \left\lfloor\frac{b}{4}\right\rfloor \} \right\rfloor$.
    \end{itemize}
\end{enumerate}

\paragraph{Parameters.} It remains to choose the value for the parameters such as $\eta$ and $b$ in order to maintain Invariant $\theta' = 0$.
To do this, we exploit the following result:
\begin{theorem}[{\cite[Theorem 4.2]{chekuri2023adaptive}}]\label{th:denmark_impo}
  Let $G$ be a dynamic graph and $\rho$ be the density of $G$ at time $t$.
  We can choose $\theta = 0, \eta = 3, b = \Theta(\log n), b \geq 2$ to maintain an out-orientation $\directedG{}^b$ in $O(\rho\cdot\log^2 n \log \rho)$ time per update in $G$, maintaining Invariant $\theta = 0$ (ref. Definition \ref{def:invariant_zero}) for $\directedG{}^b$ with
  \[
    \forall u \in V, \text{the outdegree of } u \text{ in } \directedG{} \text{ is } O(\rho).
  \]
\end{theorem}

\section{Obtaining the densest subgraph}
Finally, we show how to obtain a $(1+\epsilon)$-approximation of the maximum densest subgraph after using Algorithm \ref{alg:denmark} for constructing $\directedG{}^b$.

\begin{corollary}[{\cite[Corollary 7.1]{chekuri2023adaptive}}]
    Let $G$ be a dynamic graph subject to edge insertions (and deletions) with maximum subgraph density $\rho$. Let $G^b$ be $G$ where every edge is duplicated exactly $b$ times. Let $0 \leq \epsilon \leq 1$. We can maintain an orientation $\directedG{}^b$ such that
    \[
      \rho \leq b^{-1} \cdot \Delta(\directedG{}^b) \leq (1+\epsilon)\rho,
    \]
    with update time $O(\epsilon^{-6}\log^3 n \log \rho)$ per operation in $G$.
\end{corollary}

The proof of this corollary uses a refined version of Theorem \ref{th:denmark_impo}, but most importantly it provides us with the parameters to be used in the algorithm. 
Indeed a suitable choice of parameters for retrieving the densest subgraph would be $\gamma = \epsilon / 2, \eta = 3, b = \left\lceil \gamma^{-1} \eta \log_{1+\gamma} n \right\rceil$.

\begin{lemma}[{\cite[Lemma 7.1]{chekuri2023adaptive}}]\label{lemma:denmark_final}
    For a dynamic graph $G$, there is an algorithm that explicitly maintains a $(1+\epsilon)$ approximation of the maximum subgraph density in $O(\epsilon^{-6} \log^3 n \log \alpha)$ total time\footnote{Here, $\alpha$ denotes the \emph{arboricty} of the graph, which is defined as $\alpha = \max_{S \subseteq V, |S| \geq 2} \left\lceil \frac{|E \cap (S \times S)|}{|S|-1} \right\rceil $. It represents the minimum number of forests the edges of a graph can be partitioned into.} per operation, and that can output a subgraph realizing this density in $O(occ)$ time, where \texttt{occ} is the size of the output.
\end{lemma}

\begin{proof}[Proof by \cite{chekuri2023adaptive}]
    Define, for $i \in \mathbb{N}$ the family of sets
    \[
      T_i := \left\{  v \in V \,\mid\, d^+(v) \geq \Delta(\directedG{}^b) \cdot (1 + \eta \cdot b^{-1})^{-i}     \right\}
    \]
    Let $k$ be the smallest integer such that $|T_{k+1}| < (1+\gamma)|T_k|$. 
    $T_{k+1}$ is an approximation of the densest subgraph of $\directedG{}^b$, and therefore of $G$.

    Now, store the vertices of $\directedG{}^b$ as leaves of a balanced binary tree, sorted on their outdegree.
    Since every change in $G$ changes at most $O(b \log n \log \rho) = O(\epsilon^{-2} \log^2 n \log\rho)$ outdegrees in $\directedG{}$, this tree can be maintained in $O(\epsilon^{-2}\log^3 n\log \rho)$ additional time per operation in $G$.

    Each internal node of the balanced binary tree stores the size of the subtree rooted at that node, and we store the maximum outdegree $\Delta(\directedG{}^b)$ as a separate integer.
    Additionally, the leaves are connected among each other via a doubly linked list.

    After each operation in $G$, for each integer $i$ in the interval $[0, \epsilon^{-1}\log n]$, we can calculate $|T_i|$ by calculating first the value $V_i = \Delta(\directedG{}^b)\cdot (1+\eta\cdot b^{-1})^{-i}$, and then in $O(\log n)$ time we can identify how many vertices have outdegree at least $V_i$. This requires an additional $O(\epsilon^{-1}\log^2 n)$ time.

    Finally, by storing a pointer to the first leaf of $T_k$ we can scan the doubly linked list among the leaves containing \texttt{occ} elements to output them in $O(occ)$ time.
\end{proof}

This last proof provides us with the algorithm for finding the densest subgraph in $G$, i.e., the complete set of vertices that compose the subgraph with the approximated density that we got from Algorithm~\ref{alg:denmark}.
In particular, this process needs to be plugged into Algorithm~\ref{alg:denmark}, so that whenever a new edge is inserted, the set of vertices of the densest subgraph is also updated.
Moreover, notice that this process can be repeated to obtain a sequence of decreasingly dense subgraphs, thanks to the $V_i$ threshold and the sets $T_i$ that are built to group together such kind of subgraphs.

\paragraph{Recap.} This is a completely different approach with respect to the others that we took in previous chapters.
We did not ask for a specific kind of subset, but rather for a sufficiently dense one, in particular the densest subgraph contained in the input graph $G$. 
By seeing the densest subgraph problem under the light of a linear program we were able to translate it into an orientation finding problem.
Then, in order to find such an orientation, we treated the static graph in input as a dynamic graph subject to edge insertions, one after the other, in which we flip the orientation of chains of edges of $O(\log n)$ length. 
Eventually, the maximum outdegree in $\directedG{}^b$ gives us an approximation of the density of the graph, and we are able to extract the densest subgraph by grouping vertices in a clever way according to this threshold using Lemma \ref{lemma:denmark_final}.

\section{Implementation}
Our contribution, as noted before, lies in the first public implementation\footnote{Available on \href{https://github.com/DavideR95/densest-subgraph}{GitHub}.} of Algorithm~\ref{alg:denmark}.
Thus, this section discusses the details of our implementation and shows the results of our preliminary experimental phase.
We note that this is a preliminary version aimed at showing the potential that this strategy offers, that can be used in the future as an analysis tool for large, real-world, graphs. 
Though Algorithm~\ref{alg:denmark} is developed with dynamic graphs in mind, we can provide a static graph as input by simulating the arrival of each edge one at a time, and this is what we do in our implementation.

\subsection{Details and Dataset}
We implemented the strategy shown in Algorithm~\ref{alg:denmark} in C++, compiled with the \texttt{-O3} optimization flag of the \texttt{g++} compiler. 
Our experiments were carried on a dual-processor Intel Xeon Gold 5318Y Icelake @ 2.10GHz machine, with 48 physical cores each and 1TB of shared RAM, running Ubuntu Server 22.04 LTS.
We use a vector for storing the values of $d^+(u)$ for all $u \in V$, so to be able to access them in $O(1)$ time, using $O(|V|)$ space. 
The set $N^+(u)$ is encoded, for each $u \in V$, as a vector of pairs $\langle v, r\rangle$ where $1 \geq r \geq b$ represents the multiplicity of the directed edge $(u, v)$ in $\directedG^b$.
Finally, the set $N^-$ is implemented as a vector of buckets $B_u$, where $u \in V$. A bucket consists of a doubly linked list of incoming neighbors, together with an integer that represents the outdegree of $v$, for each edge $(v, u)$ in $\directedG^b$.

Our dataset is composed by the same undirected graphs used in Chapter~\ref{chap:cage}, coming from the SNAP repository \cite{snapnets}, we show a significant subset of in Table~\ref{denmark:tab:dataset}, where we highlight the values for $|V|$ and $|E|$ for each graph. 

\begin{table}[htb]
\caption{A sample from our dataset, sorted by $|E|$.}
\label{denmark:tab:dataset}
\footnotesize
\centering
\small
\begin{tabular}{|l|l|r|r|} \hline
\multicolumn{1}{|c|}{Graph} & \multicolumn{1}{c|}{Type} & \multicolumn{1}{c|}{$|V|$} & \multicolumn{1}{c|}{$|E|$} \\ \hline
Brady   & Biological    & 1,117 & 1,330           \\ 
ca-GrQc & Collaboration & 5,242 & 14,484 \\ 
Wing   & Mesh & 62,032 & 121,544 \\ 
Roadnet-CA & Road Network & 1,965,206 & 2,766,607  \\ 
Auto & Mesh & 448,695 & 3,314,611 \\ 
\hline
\end{tabular}%
\end{table} 

\subsection{Results and Discussion}
The results of our experimental analysis are summarized in Table~\ref{denmark:tab:experiments}, where we show our choice of parameters, the $(1+\epsilon)$ approximation of $\rho$ (which we call $\rho'$) and the real $\rho$ value for the same subset of vertices. 
Since we are able to compute the true density value, we also show a measure of the relative error of the $\rho'$ value.
Finally, we put a timeout of $30$ minutes on the running time of the algorithm, thus Table~\ref{denmark:tab:experiments} shows the largest value of $b$, with respect to $\left\lceil \gamma^{-1}\eta \log_{1+\gamma}n\right\rceil$, that allowed the algorithm to run in less than $30$ minutes on our machine.
When this happens, the corresponding value of $b$ is marked with an asterisk (\texttt{*}) in the table.

\begin{table}[htb]
    \caption{Experimental results for a subset of our dataset. Fixed parameters are $\epsilon = 0.5$, $\eta = 3$. \texttt{*}: highest value of $b$ for which the execution required less than $30$ minutes (the theoretical value for $b$ is in parenthesis).}
    \label{denmark:tab:experiments}
    \centering
    \resizebox{\textwidth}{!}{%
    \begin{tabular}{|l|r|r|r|r|r|r|}
        \hline
         Dataset & $b$ & $\rho'$ (approx) & $\rho$ & $|\rho' - \rho|$ & Time (s) \\ \hline
         Brady & 378 & 2.27986 & 2.25 & 0.02986 & 1.53 \\ \hline
         ca-GrQc & 461 & 22.96 & 22.3913 & 0.5687 & 83.36 \\ \hline
         Wing & 594 & 2.05219 & 1.94426 & 0.10793 &  917.64 \\ \hline
         Auto & $50^*$ (700) & 9.00 & 7.38572 & 1.61428 & 1691.28 \\ \hline 
         Roadnet-CA & $250^*$ (780) & 2.1 & 1.53798 & 0.56202 & 1482.08 \\ 
        \hline
    \end{tabular}%
    }
\end{table}

Our results show very good approximation errors whenever we are able to use the perfect value for $b$.
In contrast, and as expected, the error increases as we stray away from that value for running time concerns; despite this, the results obtained are not so far away from the corresponding true value of $\rho$, thus we believe that these preliminary results are encouraging.
The bottleneck of our implementation resides in the value of $b$, which is in practice high with respect to the total number of edges present in our test graphs, and this leads to numerous recursive adjustments to the edge orientation each time that we insert a new copy of an edge, slowing the entire process down.
Nonetheless we believe that our results provide a solid foundation for iterative refinements at the implementation level; eventually this could also lead to a theoretical improvement on how the value of $b$ is calculated.
Of course, increasing the value of $\epsilon$ leads to improvements both to the value of $b$ and to the running time, at the cost of a less accurate estimatation of $\rho$.

We leave a more fine tuned version of the implementation as future work. In particular, the next version should use more sophisticated data structures to speed up the computation (hints on how to do this are already provided in \cite{chekuri2023adaptive}), as well as to output the actual set of vertices that compose the densest subgraph. 

  \chapter{Conclusions}\label{chap:conclusions}

In this thesis we tackled the problem of community detection in large, real-world, networks formalized as graphs, from the combinatorial point of view. 

We first designed a novel, amortized algorithm for $k$-graphlet enumeration, taking the state of the art a step further, from $O(k^2\Delta)$ time to $O(k^2)$, dropping the $\Delta$ term that in real-world graphs can be very high.
This required careful analysis of the recursion tree spawned by a binary partition algorithm, and a deep knowledge of amortized analysis in order to achieve such a tight bound.
Moreover, under certain circumstances such as bounded-degree graphs with constant value for $k$, we are able to achieve $O(1)$ amortized time, a first in theoretical $k$-graphlet enumeration literature.

The analysis of the aforementioned recursion trees associated to binary partition algorithms, also gave us a new idea for a very practical path to follow. 
In fact, as the number of $k$-graphlets in any graph grows exponentially with $k$, we cannot rely solely on good asymptotic bounds, we also need to exploit the CPU resources that we have available.
This is why we presented CAGE, a cache-aware algorithm for $k$-graphlet enumeration that builds upon state-of-the-art binary partition algorithms, as well as our empirical observation that binary partition produces only few failure leaves in the computational recursion tree. 
We developed and analyzed several variants of CAGE, iteratively cutting more recursion levels, that were fine-tuned during an extensive experimental phase based on cache analysis and number of graphlets reported.
Eventually we were able to prune the recursion tree up to \emph{three levels}, i.e. when the graphlet being constructed reaches size equal to $k-3$, saving an \emph{exponential} amount of work.
We believe CAGE can be used as the core for more sophisticated enumeration tasks, such as orbit enumeration or isomorphic graphlet classification 
as well as arbitrary user-defined queries, boosting analysis possibilities for large networks.

Next, we turned our attention to temporal graphs, and faced the analysis-prone question of ``what is an appropriate graph substructure to be used for community analysis?'' In this thesis, we answered this question by showing how to extend classical graph community analysis based on degree and coreness to temporal graphs, including a data structure to be used for querying intervals of the graphs, employing several proposed definitions of temporal $k$-core combined with node degree.
To visualize information effectively we adopted \emph{temporal resilience plots}, showing condensed information on node connectivity for different temporal resolutions at once. 
%
We discussed in depth the results obtained by our experiments, highlighting key insights regarding the dynamics of the networks such as existence, strength and volatility (or persistence) of communities, as well as resilience of node connectivity. We observed how these insights cannot be obtained by static graph analysis, motivating the effectiveness of \textit{temporal} $k$-cores.
%
%

Finally, we presented a very recent work on finding the densest subgraph through dynamic graph orientation. 
We provided a first implementation of this technique, which is very different from the enumerative ones that we discussed in the other chapters. 
Indeed, by simulating a dynamic environment when adding edges to the input graph, we are able to keep an up-to-date, $(1+\epsilon)$ approximation of the vertices that are part of the densest subgraph.

We believe that this thesis takes a step further in nowadays subgraph listing, and provides useful tools for performing community analysis in large, real-world, graphs.

\section{Future Work}
The algorithms and analysis techniques presented in this thesis can be expanded in many ways.

First, we would like to achieve theoretical $O(k)$ amortized time for $k$-graphlet enumeration, 
in particular we may exploit the push-out amortization analysis technique (Theorem~\ref{th:pushout}) to obtain such an algorithm.
However, it will require a lot of fine tuning because push-out amortization poses very strict requirements on the time complexity of each recursive call.

Future work for the CAGE algorithm will consider applying its analysis capabilities to graph analysis tasks, as well as attempting to increase the levels of recursion-tree pruning beyond 3, which is in principle a difficult task because of the combinatorial explosion of the base cases to consider, and we need a good trade off between the pruning of the tree and the computational complexity of leaf iterations.
Additionally, we will consider adding to CAGE more features, such as orbit enumeration and isomorphic graphlet classification.

Moreover, we can push forward the experimental analysis on temporal graphs by taking the graphs' metadata into consideration, whenever they are available. 
Exploiting the metadata in combination with our proposed plots, we will be able to understand more in depth the reasons behind the behavior of any given temporal network; plus, temporal data is becoming more and more available these days, so we hope our work will provide one of the foundations for temporal community analysis.

As this field has received a lot of attention from the research community for many years, we believe that this topic is an evergreen for research on both enumeration and analysis of subgraphs, so it still offers plenty to explore.
\cleardoublepage

\appendix
\cleardoublepage
\manualmark
\markboth{\spacedlowsmallcaps{\bibname}}{\spacedlowsmallcaps{\bibname}} 
\refstepcounter{dummy}
\addtocontents{toc}{\protect\vspace{\beforebibskip}} 
\addcontentsline{toc}{chapter}{\tocEntry{\bibname}}
\label{app:bibliography}
\AtNextBibliography{\renewbibmacro*{pageref}{}}
\printbibliography
\end{document}